\documentclass[a4paper,11pt]{article}

\usepackage[english]{babel}
\usepackage[left=0.885in, right=0.885in, bottom=1in, top=1in]{geometry}\usepackage{amsmath}
\usepackage{amsxtra}
\usepackage{amsthm}
\usepackage{amssymb}
\usepackage{color}
\usepackage{dsfont}

\usepackage{xcolor}

\usepackage{hyperref}         

\def\bZ{\mathbb{Z}}

\def\cU{\mathcal{U}}

\def\cV{\mathcal{V}}
\def\cO{\mathcal{O}}
\def\cF{\mathcal{F}}
\def\cG{\mathcal{G}}
\def\cL{\mathcal{L}}

\def\cN{\mathcal{N}}
\def\cE{\mathcal{E}}
\def\cK{\mathcal{K}}

\def\eps{\varepsilon}
\def\ph{\varphi}

\def\RRR{\mathbb{R}}
\def\R{\mathbb{R}}

\def\indic{\hbox{\raise-2pt \hbox{\indbf 1}}}

\let\==\equiv

\let\0=\noindent

\def\*{{\hfill\break\null\hfill\break}}

\def\norm#1{{\left|\hskip-.05em\left|#1\right|\hskip-.05em\right|}}

\def\tende#1{\,\vtop{\ialign{##\crcr\rightarrowfill\crcr
             \noalign{\kern-1pt\nointerlineskip}
             \hskip3.pt${\scriptstyle #1}$\hskip3.pt\crcr}}\,}
\def\otto{\,{\kern-1.truept\leftarrow\kern-5.truept\to\kern-1.truept}\,}

\def\tr{{\rm tr}}

\newtheorem{theorem}{Theorem}[section]  
\newtheorem{cor}[theorem]{Corollary}
\newtheorem{prop}[theorem]{Proposition}
\newtheorem{lemma}[theorem]{Lemma}

\numberwithin{equation}{section}

\def\be{\begin{equation}}
\def\ee{\end{equation}}

\newcommand{\hc}{\mbox{h.c.}}

     \let\g=\gamma     \let\d=\delta     
        \let\k=\kappa     \let\l=\lambda
    \let\n=\nu                      
\let\s=\sigma          \let\ph=\varphi   
        
 \let\D=\Delta       \let\L=\Lambda    
             
\let\O=\Omega

\begin{document}

\title{The Bose gas in a box with\\ Neumann boundary conditions} 

\author{Chiara Boccato and Robert Seiringer}

\maketitle
\abstract{We consider a gas of bosonic particles confined in a box with Neumann boundary conditions. We prove Bose-Einstein condensation in the Gross-Pitaevskii regime, with an optimal bound on the condensate depletion.  Our lower bound for the ground state energy in the box implies (via Neumann bracketing) a lower bound for the ground state energy of the Bose gas in the thermodynamic limit.

}
\tableofcontents

\section{Introduction and main results}

We describe a system of $N$ interacting bosonic particles in a box $\L=[-L/2, L/2]^3$ through the Hamiltonian
\begin{equation}\label{eq:HL} 
  H_{N}=-\sum_{i=1}^N\Delta_i+\kappa\sum_{i<j}^N V(x_i-x_j)
\end{equation}
acting on $L_s^2(\L^N)$. This is the space of symmetrized $L^2$ functions, defined as
\[
 L^2_s(\L^{N})=\{\psi\in L^2(\L^{N}): \psi(x_{\sigma(1)},\dots,  x_{\sigma(N)})=\psi(x_1,\dots,x_N) \text{ for every }\sigma\in S_N \},
\]
where $S_N$ is the set of all permutations of $N$ objects. In \eqref{eq:HL}, $\D_i$ indicates the Laplacian with Neumann boundary conditions acting on particle $i$.
The interaction potential $V$ is a multiplication operator and we will assume it to be a nonnegative, spherically symmetric, compactly supported and bounded ($\kappa$ is a coupling constant).
We denote $\mathfrak{a}$ the scattering length of the interaction potential $\kappa V$. The scattering length is defined through the zero-energy scattering equation 
\begin{equation}\label{eq:0en}
\left[ - \Delta + \frac{\kappa}{2}  V (x)  \right] f_0 (x) = 0
\end{equation}
with the boundary condition that $f_0 (x) \to 1$, as $|x| \to \infty$. Here $\Delta$ denotes the Laplacian on $\R^3$. For $|x|$ large enough (outside the support of $V$), we have 
\begin{equation}\label{eq:0enSol} 
f_0(x) = 1- \frac{\frak{a} }{|x|} 
\end{equation}
where $\mathfrak{a}$ is the scattering length of $\kappa V$. Equivalently, $\frak{a}$ can be obtained as 
\begin{equation}\label{eq:a0}
8 \pi \frak{a}  =  \int_{\mathbb{R}^3} \k V(x) f_0(x) dx. 
\end{equation}

We are interested in static properties of the Bose gas.
The ground state energy per particle in the thermodynamic limit, i.e.,  the limit $N,|\L|\to\infty$ with $\rho=N/|\L|$ fixed, is given by
\begin{equation}\label{eq:energyPP}
 e (\rho)=\lim_{N,L\to\infty} \frac{E(N,L)}{N}
\end{equation}
where  $E(N,L)$ is the ground state energy of $H_{N}$, defined as \[E(N,L)= \inf_{\psi \in L^2_s(\Lambda^{N}),\, \norm{\psi} = 1} \langle \psi, H_N \psi\rangle.\]
For dilute  gases, i.e., when the density $\rho$ is small, the ground state energy per particle in the thermodynamic limit is described by the Lee-Huang-Yang formula \cite{LeeY,LeeHY}
\begin{equation}\label{eq:LHY} 
 e (\rho)=\lim_{\substack{N,L\to\infty\\\rho=N/|\L|}} \frac{E(N,L)}{N}= 4\pi \rho \frak{a} \left[ 1 + \frac{128}{15 \sqrt{\pi}} (\rho \frak{a}^3)^{1/2} + o ((\rho \frak{a}^3)^{1/2} ) \right],
\end{equation}
proved in \cite{YY,FSol,BasCS}. One of the main characteristics of \eqref{eq:LHY} is the universality of the first two orders, where only the scattering length appears and other details of the interaction potential do not matter.

 To compute thermodynamic quantities such as the ground state energy $e (\rho)$, a standard method (see for example  \cite{LY}) consists in dividing the box $\L$ into $M^3$ cells of size $\ell=L/M$ and reducing the problem to the study of the localized system to each cell. The choice of the boundary conditions on the cells is therefore very important, and while Dirichlet boundary condition are suited to compute upper bounds, lower bounds require for example Neumann boundary conditions. In particular, to compute a lower bound for $e (\rho)$, we distribute the $N$ particles in the $M^3$ cells (so that the $k$-th box has $n_k$ particles) and neglect interactions between particles in different cells, exploiting the positivity of the interaction potential. The lower bound is then obtained by adding the lower bounds in the different cells and minimize over all the possible ways of distributing the particles in the cells,  i.e., 
\begin{equation}\label{eq:lowerBound}
  E(N,L)\geq \inf_{\{n_k\}:\sum_kn_k=N}\sum_{k=1}^{M^3}E(n_k,\ell).
\end{equation}
Here $E(n,\ell)$ is the ground state energy of $H_{n}$ (defined in \eqref{eq:HL}, with $N$ substituted by $n$), acting on $L^2_s(\L_\ell^{3n})$, where $\L_\ell=[-\ell/2,\ell/2]^3$, with Neumann boundary conditions.
Rescaling lengths, the Hamiltonian \eqref{eq:HL} takes the form
\begin{equation}\label{eq:H1}
   H_{n,\ell}=-\sum_{i=1}^n\Delta_i+\kappa\sum_{i<j}^n\ell^2 V\big(\ell(x_i-x_j)\big)
\end{equation}
and acts on $L^2_s(\L_1^{3n})$, with $\L_1=[-1/2,1/2]^3$.
Up to a factor $\ell^2$, $H_{n}$ and $H_{n,\ell}$ are unitarily equivalent, i.e. there exists a unitary\footnote{The unitary transformation $\cU$ acts as
\begin{equation}\nonumber
\begin{split}
 \cU:&\hspace{0.8cm}\,L^2(\L_\ell^n)\longrightarrow L^2(\L_1^n)\\
 &\ph(x_1, \dots, x_n)\to (\cU\ph)(x_1, \dots, x_n)=\ell^{3n/2}\ph(\ell x_1, \dots,\ell x_n)
\end{split}
\end{equation}} $\cU$ such that $ \cU^*H_{n,\ell}\,\cU=\ell^{2}\,H_{n}$. Denoting with $e_{n,\ell}$ the ground state energy of \eqref{eq:H1}, we have clearly  $e_{n,\ell}=\ell^2E(n,\ell)$.

In the case $n=\ell=N$, \eqref{eq:H1} describes the well known Gross-Pitaevskii regime. Here the density is proportional to $N^{-2}$, and hence the energy per particle is of the same order as the spectral gap of the Laplacian. In particular, in this regime the large volume and large particles number limit is simultaneous to the low density limit. 
The Gross-Pitaevskii regime has been studied for the translation invariant Bose gas, where periodic boundary conditions are imposed on $\L_1$, and for the trapped Bose gas, where particles move in $\R^3$ and are confined by an external potential. In these cases, Bose-Einstein condensation has been proved \cite{LS,LS2,NRS} with optimal rate \cite{BBCS1,BBCS4,NNRT,BSS1,H}. In the translation invariant case, the ground state energy has been shown in \cite{BBCS3} to be
\begin{equation}
\begin{split}\label{1.groundstateGP}
e_{N} = \; &4\pi (N-1) \frak{a} + b_\Lambda \frak{a}^2  - \frac{1}{2}\sum_{p\in\Lambda^*_1} \left[ p^2+8\pi \frak{a}  - \sqrt{|p|^4 + 16 \pi \frak{a}  p^2} - \frac{(8\pi \frak{a} )^2}{2p^2}\right] + \cO (N^{-1/4}),
    \end{split}
    \end{equation}
where 
$ b_\Lambda = 2 - \lim_{M \to \infty} \sum_{p \in \bZ^3 \backslash \{ 0 \}, \,  |p| \leq M} \frac{\cos (|p|)}{p^2}$ is a boundary contribution. In addition in \cite{BBCS3} the excitation spectrum has been derived (these results have been also revisited in \cite{HST}). The result has then later been generalized to the trapped Bose gas \cite{NT,BSS2}.

\medskip

 In this paper we consider the Bose gas with Neumann boundary conditions in the Gross-Pitaevskii regime. In  Theorem \ref{mainTheorem} below we prove Bose-Einstein condensation with optimal rate and we give a bound on the ground state energy for the system described by \eqref{eq:H1}.

\begin{theorem}\label{mainTheorem}
Let $V$ be positive, compactly supported, spherically symmetric and bounded, and assume that $\kappa$ is a fixed, small enough constant independent of all parameters and $n\ell^{-1}\leq 1$.
Then, the ground state energy $e_{n,\ell}$ of  $H_{n,\ell}$ defined in \eqref{eq:H1} is such that
\begin{equation}\label{eq:gpEnergy}
\Big|e_{n,\ell}- 4\pi \frak{a}\frac{n^2}{\ell}\Big|\leq C\Big(\frac{n}{\ell}+\frac{n^2}{\ell^2} \ln(\ell)\Big)
\end{equation}
for a constant $C>0$ depending only on V.

Furthermore, let $\psi_n \in L^2_s (\Lambda_1^n)$ be a normalized wave function, with 
\[
\langle \psi_n , H_{n,\ell} \psi_n \rangle \leq e_{n,\ell} + \zeta  
\]
for some $\zeta > 0$.  Let $\gamma_n^{(1)} = \tr_{2,\dots , n} |\psi_n \rangle \langle \psi_n|$ be the one-particle reduced density associated with $\psi_n$. Then there exists a constant $C > 0$ depending only on $V$ such that 
\begin{equation}\label{eq:conv-thm} 1- \langle \ph_0 , \gamma^{(1)}_n \ph_0 \rangle \leq C\Big(\frac{\zeta}{n}+\frac{1}{\ell}\Big) \end{equation}
where $\ph_0 (x) = 1$ for all $x \in \Lambda_1$.

\end{theorem}

\medskip

\noindent \textit{Remarks}.
\begin{enumerate}
 \item For $n=\ell=N$ we recover in \eqref{eq:conv-thm} the condensate depletion rate $N^{-1}$, as shown for example in \cite[Theorem 1.1]{BBCS4} for periodic boundary conditions. However, the ground state energy is different from the translation invariant case and for the trapped case, since here we have
 \begin{equation}
\Big|e_{N}- 4\pi \frak{a}N\Big|\leq C\Big(1+\ln(N)\Big)
\end{equation}
The logarithmic behavior of the error bound is actually sharp and is specific to the Neumann boundary conditions.

  \item We need the requirement that $\kappa$ is small to prove certain properties (see \eqref{eq:PWder} below) of the ground state of the two-body problem in a box with Neumann boundary conditions. 
  \end{enumerate}

To prove Theorem \eqref{mainTheorem}, the main novelty of our analysis is the control of the Neumann boundary conditions. To do so, we consider the energy functional for the two-body problem (see \eqref{eq:Fphi} below) which naturally lives in a six-dimensional space, and we study the properties of its minimizer. We use then the minimizer to describe two-body correlations arising from interactions. In this part, we follow the ideas of \cite{BBCS1}: after transforming the Hamiltonian \eqref{eq:H1} with a unitary operator (taken from \cite{LNSS}) which maps $L^2_s(\L_1^n)$ to Fock space and extracts the contribution of the factorized part of wave functions, we act with a (generalized) Bogoliubov transformation. We define its integral kernel $\eta(x,y)$ as a function of the minimizer of the two-body problem (projected outside the space spanned by the constant in $L^2(\L_1)\times L^2(\L_1)$). In comparison to the case with periodic boundary condition and the case where the system in $\R^3$ is trapped by an external potential, the choice of Neumann boundary conditions  makes the problem considerably more involved from the technical point of view. In the former cases the kernel of the Bogoliubov transformation $\tilde \eta(x,y)$ can be chosen proportional to $(1-f_0(x-y))\ph_0^2(x+y)$ (before projecting it outside the space spanned by the constant in $L^2\times L^2$), where $f_0$ has been defined in \eqref{eq:0enSol} and $\ph_0$ represents the minimizer of the Gross-Pitaevskii functional. In our case instead the integral kernel $\eta(x,y)$ does not have such a simple structure; the center of mass and relative coordinate do not decouple and ground state of the two-body problem is not explicitly known.  While the properties of \eqref{eq:0enSol} can be studied by reducing the problem to a one-dimensional problem (depending only on a radial coordinate), here we need instead to study a full six-dimensional problem on $L^2(\L_1)\times L^2(\L_1)$.
Moreover the Neumann boundary conditions set a non-translation invariant problem with no conserved momentum (this of course rules out also the use of Fourier series and Fourier transforms).

As mentioned above, the Neumann boundary conditions allow us to deduce very easily a lower bound for the leading order of ground state energy of the Bose gas in the thermodynamic limit, for the system described by \eqref{eq:HL}, for a small coupling constant $\kappa>0$. This is the result of Corollary \ref{LHY}. 

\begin{cor}\label{LHY}
 Let $V$ satisfy the same assumptions as in Theorem \ref{mainTheorem} and $\kappa$ be small enough. Then there exists a constant $C>0$ such that $ e (\rho)$ as defined in \eqref{eq:energyPP} satisfies
 \begin{equation}\label{eq:leadingLHY}
\begin{split}
   e (\rho)\geq  4\pi\mathfrak{a}\rho\Big(1-C(\rho\mathfrak{a}^3)^{1/2}\ln(1/\rho)\Big)
   \end{split}
\end{equation}
for $\rho$ small enough. 
\end{cor}

\noindent\textit{Remarks}.
\begin{enumerate}
 \item The bound \eqref{eq:leadingLHY} is not optimal, as evident from \eqref{eq:LHY}. The optimal result has been proved in \cite{BFSol,FSol} with a different localization method which allows to avoid the explicit use of boundary conditions, at the price of dealing with a modified kinetic energy.
 
 \item To obtain Corollary \ref{LHY} we take $\ell$ proportional to $\rho^{-1/2}$. Larger lengths $\ell$ would allow for a better precision in \eqref{eq:leadingLHY}, as achieved in \cite{FSol} mentioned above. However, this requires a more precise study of \eqref{eq:H1}, with larger $n/\ell$. In the translation invariant setting, on a torus with length slightly larger than $\rho^{-1/2}$, Bose-Einstein condensation has been shown in \cite{ABS,F}, while \cite{BCapS} also derives the excitation spectrum of \eqref{eq:H1}.
 
\end{enumerate}

Even though the lower bound \eqref{eq:leadingLHY} is not optimal, we believe that our method can be extended to allow for larger $n/\ell$ (which would yield the Lee-Huang-Yang formula in the thermodynamic limit) as well as  to give the excitation spectrum in the cells, which would allow for giving precise bounds for the free energy at low temperature. We plan to return to this question in a subsequent work. Bounds for the free energy are up to now restricted to leading order \cite{Sei1,Y}.
 
 \medskip

The paper is organized as follows. In Section \ref{excit} we introduce the setting for the analysis of \eqref{eq:H1}, while the detailed estimates are done  in  Section \ref{sec:ExcitHam}. In Section \ref{proofs} we prove Theorem \ref{mainTheorem} and Corollary \ref{LHY}.  In Appendix \ref{functional} we study the two-body problem.

\section{Excitation Hamiltonians}\label{excit}

In this section we focus on the study of $H_{n,\ell}$, defined as in \eqref{eq:H1}. The Hamiltonian $H_{n,\ell}$ acts on $L^2_s(\L_1^n)$, which consists of square-integrable functions on $\L_1^n$ that are symmetric with respect to permutation of the variables. It is convenient to enlarge the space and work  on Fock space, defined as
\[
 \cF=\bigoplus_{m}L_s^2(\L_1^m).
\]
We call vacuum the vector $\Omega = \{ 1, 0, \dots \} \in \cF$.
We define, for $g \in L^2 (\Lambda_1)$, the creation operator $a^*(g)$ and the annihilation operator $a (g)$ as 
\[ \begin{split} 
(a^* (g) \Psi)^{(m)} (x_1, \dots , x_m) &= \frac{1}{\sqrt{m}} \sum_{j=1}^m g (x_j) \Psi^{(m-1)} (x_1, \dots , x_{j-1}, x_{j+1} , \dots , x_m) 
\\
(a (g) \Psi)^{(m)} (x_1, \dots , x_m) &= \sqrt{m+1} \int_\Lambda  \bar{g} (x) \Psi^{(m+1)} (x,x_1, \dots , x_m) \, dx \end{split} \]
The creation operator $a^*(g)$  is the adjoint of the annihilation operator $a (g)$  and they satisfy the canonical commutation relations: for $g,h \in L^2 (\Lambda_1)$,
\begin{equation*}
[a (g), a^* (h) ] = \langle g,h \rangle , \quad [ a(g), a(h)] = [a^* (g), a^* (h) ] = 0 \end{equation*}
We introduce operator valued distributions $\check{a}_x, {a}_x^*$ defined by 
\begin{equation*}
 a(g) = \int \bar{g} (x) \,  {a}_x \, dx , \quad a^* (g) = \int g(x) \, {a}_x^* \, dx  \end{equation*}
 for $g\in L^2(\L_1)$.
It will be convenient to work in the basis of the eigenfunctions of the Laplacian on the cube $\L_1$ with Neumann boundary conditions. We denote with $\{\ph_p\}$, for $p\in\pi\{0,1,2,3,\dots\}^3$ such an orthonormal basis, which is given by $\ph_p(x)=1$ for $p=0$ and, for $p\neq1$,
\[
 \ph_p(x)=\big(1/2\big)^{3/2}\cos(p^{(1)}(x^{(1)}+1/2))\cos(p^{(2)}(x^{(2)}+1/2))\cos(p^{(3)}(x^{(3)}+1/2))
\]
where we used the notation $(x^{(1)},x^{(2)},x^{(3)})$ for the three dimensional vector $x$.
We call $\L_1^*=\pi\{0,1,2,\dots\}^3$ the dual space to $\L_1$. We introduce the space $\L^*_{1,+}=\L^*_1\backslash\{0\}$, where the zero momentum is removed.
We adopt the notation
\begin{equation}\label{eq:ap} 
\hat a^*_p = a^* (\ph_p), \quad \text{and } \quad  \hat a_p = a (\ph_p).
\end{equation} 
 We call the number of particles operator on $\cF$ the operator
\[ \cN = \sum_{p \in \Lambda_1^*}\hat a_p^*\hat a_p
= \int {a}^*_x {a}_x \, dx. \, 
\]
Creation and annihilation operators are bounded with respect to $\cN$; it is easy to check that, for all $g \in L^2 (\Lambda_1)$,
\begin{equation*}
\| a(g)\Psi \| \leq \|g\| \| \cN^{1/2} \Psi \|, \quad \|a^*(g) \Psi\| \leq \|g\| \|(\cN+1)^{1/2}\Psi\|.
\end{equation*}
The Hamiltonian \eqref{eq:H1} lifted to Fock space takes the form
\begin{equation}\label{eq:HnGPfs}
   H_{n,\ell}=\sum_{p\in\L_1^*}p^2 \hat a^*_p\hat a_p+\frac{1}{2}\sum_{p,q,r,s\in\L_1^*} V_{\ell,pqrs} \hat  a^*_p\hat  a^*_q\hat a_r\hat a_s,
\end{equation}
with
\begin{equation}\label{eq:defV}
 \begin{split}
 V_{\ell,pqrs}&=\langle\ph_p\otimes\ph_q,\kappa\ell^2V(\ell\cdot)\,\ph_r\otimes\ph_s\rangle=\int_{\L_1} dx\int_{\L_1} dy \,\,\kappa\ell^2 V(\ell(x-y))\ph_p(x)\ph_q(y)\ph_r(x)\ph_s(y).
 \end{split}
\end{equation}

The eigenfunction of the Laplacian $\ph_0(x)=1$ corresponding to the lowest eigenvalue $p^2=0$ represents the condensate wave function. It is convenient to separate the contribution of the zero mode and consider a modified Fock space describing excitations.
We define 
\begin{equation}\label{eq:truncatedFock}
   \cF_+^{\leq n}=\bigoplus_{j=0}^nL_+^2(\L_1)^{\otimes_sj},
\end{equation}
where $L_+^2(\L_1)$ is the orthogonal complement to the one dimensional space spanned by $\ph_0$ in $L^2(\L_1)$. Additionally, in definition \eqref{eq:truncatedFock}, we truncated the Fock space up to the sector with $n$ particles. A vector $\Psi = \{ \psi^{(0)}, \psi^{(1)}, \dots , \psi^{(n)}, 0, 0,\dots\} \in \cF$ lies in $\cF_+^{\leq n}$, if $\psi^{(m)}$ is orthogonal to $\ph_0$, in each of its coordinates, for all $m = 1, \dots,n$, i.e. if
\[ \int \bar{\ph_0} (x) \, \psi^{(m)} (x,y_1, \dots , y_{m-1}) dx = 0 \]
for all $m = 1, \dots,n$.
On $\cF_+^{\leq n}$, we denote the number of particles operator with $\cN_+=\cN|_{\cF_+^{\leq n}}$.
We will use modified creation and annihilation operators
\[ b (f) = \sqrt{\frac{n- \cN_+}{n}} \, a (f), \qquad \text{and } \quad  b^* (f) = a^* (f) \, \sqrt{\frac{n-\cN_+}{n}}. \]
If $f \in L^2_+ (\Lambda_1)$, $b(f),\,b^*(f)$ map $\cF_+^{\leq n}$ into itself. Moreover, for $g\in L^2 (\Lambda_1)$ and $Q=1-|\ph_0\rangle\langle\ph_0|$, $b(g)=b(Qg)$ on $\cF_+^{\leq n}$.
Analogously as before, we define operator valued distributions $ b_x,  b^*_x$ as
\[ b(f) = \int \bar{f} (x) \,  b_x \, dx , \qquad \text{and } \quad b^* (f) = \int f(x)  b^*_x \, dx \]
satisfying modified canonical commutation relations
\begin{equation}\label{eq:comm-b}
\begin{split}  [ b_x,  b_y^* ] &= \left( 1 - \frac{\cN_+}{n} \right) \delta (x-y) - \frac{1}{n}  a_y^*  a_x \\ 
[  b_x,  b_y ] &= [  b_x^* ,  b_y^*] = 0 
\end{split} \end{equation}
and we define
\begin{equation}\label{eq:bp} 
\hat b^*_p = b^* (\ph_p), \quad \text{and } \quad  \hat b_p = b (\ph_p).
\end{equation}
Every $n$-particle wave function $\psi_n\in L^2(\L_1^n)$ can be decomposed uniquely as 
\[ \psi_n =  \sum_{m=0}^n \alpha^{(m)} \otimes_s \ph_0^{\otimes (n-m)} 
\]
where $\alpha^{(m)} \in L^2_+ (\Lambda_1)^{\otimes_s m}$ for all $m =1 , \dots , n$. 
Following \cite{LNSS}, we define a unitary operator $U_n:L^2_s (\Lambda_1^n)\to \cF_+^{\leq n}$ such that
\[
U_n \psi_n = \{ \alpha^{(0)}, \alpha^{(1)}, \dots , \alpha^{(n)} \},
\]
i.e., the unitary map $U_n$ removes the condensate contribution from $\psi_n\in L^2_s (\Lambda_1^n) $ and returns the excitations over the condensate.
As shown in \cite{LNSS}, when we conjugate couples of creation and annihilation operators with $U_n$, we obtain, for $p,q\in\L^*_{1,+} $,
\begin{equation}\label{eq:U}
\begin{split} 
U_n \, \hat a^*_0\hat  a_0 \, U_n^* &= N- \cN_+  \\
U_n \, \hat a^*_p\hat  a_0 \, U_n^* &= \hat a^*_p \sqrt{N-\cN_+ } \\
U_n \,\hat  a^*_0\hat  a_p \, U_n^* &= \sqrt{N-\cN_+ } \, \hat a_p \\
U_n \,\hat  a^*_p\hat  a_q \, U_n^* &= \hat a^*_p \hat a_q 
\end{split} \end{equation}
The operator $\cN_+=\cN-a^*_0a_0$ counts the number of excitations. Using the properties of $U_n$, it is easy to see that $U_n^*\Omega=\ph^{\otimes N}$.

 With the transformation $U_n$ we define the excitation Hamiltonian
\begin{equation}
 \begin{split}
 \cL_n:=U_nH_{n,\ell}U_n^*:\cF_+^{\leq n}\to \cF_+^{\leq n}
  \end{split}
\end{equation}
As shown in \cite{LNSS}, $ \cL_n$ consists of the sum
\begin{equation}
 \begin{split}
 \cL_{n,\ell}=\cL_{n,\ell}^{(0)}+\cL_{n,\ell}^{(1)}+\cL_{n,\ell}^{(2)}+\cL_{n,\ell}^{(3)}+\cL_{n,\ell}^{(4)}
  \end{split}
\end{equation}
with
 \begin{equation}\label{eq:cLp}
 \begin{split}
 \cL_{n,\ell}^{(0)}&=\frac{1}{2} V_{\ell,0000} (n-\cN_+)(n-\cN_+-1)\\
 \cL_{n,\ell}^{(1)}&=\sqrt{n}\sum_{p\in \L^*_{1,+}}V_{\ell,000p}(n-\cN_+-1)\hat b_p+\hc\\
 \cL_{n,\ell}^{(2)}&=\sum_{p\in \L^*_{1,+}}p^2\hat a^*_p\hat a_p+\sum_{p,q\in \L^*_{1,+}}(V_{\ell,0p0q}+V_{\ell,0pq0})\hat a^*_p\hat a_q(n-\cN_+)\\
 &+\frac{1}{2}\sum_{p,q\in \L^*_{1,+}}(nV_{\ell,pq00}\hat b^*_p\hat b^*_q+\hc)\\
 \cL_{n,\ell}^{(3)}&=\sum_{p,q,r\in \L^*_{1,+}}(n^{1/2}V_{\ell,0pqr}\hat a^*_p\hat a_q\hat b_r+\hc)\\
 \cL_{n,\ell}^{(4)}&=\frac{1}{2}\sum_{p,q,r,s\in \L^*_{1,+}}V_{\ell,pqrs} \hat a^*_p\hat a^*_q\hat a_r\hat a_s
 \end{split}
\end{equation}

Conjugation with the map $U_n$ does not extract from $H_{n,\ell}$ all the leading order contributions to the energy (by taking the vacuum expectation value); it extracts the contribution of the condensate part of wave functions,  but it does not extract the contribution due to correlations (recall that $\langle\Omega,\cL_{n,\ell}^{(0)}\Omega\rangle= \langle \ph^{\otimes N},H_{n,\ell} \ph^{\otimes N}\rangle$). 
In fact, the ground state wave function is far from being factorized and correlations among particles play a crucial role. In order to describe the correlation structure of the ground state wave function we need to transform $\cL_{n}$ further.

To model correlations we use the solution of the two-body problem with potential $V$: this describes the simplest scattering process. We find it more natural to work now on the rescaled double box $\L_\ell\times\L_\ell$, and impose Neumann boundary conditions (recall that $\L_\ell=[-\ell/2,\ell/2]$). We look for the minimizer of the functional   \begin{equation}\label{eq:Fphi} 
\begin{split}
 F[g]=\int_{\L_\ell\times\L_\ell} dxdy\, \Big[\kappa V(x-y)|g( x, y)|^2+|\nabla_xg(x, y)|^2+|\nabla_yg(x, y)|^2\Big]
\end{split} 
\end{equation}
among functions $g\in H^{1}(\L_\ell\times\L_\ell)$ with $\|g\|_{L^2(\L_\ell\times\L_\ell)}=1$. In the next proposition we state the properties of the minimizer we shall need.\\

\begin{prop}\label{prop:Ff} Let $\ell>1$ and $\L_\ell=[-\ell/2,\ell/2]^3\subset\mathbb{R}^3$. Let the functional $F:H^1(\L_\ell\times \L_\ell)\to\mathbb{R}$ be defined in \eqref{eq:Fphi}.
Then, in the subclass of functions such that
\[
 \|g\|^2_2=\int_{\L_\ell\times\L_\ell} dxdy\, |g( x, y)|^2=1,
\]
there is a unique function $f$ (up to a constant phase factor) that minimizes $ F$. 
If $\mathfrak{a}$ is the scattering length of the potential $V$ (defined in \eqref{eq:a0}), we have, for $\ell $ sufficiently large,
  \begin{equation}\label{eq:slength} 
\begin{split}
  \lambda_\ell:=\inf_{\substack{g\in H^1(\L_\ell\times \L_\ell)} }\left\{F[g]: \, \int_{\L_\ell\times\L_\ell} dxdy\, |g( x, y)|^2=1,\right\}=\frac{8\pi \mathfrak{a}}{\ell^3}\Big(1+\mathcal{O}\Big(\frac{\mathfrak{a}}{\ell}\ln (\ell/\mathfrak{a})\Big)\Big)
\end{split} 
\end{equation}
Moreover, the following properties of the minimizer $f$ hold.
\begin{enumerate}
\item [i)] We have
  \begin{equation}\label{eq:der} 
\begin{split}
 \int_{\L_\ell\times\L_\ell} dxdy\,\Big[|\nabla_xf(x, y)|^2+|\nabla_yf(x, y)|^2\Big]\leq C\kappa \ell^{-3}
\end{split} 
\end{equation}
\item[ii)] There exists a constant $C>0$ such that
  \begin{equation}\label{eq:sup} 
\begin{split}
|f(x, y)|\leq C\ell^{-3}
\end{split} 
\end{equation}
for every $x,y\in\L_\ell$ .
\item [iii)] There exists a constant $C>0$ such that
\begin{equation}\label{eq:norm} 
\begin{split}
 \int_{\L_\ell\times\L_\ell} dxdy\,\big|\ell^{-3}-f(x, y)\big|^2\leq C^2\kappa^2\ell^{-2}
\end{split} 
\end{equation}
and 
  \begin{equation}\label{eq:L1norm} 
\begin{split}
 \int_{\L_\ell\times\L_\ell} dxdy\,\big|\ell^{-3}-f(x, y)\big|\leq C\k\ell^2
\end{split} 
\end{equation}
\item [iv)] There exists a constant $C>0$ such that
  \begin{equation}\label{eq:PWdecay} 
\begin{split}
|1-\ell^{3}f(x, y)|\leq C\k\left(\frac{1}{|x-y|+1}\right)
\end{split} 
\end{equation}
 \item [v)] For $\k$ small enough, there exists a constant $C>0$ such that 
 \begin{equation}\label{eq:PWder}
 \begin{split}
|\nabla_{x+y}f(x,y)|&\leq C\kappa\ell^{-3}\big(d\big(\tfrac{x+y}2\big)^{5/3}+1\big)^{-1}
   \end{split} 
 \end{equation}
 where  $d(x)$ is the  distance of $x$ to the boundary of the box $\Lambda_\ell$. 
\end{enumerate}
\end{prop}
We postpone the proof of Prop.~\ref{prop:Ff} to Appendix \ref{functional}. The minimizer of \eqref{eq:Fphi} satisfies the  eigenvalue equation 
\begin{equation}\label{eq:6dScat} 
\begin{split}
\Big[-(\Delta_x+\Delta_y)+\kappa V(x-y)\Big]f(x,y)=\lambda_\ell f( x,y),
\end{split} 
\end{equation}
for $x,y\in\L_\ell$. We define $f_\ell(x,y)=f(\ell x,\ell y)$; by scaling, $f_\ell(x,y)$ satisfies
\begin{equation}\label{eq:6dScatResc} 
\begin{split}
\Big[-(\Delta_x+\Delta_y)+\k\ell^2V(\ell(x-y))\Big]f_\ell(x,y)=\ell^2\lambda_\ell f_\ell( x,y),
\end{split} 
\end{equation}
for $x,y\in\L_1$.  We set $w_\ell=1-\ell^3f_\ell$, which solves
\begin{equation}\label{eq:6dScatRescOm} 
\begin{split}
(\Delta_x+\Delta_y)w_\ell(x,y)+\k\ell^5V(\ell(x-y))f_\ell(x,y)=\ell^5\lambda_\ell f_\ell( x,y),
\end{split} 
\end{equation}
for $x,y\in\L_1$. Using the function $w_\ell$, we construct a Hilbert-Schmidt operator  $\eta:L^2(\L_1)\to L^2(\L_1)$. We set
\begin{equation}\label{eq:etaQkQ}
 \eta=(1-|\ph_0\rangle\langle\ph_0|)k(1-|\ph_0\rangle\langle\ph_0|)
\end{equation}
where $k:L^2(\L_1)\to L^2(\L_1)$ is the Hilbert-Schmidt operator with integral kernel
\begin{equation}\label{eq:kappa}
  k(x,y)=-nw_\ell(x,y)
\end{equation}
It will be useful to decompose $\eta=k+\mu$, with 
\begin{equation}\label{eq:etaDecomp}
  \mu=-|\ph_0\rangle\langle\ph_0|k -k|\ph_0\rangle\langle\ph_0|+|\ph_0\rangle\langle\ph_0|k|\ph_0\rangle\langle\ph_0|
\end{equation}
Therefore, we can express the integral kernel of the operator $\eta$ as
\begin{equation}\label{eq:etaQkQIK}
 \eta(x,y)= k(x,y)+ \mu(x,y)
\end{equation}
with
\[\begin{split}
 \mu(x,y)
 =\,&n\int dz\,w_\ell(z,y) +n\int dz \,w_\ell(x,z)-n\int dz_1dz_2\,w_\ell(z_1,z_2)
 \end{split}
\]
Using \eqref{eq:6dScatRescOm} and \eqref{eq:etaQkQIK} we have
\begin{equation}\label{eq:deltaeta}
\begin{split} 
(\Delta_x+\Delta_y) \eta(x,y) = n\ell^5 ( \kappa V(\ell(x-y)) - \lambda_\ell) f_\ell(x,y)  + n  \int dz\, \Delta_xw_\ell(x,z) + n  \int dz\, \Delta_yw_\ell(z,y) 
\end{split} 
\end{equation}

We collect in Proposition \ref{prop:eta} below properties of the operators $\eta$, $k$ and $\mu$ (we postpone the proof of Proposition \ref{prop:eta} to the end of Appendix \ref{functional}).

\begin{prop}\label{prop:eta}
 Let $\eta$ be defined as in \eqref{eq:etaQkQ} and let $\kappa$ be small enough and $n/\ell\leq1$. Then the following estimates hold true.
 \begin{itemize}
  \item[i)] We have
  \begin{equation}\label{eq:normEta}
 \begin{split}
 \int_{\L_1\times\L_1} dxdy\,\big|\eta(x, y)\big|^2=\|\eta\|^2_2\leq C\k^2 n^2\ell^{-2}
\end{split} 
\end{equation}
and
  \begin{equation}\label{eq:normDEta}
 \begin{split}
 \int_{\L_1\times\L_1} dxdy\,\Big[\big|\nabla_x\eta(x, y)\big|^2+\big|\nabla_y\eta(x, y)\big|^2\Big]\leq C\k n^2\ell^{-1}
\end{split} 
\end{equation}
for a constant $C>0$.
Moreover, for any $x,y\in\L_1$,
\begin{equation}\label{eq:supEta}
 |\eta(x,y)|\leq C n
\end{equation}
and
\begin{equation}\label{eq:supEtaBetter}
 |\eta(x,y)|\leq C \frac{\k n}{\ell}\left[\frac{1}{|x-y|+\ell^{-1}}\right]
\end{equation}
for a constant $C>0$.
\item[ii)] We indicate with $\eta_x(y)$ the function $\eta(y,x)$. For any $x\in\L_1$
\begin{equation}\label{eq:supEtaL2}
 \|\eta_x\|_2\leq C \k n\ell^{-1}
\end{equation}
for a constant $C>0$.

\item[iii)] Decomposing\footnote{As we did for $\eta$, we are going to use the symbols $\sigma, \, r$ and $p$ to indicate both the operators and their integral kernels.} $\s:=\sinh(\eta)=\eta+r$ and $\g:=\cosh(\eta)=1+p$, there exists a $C>0$ such that
  \begin{equation}\label{eq:normsp}
 \begin{split}
\|\s\|_2,\|p\|_2\leq C\|\eta\|_2
  \end{split}
\end{equation}
Moreover
  \begin{equation}\label{eq:supr}
 \begin{split}
|r(x,y)|\leq C\|\eta\|_2\|\eta_x\|_2\|\eta_y\|_2, \quad|p(x,y)|\leq C\|\eta_x\|_2\|\eta_y\|_2
  \end{split}
\end{equation}
for every $x,y \in \L_1$. This implies that
  \begin{equation}\label{eq:supL2}
 \begin{split}
\|r_x\|_2\leq C\|\eta\|^2_2\|\eta_x\|_2, \quad\|p_x\|\leq C\|\eta\|_2\|\eta_x\|_2
  \end{split}
\end{equation}

\end{itemize}
\end{prop}


\vspace{1cm}

With $\eta$ introduced above, we define the generalized Bogoliubov transformation 
\begin{equation}\label{eq:defB}
 e^B=\exp{\left[\frac{1}{2}\int_{\L_1\times\L_1}dx dy\,\eta(x,y)\,b^*_xb^*_y-\hc\right]}
\end{equation}
Equivalently we can express it as
\begin{equation}\label{eq:Bsum}
e^B=\exp{\Big[\frac{1}{2}\sum_{p,q\in\L_{1,+}^*}\big(\eta_{pq}\hat b^*_p\hat b^*_q-\hc\big)\Big]}
\end{equation}
with
\begin{equation}\label{eq:etapq}
 \eta_{pq}=\langle\ph_p\otimes\ph_q,\eta\rangle
\end{equation}
Note that $e^B:\cF_+^{\leq n}\to \cF_+^{\leq n}$ is unitary.
In Section \ref{Bogo} below we present some key properties of $e^B$. 
With the generalized Bogoliubov transformation $e^B$ we define a new excitation Hamiltonian  $\cG_{n,\ell}: \cF^{\leq n}_+\to\cF^{\leq n}_+$ as
\begin{equation}\label{eq:defG}
 \cG_{n,\ell}=e^{-B}\cL_{n,\ell}e^B=e^{-B}U_nH_{n,\ell}U_n^*e^B
\end{equation}
Proposition \ref{prop:G} (which will be proved at the end of Section \ref{sec:ExcitHam}) presents  our  main estimates of $\cG_{n,\ell}$.

\begin{prop}\label{prop:G}
 Let $V$ be positive, compactly supported, spherically symmetric and bounded. 
 Moreover, define
 \begin{equation}\label{eq:excitHam} 
\begin{split} 
\cK=\sum_{p\in\L^*_{1,+}}p^2 \hat a^*_p\hat a_p\qquad\text{and}\qquad\cV_\ell=\frac{1}{2}\sum_{p,q,r,s\in\L^*_{1,+}}V_{\ell,pqrs} \hat a^*_p\hat a^*_q\hat a_r\hat a_s.
\end{split} 
\end{equation}
where we defined $V_{\ell,pqrs}$ in \eqref{eq:defV}.
 For  $\kappa$ be small enough and $n/\ell\leq1$, we have
 \begin{equation}\label{eq:estForCond}
\begin{split} 
\cG_{n,\ell}  &= C_{n,\ell}+\cK +\cV_\ell
+\cE_{n,\ell}
\end{split} 
\end{equation}
where $C_{n,\ell}$ is given by
 \begin{equation}\label{eq:constantTerm}
\begin{split} 
C_{n,\ell} &= \frac{n^2}{2} \int dxdy\,\k\ell^2 V(\ell(x-y))-\frac{1}{2}\langle\eta,(\Delta_1+\Delta_2)\eta\rangle\\
&\quad+  n\sum_{p,q\in\L^*_{1,+}}V_{\ell,pq00}\langle\eta,\ph_p\otimes\ph_q\rangle+\frac{1}{2}\sum_{p,q,r,s\in\L^*_{1,+}}V_{\ell,pqrs}\langle\ph_s\otimes\ph_r,\eta\rangle\langle\eta,\ph_p\otimes\ph_q\rangle,
\end{split} 
\end{equation}
and the error $\cE_{n,\ell}$ is such that for any $\delta>0$ there exists a constant $C>0$ so that
 \begin{equation}
\begin{split} 
 \pm\cE_{n,\ell}\leq \delta (\cK+\cV_{\ell})+C\k\frac{n}{\ell}(\cN_++1)
\end{split} 
\end{equation}

\end{prop}

\section{
Analysis of the excitation Hamiltonian}\label{sec:ExcitHam}

In this section we analyze the excitation Hamiltonian $\cG_{n,\ell}$ defined in \eqref{eq:defG}. We decompose it as 
\[
 \cG_{n,\ell}=\cG_{n,\ell}^{(0)}+\cG_{n,\ell}^{(1)}+\cG_{n,\ell}^{(2)}+\cG_{n,\ell}^{(3)}+\cG_{n,\ell}^{(4)}
\]
with
\[
 \cG_{n,\ell}^{(j)}=e^{-B}\cL_{n,\ell}^{(j)}e^B
\]
where $\cL_{n,\ell}^{(j)}$ was defined in \eqref{eq:cLp}, for $j\in \{0,1,2,3,4\}$. We examine $\cG_{n,\ell}$ and identify its main contributions. The goal of this section is to prove Proposition \ref{prop:G}. While the analysis  is similar to \cite[Section 4]{BBCS1}, special care needs to be taken due to the Neumann boundary conditions.  
In Subsections \ref{g0}, \ref{g1}, \ref{g2}, \ref{g3}, \ref{g4} we extract from $\cG_{n,\ell}^{(0)},\, \cG_{n,\ell}^{(1)}, \, \cG_{n,\ell}^{(2)}, \,\cG_{n,\ell}^{(3)}$ and $\cG_{n,\ell}^{(4)}$ the main contributions which will be expressions that are constant, linear and quadratic in creation and annihilation operators, and we prove that cubic and quartic contributions are small. In Subsection \ref{proofProp} we bound the linear and quadratic contributions, obtaining Proposition \ref{prop:G}.
Throughout the whole section we will use some properties of the generalized Bogoliubov transformation $e^B$, which we review in Subsection \ref{Bogo}.

\subsection{Generalized Bogoliubov transformation}\label{Bogo}

The generalized Bogoliubov transformation in the form \eqref{eq:defB} has been introduced in \cite{BS}; we refer to \cite[Section 3]{BS} for a detailed discussion about it; we mention below only the results that are relevant in our analysis.

As proved in \cite{Sei,BS}, $e^B$ does not change substantially the number of excitations. This is the content of the following Lemma. 

\begin{lemma}\label{lm:TNT}
Let $\eta \in L^2(\L_1\times\L_1)$ be such that $\eta(x,y)=\eta(y,x)$ and let $B$ be defined as in (\ref{eq:defB}). Then, for every $m_1, m_2 \in \bZ$, there exists a constant $C > 0$  such that, on $\cF_+^{\leq n}$, 
\[ e^{-B} (\cN_+ +1)^{m_1} (n+1-\cN_+ )^{m_2} e^{B} \leq Ce^{C\|\eta\|_2} (\cN_+ +1)^{m_1} (n+1- \cN_+ )^{m_2} \, \]
\end{lemma}

The action of $e^B$ on creation and annihilation operators can be expressed as follows. We define 
\begin{equation}\label{eq:defGaSi} \begin{split} 
{\gamma}_x (y) &\; =\cosh_\eta(y,x) = \sum_{n \geq 0} \frac{1}{(2n)!} \,\eta^{2n}(y,x)\,,  \\    
{\s}_x (y) & \; =\sinh_\eta(y,x)= \sum_{n \geq 0} \frac{1}{(2n+1)!} \eta^{2n+1}(y,x)\,,
\end{split} \end{equation}
where $\eta^m$ indicates the product in the sense of operators (the symbol $\eta$ denotes the Hilbert-Schmidt operator whose kernel is $\eta(x,y)$). Note that $\eta^0(y,x)$ has to be interpreted as a $\delta$ distribution.
With these definitions, we write
\begin{equation}
 \begin{split}\label{eq:d}
   e^{-B} {b}_x e^{B} = b ( {\g}_x)  +  b^* ({\s}_x) + {d}_x, \qquad 
e^{-B} {b}^*_x e^{B} = b^* ( {\gamma}_x)  +  b ({\s}_x) + {d}^*_x
 \end{split}
\end{equation}
for a remainder operator $d_x$.
Lemma \ref{lm:dx} below states that $d_x$ is a bounded operators on $\cF^{\leq n}_+$  and it is small on states with a small number of excitations; the main contributions in the right hand sides of \eqref{eq:d} correspond to those of the usual Bogoliubov transformation.
Lemma \ref{lm:dx} is a generalization of \cite[Lemma 2.3]{BBCS3}, and can be proved in the same way.

\begin{lemma}\label{lm:dx}
 Let $\eta\in L^2(\L_1\times\L_1)$ be such that $\eta(x,y)=\eta(y,x)$ and let $j\in \mathbb{Z}$. Let the remainder operator $d_x$ be defined as in \eqref{eq:d}. Then, if $\|\eta\|$ is small enough, there exists a $C>0$ such that 
 \begin{equation}\label{eq:dxy-bds} 
 \begin{split} 
 \| (\cN_+ + 1)^{j/2} {d}_x \xi \| & \leq n^{-1} C\, \Big[ \, \|\eta_x\| \| (\cN_+ + 1)^{(j+3)/2} \xi \| + \| \eta \| \| b_x (\cN_+ + 1) ^{(j+2)/2}\xi \| \Big] \\[6pt]
 \end{split} \end{equation}
  \begin{equation}\label{eq:dxy-bdsN2} 
 \begin{split} 
 \| (\cN_+ + 1)^{j/2} {a}_y {d}_x \xi \| &\leq n^{-1}C\, \Big[ \, \|\eta_x\| \|\eta_y\|  \| (\cN_+ + 1)^{(j+2)/2} \xi \|  + |\eta(y,x)| \| (\cN_+ +1)^{(j+2)/2}  \xi \| \\
& \hspace{1.7 cm} + \| \eta_y \| \| b_x (\cN_++1)^{(j+1)/2} \xi \| +  \|\eta_x\| \|a_y (\cN_+ + 1)^{(j+ 3)/2} \xi \|\\
& \hspace{1.7 cm}  + \| \eta \| \|a_y a_x (\cN_+ +1)^{(j+2)/2}  \xi \|   \, \Big] \\[3pt]
\end{split} \end{equation}
 \begin{equation}\label{eq:dxy-bdsN3} 
 \begin{split} 
 \| (\cN_+ + 1)^{j/2} d_x d_y \xi \|
&\leq  n^{-2}C\Big [ \| \eta_x \|\| \eta_y \| \|(\cN_++1)^{(j+6)/2}  \xi \|+\| \eta \|\| \eta_x \| \|a_y(\cN_++1)^{(j+5)/2} \xi \|\\
& \hspace{0.6 cm} +  \| \eta \||\eta(y,x)| \| (\cN_+ +1)^{(j+4)/2}  \xi \|  +\| \eta \|  \|\eta_y\| \|a_x (\cN_+ + 1)^{(j+ 5)/2} \xi \|\\
& \hspace{0.6 cm}  + \| \eta \| \|a_y a_x (\cN_+ +1)^{(j+4)/2}  \xi \|   \, \Big] \\
\end{split} \end{equation}
for all $\xi \in \cF^{\leq n}_+$. Moreover, for any $g\in L^2(\L_1\times\L_1)$ such that $g(x,y)=g(y,x)$,
 \begin{equation}\label{eq:dxy-bds2} 
 \begin{split} 
 \| (\cN_+ + 1)^{j/2}\int dx\,g(x,y) {d}_x \xi \| & \leq n^{-1} C \|g_y\| \| (\cN_+ + 1)^{(j+3)/2} \xi \|\\
  \| (\cN_+ + 1)^{j/2}\int dx\,g(x,y) {d}^*_x \xi \| & \leq n^{-1} C \|g_y\| \| (\cN_+ + 1)^{(j+3)/2} \xi \|\\
\end{split} \end{equation}
\end{lemma}

\subsection{Analysis of $\cG_{n,\ell}^{(0)}$}\label{g0}
Recall from \eqref{eq:cLp} that
\begin{equation}\label{eq:cLn0} 
\begin{split} 
\cL^{(0)}_{n,\ell} = \frac{1}{2} \int dxdy\,\k\ell^2 V(\ell(x-y)) (n-\cN_+)(n-\cN_+-1)\\
\end{split} 
\end{equation}
We define the error operator $\cE_{n,\ell}^{(0)}$ through
\begin{equation}\label{eq:termsG0}
\begin{split} 
\cG_{n,\ell}^{(0)}= & \frac{n^2}{2} \int dxdy\,\k\ell^2V(\ell(x-y))+\cE_{n,\ell}^{(0)}\\
\end{split} 
\end{equation}
and we estimate it in the next proposition.
\begin{prop}\label{prop:G0}
Let $\cE_{n,\ell}^{(0)}$  be as defined in \eqref{eq:termsG0}. Then, under the same assumptions as in Proposition \ref{prop:G}, there exists a $C > 0$ such that
\begin{equation}\label{eq:errL0}
\begin{split} 
\pm\cE_{n,\ell}^{(0)}\leq  C\k n\ell^{-1}\cN_+\\ 
\end{split} 
\end{equation}
as operator inequalities on $\cF^{\leq n}_+$.
\end{prop}
\begin{proof}
Equation \eqref{eq:termsG0} implies that
\begin{equation}\label{eq:cE0}
\begin{split} 
\cE_{n,\ell}^{(0)}= 
 &-  \, e^{-B}\big(n+n\cN_++\cN_+/2-\cN_+^2/2\big)e^{B}\int dxdy\,\k\ell^2 V(\ell(x-y))\
\end{split} 
\end{equation}
The bounds in \eqref{eq:errL0} follow from $ \int dxdy\,\k\ell^3V(\ell(x-y))\leq \k\int dx \,V(x)$, Lemma \ref{lm:TNT} and \eqref{eq:normEta}.
\end{proof}

\subsection{Analysis of $\cG_{n,\ell}^{(1)}$}\label{g1}

On $\cF_+^{\leq n}$ we can write \eqref{eq:cLp} as 
\begin{equation}\label{eq:cLn1b} 
\begin{split} 
 \cL_{n,\ell}^{(1)}&=\sqrt{n}(n-\cN_+-1)\int dx dy\,\k\ell^2V(\ell(x-y)) b_x +\hc \\
\end{split} 
\end{equation}
We define the error $\cE_n^{(1)}$ by
\begin{equation}\label{eq:cGn1main} 
\begin{split} 
\cG_{n,\ell}^{(1)}= e^{-B}\cL_{n,\ell}^{(1)}e^B&=n^{3/2}\int dx dy\,\k\ell^2V(\ell(x-y)) \big[b(\gamma_x)+b^*(\sigma_x) +\hc\big] +\cE_{n,\ell}^{(1)}
\end{split} 
\end{equation}
where $\gamma_x$ and $\sigma_x$ were defined in \eqref{eq:defGaSi}.
We estimate $\cE_{n,\ell}^{(1)}$ in the next proposition.

\begin{prop}\label{prop:G1} Let $\cE_{n,\ell}^{(1)}$ be defined as in \eqref{eq:cGn1main}. Then, under the same assumptions as Proposition \ref{prop:G}, there exists a $C > 0$ such that
\begin{equation}\label{eq:bounds1}
\begin{split} 
 \pm\,\cE_{n,\ell}^{(1)}\leq C \k n\ell^{-1}(\cN_++1)
\end{split} 
\end{equation}
as operator inequalities on $\cF^{\leq n}_+$.
\end{prop}
\begin{proof}
Comparing \eqref{eq:cLn1b} and \eqref{eq:cGn1main} we obtain 
\begin{equation}\label{eq:cGn1Err} 
\begin{split} 
\cE_{n,\ell}^{(1)}&=-n^{1/2}\int dx dy\,\k\ell^2V(\ell(x-y))\,\Big[e^{-B}(\cN_++1)b_xe^{B}+\hc\Big]\\
&\quad+n^{3/2}\int dx dy\,\k\ell^2V(\ell(x-y)) \Big[e^{-B}b_xe^{B}-\big(b(\gamma_x)+b^*(\sigma_x)\big) +\hc\Big] \\
&=:\text{D}_1+\text{D}_2
\end{split} 
\end{equation}
We analyze $\text{D}_1$ first. Using the identity $(\cN_++1)^{1/2}b_x=b_x\cN_+^{1/2}$, we write it as
\begin{equation}
\begin{split} 
\text{D}_1=-n^{1/2}\int dx dy\,\k\ell^2V(\ell(x-y))\,\Big[e^{-B}(\cN_++1)^{1/2}b_x\cN_+^{1/2}e^{B}+\hc\Big]\\
\end{split} 
\end{equation}
For any $\xi\in\cF_+^{\leq n}$ we have
 \begin{equation}\nonumber 
\begin{split} 
 |\langle\xi,\text{D}_1\xi\rangle| &\leq C n^{1/2}\ell^{-1}\int  dx\,|\langle\xi,e^{-B}(\cN_++1)^{1/2}b_x\cN_+^{1/2}e^{B}\xi\rangle|\int dy\,\k\ell^3V(\ell(x-y))  \\
 &\leq C\k n^{1/2}\ell^{-1}\|(\cN_++1)^{1/2}e^{B}\xi\|\int  dx\,\|b_x\cN_+^{1/2}e^{B}\xi\|  \\
\end{split} 
\end{equation}
With Lemma \ref{lm:TNT} and Cauchy-Schwarz we obtain
 \begin{equation}\nonumber 
\begin{split} 
 |\langle\xi,\text{D}_1\xi\rangle|  &\leq C\k n^{1/2}\ell^{-1}\|(\cN_++1)^{1/2}\xi\|\|(\cN_++1)\xi\| \leq C\k n\ell^{-1}\|(\cN_++1)^{1/2}\xi\|^2
\end{split} 
\end{equation}

We consider now $\text{D}_2$. Using \eqref{eq:d}, we have
\begin{equation}
\begin{split} 
\text{D}_2&=n^{3/2}\int dx dy\,\k\ell^2V(\ell(x-y)) \big[d_x +\hc\big] \\
\end{split} 
\end{equation}
By \eqref{eq:dxy-bds} (with $j=-1$) and Cauchy-Schwarz, we conclude that 
\begin{equation}
\begin{split} 
 |\langle\xi,\text{D}_2\xi\rangle| &\leq n^{3/2}\ell^{-1}\int dx\,|\langle(\cN_++1)^{1/2}\xi,(\cN_++1)^{-1/2} d_x\xi\rangle| \int dy\,\k\ell^3V(\ell(x-y))\\
 &\leq C\k n^{1/2}\ell^{-1}\|(\cN_++1)^{1/2}\xi\|\int dx\,\Big[ \, \|\eta_x\| \| (\cN_+ + 1) \xi \| + \| \eta \| \| b_x(\cN_++1)^{1/2} \xi \| \Big] \\
 &\leq C\k n^{1/2}\ell^{-1}\|\eta\|\|(\cN_++1)^{1/2}\xi\| \| (\cN_+ + 1) \xi \| 
\end{split} 
\end{equation}
This concludes the proof of estimate \eqref{eq:bounds1} if we use bound \eqref{eq:normEta} for the norm of $\eta$.
\end{proof}

\subsection{Analysis of $\cG_{n,\ell}^{(2)}$}\label{g2}

Recall, from \eqref{eq:cLp} and \eqref{eq:excitHam}, that
\begin{equation}\nonumber
\begin{split} 
 \cL_{n,\ell}^{(2)}&=\cK+\cL_{n,\ell}^{(2,V)}
\end{split} 
\end{equation}
with
\begin{equation}\label{eq:QV}
\begin{split} 
  \cL_{n,\ell}^{(2,V)}&=\sum_{p,q\in\L^*_{1,+}}(V_{\ell,0p0q}+V_{\ell,0pq0})\hat a^*_p\hat a_q(n-\cN_+)
+\frac{1}{2}\sum_{p,q\in\L^*_{1,+}}(nV_{\ell,pq00}\hat b^*_p\hat b^*_q+\hc)\\
\end{split} 
\end{equation}
We consider now
\begin{equation} 
\begin{split} 
\cG_{n,\ell}^{(2)}= e^{-B}\cL_{n,\ell}^{(2)}e^B&=e^{-B}\cK e^B+e^{-B}\cL_{n,\ell}^{(2,V)}e^B
\end{split} 
\end{equation}
To prove Proposition \ref{prop:G2K} and Proposition \ref{prop:G2V}  below, we will use the bounds contained in the following lemma, taken from \cite[Lemma 3.6]{BCS}. 
\begin{lemma}\label{lm:L2hat}
Let $V\in L^1(\mathbb{R}^3)$, $V\geq 0$.
Let $j_1, j_2 \in L^2 (\L_1 \times \L_1)$. Consider the operators 
\begin{equation}
\begin{split}
 A_{1} &= \int dx dy \, \k\ell^3 V(\ell (x-y))  a^\sharp (j_{1,x}) a^\sharp (j_{2,y}) \\
A_{2} &= \int dx dy \, \k\ell^3 V(\ell (x-y))  a^\sharp (j_{1,x}) a_y  \\
\end{split}
\end{equation}
where $a^\sharp$ indicates either $a$ or $a^*$.
Then, for every $\xi \in \cF_+^{\leq n}$, we have
\begin{equation}\label{eq:A12} 
\begin{split} 
| \langle \xi, A_{1} \xi \rangle | &\leq  \kappa \|V\|_1 \| j_1 \|_2 \| j_2 \|_2  \| (\cN+1)^{1/2} \xi \|^2 \\
| \langle \xi, A_{2} \xi \rangle | &\leq   \kappa \|V\|_1  \| j_1 \|_2  \| (\cN+1)^{1/2} \xi \|^2 \\
\end{split}  
\end{equation}
\end{lemma}

\begin{proof}
By Cauchy-Schwarz, we have
\begin{equation}\label{eq:A1} 
\begin{split} 
| &\langle \xi, A_{1} \xi \rangle |
= \int dx dy \,\k\ell^3 V(\ell (x-y))  \|a^\sharp (j_{1,x})\xi\|\| a^\sharp (j_{2,y})\xi\| \\
&\leq \int dx dy \,\k\ell^3 V(\ell (x-y)) \|j_{1,x}\|\|j_{2,y}\|\| (\cN_++1)^{1/2}\xi\|^2 \\
&\leq \| (\cN_++1)^{1/2}\xi\|^2\left[\int dx dy \,\k\ell^3 V(\ell (x-y)) \|j_{1,x}\|^2\right]^{1/2}\left[\int dx dy \,\k\ell^3V(\ell (x-y))\|j_{2,y}\|^2\right]^{1/2} \\
&\leq  \k \,\|V\|_1 \| j_1 \|_2 \| j_2 \|_2  \| (\cN+1)^{1/2} \xi \|^2.
\end{split}  
\end{equation}
Similarly we obtain the second estimate in \eqref{eq:A12}.

\end{proof}

\subsubsection{Analysis of $e^{-B}\cK e^B$}
We define the error operator $\cE_{n,\ell}^{(2,K)}$ through
\begin{equation}\label{eq:cEn2K} 
\begin{split} 
e^{-B}\cK e^B&=\,  \cK -\frac{1}{2}\langle\eta,(\Delta_1+\Delta_2)\eta\rangle-\frac{1}{2}\sum_{p,r\in\L^*_{1,+}}\langle\ph_p\otimes\ph_r,\big(\Delta_1+\Delta_2\big)\eta\rangle \hat b_p^*\hat b^*_r\\
&\quad-\frac{1}{2}\sum_{p,r\in\L^*_{1,+}}\langle\big(\Delta_1+\Delta_2\big)\eta,\ph_p\otimes\ph_r\rangle \hat b_p\hat b_r+\cE_{n,\ell}^{(2,K)}.
\end{split} 
\end{equation}
We estimate it in the following proposition.

\begin{prop}\label{prop:G2K}
Let $\cE_{n,\ell}^{(2,K)}$ be defined as in \eqref{eq:cEn2K}. Then, under the same assumptions as Proposition \ref{prop:G}, for every $\delta>0$ there exists a $C>0$  such that
\begin{equation}\label{eq:errL2K}
\begin{split} 
\pm\cE_{n,\ell}^{(2,K)}\leq\delta\cV_\ell+C\k n\ell^{-1}(\cN_++1)
\end{split} 
\end{equation}
as operator inequalities on $\cF^{\leq N}_+$.
\end{prop}

\begin{proof}
We write
\begin{equation}\label{eq:cK} 
\begin{split} 
e^{-B}\cK e^{B}=&\,\cK+\int_0^1ds\, e^{-sB}[\cK,B]e^{sB}\\
\end{split} 
\end{equation}
With definition \eqref{eq:Bsum} we have
\begin{equation}\label{eq:cKB}
\begin{split} 
[\cK,B]&=\frac{1}{2}\sum_{p,q,r\in\L^*_{1,+}}r^2\langle\ph_p\otimes\ph_q,\eta\rangle[\hat a^*_r\hat a_r, \hat b^*_p\hat b^*_q]-\frac{1}{2}\sum_{p,q,r\in\L^*_{1,+}}r^2\langle\eta,\ph_p\otimes\ph_q\rangle[\hat a^*_r\hat a_r, \hat b_p\hat b_q]
\end{split}
\end{equation}
We use now 
\begin{equation}\nonumber
\begin{split} 
[\hat a_r^* \hat a_r,\hat b^*_p\hat b^*_q]&=\hat b_p^*[\hat a_r^* \hat a_r,b^*_q]+[\hat a_r^* \hat a_r,\hat b^*_p]\hat b^*_q=\delta_{rq}\hat b_p^*\hat b^*_r+\delta_{rp}\hat b^*_r\hat b^*_q
\end{split} 
\end{equation}
to obtain
\begin{equation}\label{eq:KBcomm}
\begin{split} 
[\cK,B]&=-\frac{1}{2}\sum_{p,r\in\L^*_{1,+}}\langle\ph_p\otimes\ph_r,\big(\Delta_1+\Delta_2\big)\eta\rangle \hat b_p^*\hat b^*_r+\hc\\
&=-\frac{1}{2}\int dxdy\,\big[\big(\Delta_y+\Delta_y\big) \eta(x,y)\big]b_xb_y+\hc\\
\end{split} 
\end{equation}
With relations \eqref{eq:d}, we decompose
\begin{equation}\nonumber
\begin{split} 
 \int_0^1ds\,e^{-sB}[\cK,B]e^{sB} &=-\frac{1}{2}\int_0^1ds\,\int dxdy\,\big[(\Delta_x+\Delta_y)\eta(x,y)\big]e^{-sB} b_ye^{sB}e^{-sB}b_xe^{sB}+\hc\\
  &=\Big(\text{E}_1+\text{E}_2+\text{E}_3\Big)+\hc
\end{split} 
\end{equation}
with
\begin{equation}\label{eq:E1E2E3}
\begin{split} 
\text{E}_1  &=-\frac{1}{2}\int_0^1ds\,\int dxdy\,\big[(\Delta_x+\Delta_y)\eta(x,y)\big]\big(b(\g_{y}^{(s)})+b^*(\s_{y}^{(s)})\big)\big(b(\g_{x}^{(s)})+b^*(\s_{x}^{(s)})\big)\\
  \text{E}_2 &=-\frac{1}{2}\int_0^1ds\,\int dxdy\,\big[(\Delta_x+\Delta_y)\eta(x,y)]\big(b(\g_{y}^{(s)})+b^*(\s_{y}^{(s)})\big)d_x^{(s)}\\
  &\quad-\frac{1}{2}\int dxdy\,\big[(\Delta_x+\Delta_y)\eta(x,y)]d_y^{(s)}\big(b(\g_{x}^{(s)})+b^*(\s_{x}^{(s)})\big)\\
  \text{E}_3&=-\frac{1}{2}\int_0^1ds\,\int dxdy\,\big[(\Delta_x+\Delta_y)\eta(x,y)]d_y^{(s)}d_x^{(s)}\\
\end{split} 
\end{equation}
where $\g_{x}^{(s)}=\cosh(s\eta_x)$, $\s_{x}^{(s)}=\sinh(s\eta_x)$ (recall the notation $\eta_x(y)=\eta(y,x)$) and $d_y^{(s)}$ is defined as in \eqref{eq:d} with $\eta$ in $B$, $\gamma$ and $\sigma$ substituted by $s\eta$.
We expand $\text{E}_1$ as 
\begin{equation}\nonumber
\begin{split} 
\text{E}_1  &=-\frac{1}{2}\int_0^1ds\,\int dxdy\,\big[(\Delta_x+\Delta_y)\eta(x,y)]\big(b(\g_{y}^{(s)})b(\g_{x}^{(s)})+b^*(\s_{y}^{(s)})b(\g_{x}^{(s)})\\
&\quad+b(\g_{y}^{(s)})b^*(\s_{x}^{(s)})+b^*(\s_{y}^{(s)}))b^*(\s_{x}^{(s)})\big)\\
&=\text{E}_{11}+\text{E}_{12}+\text{E}_{13}+\text{E}_{14}
\end{split} 
\end{equation}
Writing $\gamma^{(s)}=1+p^{(s)}$ we express $\text{E}_{11}$ as 
\begin{equation}\label{eq:E11}
\begin{split} 
\text{E}_{11}
&=-\frac{1}{2}\int_0^1ds\,\int dxdy\,\big[(\Delta_x+\Delta_y)\eta(x,y)]b_yb_x+\tilde{\text{E}}_{11}
\end{split} 
\end{equation}
with
\begin{equation}\nonumber
\begin{split} 
\tilde{\text{E}}_{11}
&=-\frac{1}{2}\int_0^1ds\,\int dxdy\,\big[(\Delta_x+\Delta_y)\eta(x,y)]
\big(b_yb(p_{x}^{(s)})+b(p_{y}^{(s)})b_x+b(p_{y}^{(s)})b(p_{x}^{(s)})\big)\\
\end{split} 
\end{equation}
The first term in \eqref{eq:E11} contributes to \eqref{eq:cEn2K}; with equation \eqref{eq:deltaeta} we write $\tilde{\text{E}}_{11}$ as
\begin{equation}\nonumber
\begin{split} 
\tilde{\text{E}}_{11}
&=\frac{n}{2}\int_0^1ds\,\int dxdy\,\big[\ell^5\lambda_\ell -\k\ell^5V(\ell(x-y))\big]f_\ell(x,y)
\big(b_yb(p_{x}^{(s)})+b(p_{y}^{(s)})b_x+b(p_{y}^{(s)})b(p_{x}^{(s)})\big)
\end{split} 
\end{equation}
(Here we also used that the last two terms in \eqref{eq:deltaeta} are zero when projected onto the orthogonal subspace to $\ph_0$.)
To estimate $\tilde{\text{E}}_{11}$ we bound $f_\ell$ using \eqref{eq:sup} and Cauchy-Schwarz; for the term proportional to $\l_\ell$ we use \eqref{eq:slength}, Cauchy-Schwarz and \eqref{eq:normsp}, while for the term proportional to $V$ we use Lemma \ref{lm:L2hat} and  \eqref{eq:normsp}. This leads to
\begin{equation}\label{eq:boundE11t}
|\langle\xi, \tilde{\text{E}}_{11}\xi\rangle|\leq C\k n\ell^{-1}\|(\cN_++1)^{1/2}\xi\|^2
\end{equation}
Similarly we estimate  $\text{E}_{12}$ and  $\text{E}_{14}$, with the result that 
\begin{equation}\label{eq:boundE12t}
|\langle\xi, \tilde{\text{E}}_{12}\xi\rangle|, |\langle\xi, \tilde{\text{E}}_{14}\xi\rangle|\leq C \kappa n\ell^{-1}\|(\cN_++1)^{1/2}\xi\|^2.
\end{equation}
We consider now $\text{E}_{13}$. Here we cannot use Lemma \ref{lm:L2hat} directly (since the $L^2$ norm of $\gamma^{(s)}$ is not finite); in fact, this is not an error term and we will extract from it an important contribution to \eqref{eq:cEn2K}. We write $b(\g_{y}^{(s)})b^*(\s_{x}^{(s)})=b_yb^*(\s_{x}^{(s)})+b(p_{y}^{(s)})b^*(\s_{x}^{(s)})$ and we put the product $b_yb^*(\s_{x}^{(s)})$ in normal order. Splitting $\s^{(s)}=\eta^{(s)}+r^{(s)}$ we arrive at
\begin{equation}\label{eq:E13}
\begin{split} 
\text{E}_{13}
 &=-\frac{1}{2}\int_0^1ds\,\int dxdy\,\big[(\Delta_x+\Delta_y)\eta(x,y)\big]b(\g_{y}^{(s)})b^*(\s_{x}^{(s)})\\
&=-\frac{1}{2}\int_0^1ds\,s\int dxdy\,\big[(\Delta_x+\Delta_y)\eta(x,y)\big]
\eta(y,x)+\tilde{\text{E}}_{13}
\end{split} 
\end{equation}
with
\begin{equation}\nonumber
\begin{split} 
\tilde{\text{E}}_{13}
&=-\frac{1}{2}\int_0^1ds\,\int dxdy\,\big[(\Delta_x+\Delta_y)\eta(x,y)\big]
\big(r^{(s)}(y,x)\\
&\quad+b^*(\s_{x}^{(s)})b_y-n^{-1}a^*(\s_{x}^{(s)})a_y-n^{-1}\s^{(s)}(y,x)\cN_++b(p_{y}^{(s)})b^*(\s_{x}^{(s)})\big).
\end{split} 
\end{equation}
The first contribution in \eqref{eq:E13} appears in \eqref{eq:cEn2K} (the integration over $s$ gives an additional factor $1/2$, but we still need to add its hermitian conjugate, which is equal to the term itself), while $\tilde{\text{E}}_{13}$ is now an error term. As we did for $\tilde{\text{E}}_{11}$, in $\tilde{\text{E}}_{13}$ we plug in equation \eqref{eq:deltaeta} and we use estimates \eqref{eq:sup} and \eqref{eq:slength}. For the term proportional to $r^{(s)}(y,x)$ we use \eqref{eq:supr},  Cauchy-Schwarz and \eqref{eq:normEta}, while for the term proportional to $\s^{(s)}(y,x)=s\eta(y,x)+r^{(s)}(y,x)$ we use additionally \eqref{eq:supEta}. For all the other contributions in $\tilde{\text{E}}_{13}$ we use Lemma \ref{lm:L2hat}, \eqref{eq:normsp} and \eqref{eq:normEta}. This way we obtain 
\begin{equation}\label{eq:boundE13r}
|\langle\xi, \tilde{\text{E}}_{13}\xi\rangle|\leq C\kappa n\ell^{-1}\|(\cN_++1)^{1/2}\|^2.
\end{equation}

We consider now $ \text{E}_2 $, which we split in $ \text{E}_2=\text{E}_{21}+\text{E}_{22}$ with
\begin{equation}\label{eq:E2}
\begin{split} 
  \text{E}_{21} &=-\frac{1}{2}\int_0^1ds\,\int dxdy\,\big[(\Delta_x+\Delta_y)\eta(x,y)\big]\big(b(\g_{y}^{(s)})+b^*(\s_{y}^{(s)})\big)d_x^{(s)}\\
   \text{E}_{22}&=-\frac{1}{2}\int_0^1ds\,\int dxdy\,\big[(\Delta_x+\Delta_y)\eta(x,y)\big]d_y^{(s)}\big(b(\g_{x}^{(s)})+b^*(\s_{x}^{(s)})\big)\\
\end{split} 
\end{equation}
We focus on $  \text{E}_{21}$ first. As before, using equation \eqref{eq:deltaeta} and observing that only the first term contributes on the orthogonal subspace to $\ph_0$, we get
\begin{equation}\nonumber
\begin{split} 
  \text{E}_{21} 
  &=\frac{n}{2}\int_0^1ds\,\int dxdy\,\Big[\ell^5\lambda_\ell -\k\ell^5V(\ell(x-y))\Big]f_\ell(x,y)\big(b(\g_{y}^{(s)})+b^*(\s_{y}^{(s)})\big)d_x^{(s)}\\
\end{split} 
\end{equation}
We write it as
\begin{equation}\nonumber
\begin{split} 
  \text{E}_{21} &=\frac{n}{2}\int_0^1ds\,\int dxdy\,\Big[\ell^5\lambda_\ell -\k\ell^5V(\ell(x-y))\Big]f_\ell(x,y)b_yd_x^{(s)}\\
   &+\frac{n}{2}\int_0^1ds\,\int dxdy\,\Big[\ell^5\lambda_\ell -\k\ell^5V(\ell(x-y))\Big]f_\ell(x,y)\big(b(p_{y}^{(s)})+b^*(\s_{y}^{(s)})\big)d_x^{(s)}\\
    &=\text{E}_{211}+\text{E}_{212}
\end{split} 
\end{equation}
To estimate $\text{E}_{212}$ we use \eqref{eq:dxy-bds} (with $j=0$ and the factors $s$ bounded by 1, together with the bound $\cN_+^2\leq n^2$) and Proposition \ref{prop:eta} to obtain
\begin{equation}\nonumber
\begin{split} 
 |\langle\xi, \text{E}_{212} \xi\rangle|&\leq Cn\int dxdy\,\Big|\ell^5\lambda_\ell f_\ell( x,y)-\k\ell^5V(\ell(x-y))f_\ell(x,y)\Big|\\
 &\hspace{2cm} \times \int_0^1ds\|\big(b^*(p_{y}^{(s)})+b(\s_{y}^{(s)})\big)\xi\|  \Big[ \, \|\eta_x\| \| (\cN_+ + 1)^{1/2} \xi \| + \| \eta \| \| b_x\xi \| \Big].
\end{split} 
\end{equation}
With the aid of the Cauchy-Schwarz inequality and  estimates \eqref{eq:sup} and \eqref{eq:slength} we obtain
\begin{equation}\nonumber
\begin{split} 
 |\langle\xi, \text{E}_{212} \xi\rangle|&\leq Cn\left[\int dxdy\,\Big|\ell^{-1}-\k\ell^2V(\ell(x-y))\Big|\int_0^1ds\big(\|p_y^{(s)}\|^2+\|\s_y^{(s)}\|^2\big)\|(\cN_++1)^{1/2}\xi\|^2 \right]^{1/2}\\
 &\hspace{0.3cm} \times\left[\int dxdy\,\Big|\ell^{-1}-\k\ell^2V(\ell(x-y))\Big| \Big( \, \|\eta_x\|^2 \| (\cN_+ + 1)^{1/2} \xi \|^2 + \| \eta \|^2 \| b_x\xi \|^2 \Big)\right]^{1/2}.
\end{split} 
\end{equation}
Using estimates \eqref{eq:supL2} and \eqref{eq:normEta} to bound the norm of $p,\,\s$ and $\eta$ we get
\begin{equation}\nonumber
\begin{split} 
 |\langle\xi, \text{E}_{212} \xi\rangle|
 &\leq C\kappa n\ell^{-1}\|(\cN_++1)^{1/2}\xi\|^2.
\end{split} 
\end{equation}
To estimate $\text{E}_{211}$ we use the second bound in Lemma \ref{lm:dx}
\begin{equation}\nonumber
\begin{split} 
 |\langle\xi, \text{E}_{211} \xi\rangle| &\leq\frac{n}{2}\|(\cN_++1)^{1/2}\xi\|\int_0^1ds\,\int dxdy\,\Big|\ell^5\lambda_\ell f_\ell( x,y)+\k\ell^5V(\ell(x-y))f_\ell(x,y)\Big|\\
 &\hspace{1.7 cm}  \|(\cN_++1)^{-1/2}b_yd_x^{(s)}\xi\|\\
 &\leq C\|(\cN_++1)^{1/2}\xi\|\int dxdy\,\Big|\ell^5\lambda_\ell f_\ell( x,y)-\k\ell^5V(\ell(x-y))f_\ell(x,y)\Big|\\
 & \hspace{1.7 cm}  \Big[ \, \|\eta_x\| \|\eta_y\|  \| (\cN_+ + 1)^{1/2} \xi \|  + |\eta(y,x)| \| (\cN_+ +1)^{1/2}  \xi \| \\
& \hspace{1.7 cm} + \| \eta_y \| \| b_x  \xi \| +  \|\eta_x\| \|a_y (\cN_+ + 1) \xi \| + \| \eta \| \|a_y a_x (\cN_+ +1)^{1/2}  \xi \|   \, \Big]\\
&\leq  Cn\ell^{-1}\|(\cN_++1)^{1/2}\xi\|^2\\&\quad+ C\|(\cN_++1)^{1/2}\xi\|\Big[\int dxdy\,\k\ell^2V(\ell(x-y)) \|a_y a_x (\cN_+ +1)^{1/2}  \xi \|^2\Big]^{1/2}.
\end{split} 
\end{equation}
In the last step we used \eqref{eq:sup}, Cauchy-Schwarz (similarly as above) and additionally \eqref{eq:supEta} for the term containing $|\eta(x,y)|$.
With \eqref{eq:normEta} we conclude that 
\begin{equation}\nonumber
\begin{split} 
 |\langle\xi, \text{E}_{211} \xi\rangle| 
&\leq  C\kappa n\ell^{-1}\|(\cN_++1)^{1/2}\xi\|^2+ C\kappa^{1/2}n^{1/2}\ell^{-1/2}\|(\cN_++1)^{1/2}\xi\|\|\cV_\ell^{1/2}\xi\|.
\end{split} 
\end{equation}
Therefore
\begin{equation}\label{eq:boundE21}
\begin{split} 
 |\langle\xi, \text{E}_{21} \xi\rangle| 
&\leq  C\kappa n\ell^{-1}\|(\cN_++1)^{1/2}\xi\|^2+ C\kappa^{1/2}n^{1/2}\ell^{-1/2}\|(\cN_++1)^{1/2}\xi\|\|\cV_\ell^{1/2}\xi\|
\end{split} 
\end{equation}
The second term in \eqref{eq:E2} can be estimated as follows
\begin{equation}
\begin{split} 
    |\langle\xi, &\text{E}_{22} \xi\rangle| \leq\frac{1}{2}\|(\cN_++1)^{1/2}\xi\|\int_0^1ds\,\int dxdy\,|(\Delta_x+\Delta_y)\eta(x,y)|\\
    &\hspace{2cm} \times \|(\cN_++1)^{-1/2}d_y^{(s)}\big(b(\g_{x}^{(s)})+b^*(\s_{x}^{(s)})\big)\xi\|\\
     &\leq Cn^{-1}\int_0^1ds\,\int dxdy\, |(\Delta_x+\Delta_y)\eta(x,y)|\|(\cN_++1)^{1/2}\xi\|\\
    & \quad \times \Big[ \, \|\eta_y\| \| \big(b(\g_{x}^{(s)})+b^*(\s_{x}^{(s)})\big)(\cN_+ + 1) \xi \| + \| \eta \| \| b_y \big(b(\g_{x}^{(s)})+b^*(\s_{x}^{(s)})\big)(\cN_+ + 1) ^{1/2}\xi \| \Big]\\
\end{split} 
\end{equation}
Substituting \eqref{eq:deltaeta} for $(\Delta_x+\Delta_y)\eta(x,y)$ and arguing as before we obtain 
\begin{equation}\label{eq:boundE22}
\begin{split} 
 |\langle\xi, \text{E}_{22} \xi\rangle| 
&\leq  C\kappa n\ell^{-1}\|(\cN_++1)^{1/2}\xi\|^2+ C\kappa^{1/2}n^{1/2}\ell^{-1/2}\|(\cN_++1)^{1/2}\xi\|\|\cV_\ell^{1/2}\xi\|
\end{split} 
\end{equation}

Finally we examine $\text{E}_3$ in \eqref{eq:E1E2E3}; with the third estimate in Lemma \ref{lm:dx} we have
\begin{equation}
\begin{split} 
 |&\langle\xi, \text{E}_{3} \xi\rangle| \leq C\|(\cN_++1)^{1/2}\xi\|\int_0^1ds\,\int dxdy\,| (\Delta_x+\Delta_y)\eta(x,y)|\|(\cN_++1)^{-1/2}d_y^{(s)}d_x^{(s)}\xi\|\\
 &\leq Cn^{-1}\|(\cN_++1)^{1/2}\xi\|\int dxdy\,|(\Delta_x+\Delta_y)\eta(x,y)| \Big [ \| \eta_x \|\| \eta_y \| \|(\cN+1)^{3/2}  \xi \|\\
&\quad+\| \eta_x \| \|b_y(\cN+1) \xi \| + |\eta(y,x)| \| (\cN +1)^{1/2}  \xi \| +  \|\eta_y\| \|a_x (\cN_+ + 1) \xi \| + \|a_y a_x (\cN +1)^{1/2}  \xi \|   \, \Big] \\
\end{split} 
\end{equation}
leading to 
\begin{equation}\label{eq:boundE3}
\begin{split} 
 |\langle\xi, \text{E}_{3} \xi\rangle| 
&\leq  C\kappa n\ell^{-1}\|(\cN_++1)^{1/2}\xi\|^2+ C\kappa^{1/2}n^{1/2}\ell^{-1/2}\|(\cN_++1)^{1/2}\xi\|\|\cV_\ell^{1/2}\xi\|
\end{split} 
\end{equation}
The estimates \eqref{eq:boundE11t}, \eqref{eq:boundE12t}, \eqref{eq:boundE13r},  \eqref{eq:boundE21}, \eqref{eq:boundE22} and \eqref{eq:boundE3} prove \eqref{eq:errL2K}.
\end{proof}

\subsubsection{Analysis of $e^{-B}\cL_{n,\ell}^{(2,V)}e^B$}

With $\cL_{n,\ell}^{(2,V)}$ introduced in \eqref{eq:QV}, we define the error operator $\cE_{n,\ell}^{(2,V)}$ through
\begin{equation}\label{eq:cEn2V} 
\begin{split} 
e^{-B}\cL_{n,\ell}^{(2,V)}e^B&=\, n\sum_{p,q\in\L^*_{1,+}}V_{\ell,pq00}\langle\eta,\ph_p\otimes\ph_q\rangle+\frac{1}{2}\sum_{p,q\in\L^*_{1,+}}(nV_{\ell,pq00}\hat b^*_p\hat b^*_q+\hc)+\cE_{n,\ell}^{(2,V)}
\end{split} 
\end{equation}
Proposition \ref{prop:G2V} provides an estimate for $\cE_{n,\ell}^{(2,V)}$.

\begin{prop}\label{prop:G2V} Let $\cE_{n,\ell}^{(2,V)}$ be defined as in \eqref{eq:cEn2V}. Then, under the same assumptions as Proposition \ref{prop:G}, for every $\delta>0$ there exists a $C>0$ such that
\begin{equation}\label{eq:errL2V}
 \pm\cE_{n,\ell}^{(2,V)}\leq\delta\cV_\ell+ C\k n\ell^{-1}(\cN_++1)
\end{equation}
as operator inequalities on $\cF^{\leq N}_+$.
\end{prop}

\begin{proof}[Proof of Proposition \ref{prop:G2V}]
We split $e^{-B}\cL_{n,\ell}^{(2,V)}e^{B}$ as 
 \begin{equation}\nonumber
 \begin{split}
e^{-B}\cL_{n,\ell}^{(2,V)}e^{B}&=\text{F}_1+\text{F}_2+\text{F}_3
\end{split}
\end{equation}
with
 \begin{equation}\label{eq:QV2}
 \begin{split}
\text{F}_1&=\sum_{p,q\in\L^*_{1,+}}V_{\ell,0p0q}\,e^{-B}\hat a^*_p\hat a_q(n-\cN_+)e^{B}\\
\text{F}_2&=\sum_{p,q\in\L^*_{1,+}}V_{\ell,0pq0}\,e^{-B} \hat a^*_p\hat a_q(n-\cN_+)e^{B}\\
\text{F}_3&=\frac{1}{2}\sum_{p,q\in\L^*_{1,+}}(nV_{\ell,pq00}\,e^{-B} \hat b^*_p\hat b^*_qe^{B}+\hc)\\
\end{split}
\end{equation}
It is convenient to rewrite
 \begin{equation}\nonumber
 \begin{split}
\text{F}_1
&=\kappa n\int dx dy \,\ell^2 V(\ell(x-y))\,e^{-B} \Big(b^*_yb_y -\frac{1}{n}a^*_ya_y\Big)e^{B}\\
\end{split}
\end{equation}
The expectation of $\text{F}_1$ on any $\xi\in\cF_+^{\leq n}$ can be estimated as
 \begin{equation}\label{eq:F1}
 \begin{split}
|\langle\xi,\text{F}_1\xi\rangle|&\leq \kappa n\int  dy \,|\langle\xi,e^{-B} \Big(b^*_yb_y -n^{-1}a^*_ya_y\Big)e^{B}\xi\rangle|\int dx\,\ell^2 V(\ell(x-y))\,
\leq C\kappa n\ell^{-1}\langle\xi, \cN_+\xi\rangle\\
\end{split}
\end{equation}
where we used Lemma \ref{lm:TNT}. Similarly we have
\begin{equation}\label{eq:F2}
 \begin{split}
|\langle\xi,\text{F}_2\xi\rangle|
&\leq \kappa n\int dx dy \,\ell^2 V(\ell(x-y))\,|\langle\xi,e^{-B} \Big(b^*_yb_x -\frac{1}{n}a^*_ya_x\Big)e^{B}\xi\rangle|
\leq C\kappa n\ell^{-1}\langle\xi, \cN_+\xi\rangle
\end{split}
\end{equation}
We focus now on on the last contribution in \eqref{eq:QV2}.
\begin{equation}
\begin{split}
\text{F}_3
&=\frac{\kappa n}{2}\int dxdy\,\ell^2V(\ell(x-y))\big(e^{-B}b_xb_ye^{B}+\hc\big)\\
\end{split}
\end{equation}
Using equations \eqref{eq:d}, we get
\begin{equation}\label{eq:splitF3}
\begin{split}
\text{F}_3
&=\frac{\kappa n}{2}\int dxdy\,\ell^2V(\ell(x-y))\big[(b(\g_x)+b^*(\s_x)+d_x)(b(\g_y)+b^*(\s_y)+d_y)+\hc\big]\\
&=\frac{\kappa n}{2}\int dxdy\,\ell^2V(\ell(x-y))\big[(b(\g_x)+b^*(\s_x))(b(\g_y)+b^*(\s_y))+\hc\big]\\
&\quad+\frac{\kappa n}{2}\int dxdy\,\ell^2V(\ell(x-y))\big[(b(\g_x)+b^*(\s_x))d_y+d_x(b(\g_y)+b^*(\s_y))+\hc\big]\\
&\quad+\frac{ \kappa n}{2}\int dxdy\,\ell^2V(\ell(x-y))\big[d_xd_y+\hc\big]\\
&=\text{F}_{31}+\text{F}_{32}+\text{F}_{33}
\end{split}
\end{equation}
We start analyzing $\text{F}_{31}$. After normal ordering, a simple calculation (similar to the one done in \eqref{eq:E13}) gives
\begin{equation}\label{eq:mainbb}
\begin{split}
\text{F}_{31}
&=\kappa n\int dxdy\,\ell^2V(\ell(x-y))\eta(x,y)+\frac{\kappa n}{2}\int dxdy\,\ell^2V(\ell(x-y))\big[b_xb_y+\hc\big]+\cE_{31}\\
\end{split}
\end{equation}
with
\begin{equation}\label{eq:mainbberr}
\begin{split}
\cE_{31}
&=\frac{\kappa n}{2}\int dxdy\,\ell^2V(\ell(x-y))\big[r(x,y)-\eta(x,y)n^{-1}\cN_+\\
&\hspace{1.cm}+b^*(r_y)b_x+b^*(\eta_y)b_x+n^{-1}a^*(\eta_y)a_x+b(p_x)b^*(\eta_y)+b(p_x)b^*(r_y)\\
&\hspace{1.cm}+b(p_x)b_y+b(\g_x)b(p_y)+b^*(\s_x)b(\g_y)+b^*(\s_x)b^*(\s_y)+\hc\big]\\
\end{split}
\end{equation}
where again we used the notation $\g_x=\delta_x+p_x$ and $\s_x=\eta_x+r_x$. Lemma \ref{lm:L2hat} and Proposition \ref{prop:eta}
show that $\cE_{31}$ satisfies
\begin{equation}\label{eq:calE31}
 |\langle\xi,\cE_{31}\rangle\xi|\leq C\kappa n\ell^{-1}\|(\cN_++1)^{1/2}\xi\|^2
\end{equation}
 
 Next we consider  $\text{F}_{32}$. Again splitting  $\g=1+p$ and $\s=\eta+r$, we write
\begin{equation}
 \text{F}_{32}=:\text{F}_{321}+\text{F}_{322}+\text{F}_{323}+\hc
\end{equation}
with
\begin{equation}
\begin{split}
\text{F}_{321}&=\frac{\kappa n}{2}\int dxdy\,\ell^2V(\ell(x-y))b_xd_y\\
\text{F}_{322}&=\frac{\kappa n}{2}\int dxdy\,\ell^2V(\ell(x-y))d_xb_y\\
\text{F}_{323}&=\frac{\kappa n}{2}\int dxdy\,\ell^2V(\ell(x-y))\big[(b(p_x)+b^*(\s_x))d_y+d_x(b(p_y)+b^*(\s_y))\big]\\
\end{split}
\end{equation}
To bound $\text{F}_{321}$, we use \eqref{eq:dxy-bdsN2} and Proposition \ref{prop:eta}:
\begin{equation}\label{eq:F321}
\begin{split}
|\langle\xi&,\text{F}_{321}\rangle\xi|\leq C\kappa\|(\cN_++1)^{1/2}\xi\|\int dxdy\,\ell^2V(\ell(x-y)) \Big[ \, \|\eta_x\| \|\eta_y\|  \| (\cN_+ + 1)^{1/2} \xi \| \\ 
&\quad + |\eta(y,x)| \| (\cN +1)^{1/2}  \xi \|  + \| \eta_y \| \| b_x  \xi \| +  \|\eta_x\| \|a_y (\cN_+ + 1) \xi \| \Big]\\
& \quad+C\kappa n^{1/2}\|\eta\|_2\|(\cN_++1)^{1/2}\xi\|\int dxdy\,\ell^2V(\ell(x-y))\|a_y b_x   \xi \|  \\
&\leq C\kappa n\ell^{-1}\|\eta\|_2\|(\cN_++1)^{1/2}\xi\|^2+C\kappa^{1/2}n^{1/2}\ell^{-1/2}\|\eta\|_2\|(\cN_++1)^{1/2}\xi\|\|\cV_\ell^{1/2}\xi\|\\
\end{split}
\end{equation}
The estimate for $\text{F}_{322}$ follows from \eqref{eq:dxy-bds}:
\begin{equation}\label{eq:F322}
\begin{split}
|\langle\xi,&\text{F}_{322}\xi\rangle|\leq \kappa n\|(\cN_++1)^{1/2}\xi\|\int dxdy\,\ell^2V(\ell(x-y))\|(\cN_++1)^{-1/2}d_xb_y\xi\|\\
&\leq C\kappa n\|(\cN_++1)^{1/2}\xi\|\int dxdy\,\ell^2V(\ell(x-y)) \|\eta_x\| \| b_y\xi \|\\
&\quad+C\kappa n^{1/2}\| \eta \|\|(\cN_++1)^{1/2}\xi\|\int dxdy\,\ell^2V(\ell(x-y))  \| b_x b_y\xi \| \\
&\leq C\kappa n\ell^{-1}\| \eta \|_2\|(\cN_++1)^{1/2}\xi\|^2+C\kappa^{1/2} n^{1/2}\ell^{1/2}\| \eta \|_2\|(\cN_++1)^{1/2}\xi\|  \| \cV_\ell^{1/2}\xi \|
\end{split}
\end{equation}
and similarly the estimate for $\text{F}_{323}$:
\begin{equation}
\begin{split}
|\langle\xi,&\text{F}_{323}\xi\rangle|=\frac{\kappa n}{2}\int dxdy\,\ell^2V(\ell(x-y))\big[(b(p_x)+b^*(\s_x))d_y+d_x(b(p_y)+b^*(\s_y))\big]\\
&\leq \kappa n\int dxdy\,\ell^2V(\ell(x-y))\|(b^*(p_x)+b(\s_x))\xi\|\|d_y\xi\|\\
&\quad+ \kappa n\,\|(\cN_++1)^{1/2}\xi\|\int dxdy\,\ell^2V(\ell(x-y))\|(\cN_++1)^{-1/2}d_x(b(p_y)+b^*(\s_y))\xi\|\\
&\leq C\kappa n^{1/2}\|(\cN_++1)^{1/2}\xi\|\int dxdy\,\ell^2V(\ell(x-y))\big(\|p_x\|_2+\|\s_x\|_2\big) \\
&\hspace{7cm}\times\Big[ \, \|\eta_y\| \| (\cN_+ + 1)^{1/2} \xi \| + \| \eta \| \| b_y\xi \| \Big] \\
&\quad+ C\kappa\,\|(\cN_++1)^{1/2}\xi\|\int dxdy\,\ell^2V(\ell(x-y))\Big[ \, \|\eta_x\| \| (\cN_+ + 1)(b(p_y)+b^*(\s_y))\xi \| \\
&\quad\hspace{6cm}+ \| \eta \| \| b_x (\cN_+ + 1) ^{1/2}(b(p_y)+b^*(\s_y))\xi \| \Big]\\
\end{split}
\end{equation}
We normal order the last term
and use estimates \eqref{eq:normEta}, \eqref{eq:supEta} and \eqref{eq:supr} to get
\begin{equation}\label{eq:F323}
\begin{split}
|\langle\xi,&\text{F}_{323}\xi\rangle|\leq C\kappa n\ell^{-1}\|(\cN_++1)^{1/2}\xi\|^2\\
\end{split}
\end{equation}

Finally we consider  the last contribution in \eqref{eq:splitF3}. With  estimate  \eqref{eq:dxy-bdsN3} 
and Proposition \ref{prop:eta} we conclude that
\begin{equation}\label{eq:F33}
\begin{split}
|\langle&\xi,\text{F}_{33}\xi\rangle|\leq C\kappa n\ell^{-1}\|\eta\|^2\|(\cN_++1)^{1/2}\xi\|^2\\
&\quad+C\kappa n^{1/2}\ell^{-1/2}\| \eta \|^2\|(\cN_++1)^{1/2}\xi\|\Big[\int dxdy\,\ell^2V(\ell(x-y)) \|a_y a_x \xi \|^2\Big]^{1/2}\\
&\leq C\kappa n\ell^{-1}\|\eta\|^2\|(\cN_++1)^{1/2}\xi\|^2+C\kappa^{1/2} n^{1/2}\ell^{-1/2}\| \eta \|^2\|(\cN_++1)^{1/2}\xi\|\|\cV_\ell^{1/2}\xi\|\\
\end{split}
\end{equation}
Estimates \eqref{eq:F1}, \eqref{eq:F2}, \eqref{eq:calE31}, \eqref{eq:F321}, \eqref{eq:F322}, \eqref{eq:F323} and \eqref{eq:F33} prove \eqref{eq:errL2V}. 
\end{proof}

\subsection{Analysis of $\cG_{n,\ell}^{(3)}$}\label{g3}
As defined in \eqref{eq:cLp},
\begin{equation}\nonumber
\begin{split} 
 \cL_{n,\ell}^{(3)}&=\sum_{p,q,r\in\L^*_{1,+}}(n^{1/2}V_{\ell,qr0p}\hat b^*_r\hat a^*_q\hat a_p+\hc)\\
 &=n^{1/2}\int dx dy\, \k\ell^2V(\ell(x-y))\Big[b^*_xa^*_ya_x+\hc\Big]\\
 \end{split} 
\end{equation}
 We define $\cE_{n,\ell}^{(3)}$ through
\begin{equation}\label{eq:cGn3main}
\begin{split} 
\cG_{n,\ell}^{(3)}=e^{-B}\cL_{n,\ell}^{(3)}e^B&=n^{1/2}\int dx dy\,\k\ell^2V(\ell(x-y))\eta(y,x)  \big[b(\gamma_x)+b^*(\sigma_x) +\hc\big] 
+\cE_{n,\ell}^{(3)}.\\
 \end{split} 
\end{equation}

\begin{prop}\label{prop:G3} Let $\cE_{n,\ell}^{(3)}$ be defined as in \eqref{eq:cGn3main}. Then, under the same assumptions as Proposition \ref{prop:G}, for any $\delta>0$ there exists a constant $C > 0$ such that
\begin{equation}\label{eq:errL3}
\begin{split} 
 \pm\,\cE_{n,\ell}^{(3)}\leq \delta\cV_\ell+ C\k n\ell^{-1}(\cN_++1)
\end{split} 
\end{equation}
as operator inequalities on $\cF^{\leq n}_+$.
\end{prop}

To prove Proposition \ref{prop:G3}, we need the following lemma,  taken from \cite[Lemma 3.8]{BCS}.
\begin{lemma}\label{lm:LN4}
Let $V\in L^1(\mathbb{R}^3)$, $V\geq 0$.
Let $j_1, j_2 \in L^2 (\L_1 \times \L_1)$ with 
\[ M_i := \max \left\{ \sup_x \int dy |j_i (x,y)|^2 , \sup_y \int dx |j_i (x,y)|^2 \right\} < \infty \]
for $i=1,2$. Then we have
\[ \begin{split} 
\int dx dy  \, \k\ell^3 V(\ell (x-y))& \| a^\sharp (j_{1,x}) a^\sharp (j_{2,y}) \xi \|^2 \leq C\k \min (M_1 \|j_2 \|^2_2 , M_2 \| j_1 \|_2^2 ) \| (\cN_++1) \xi \|^2 \\
\int dx dy  \, \k\ell^3 V(\ell (x-y))& \| a^\sharp (j_{1,x}) a_y \xi \|^2 \leq C\k M_1 \| (\cN_++1) \xi \|^2 \end{split} \]
for all $\xi \in \cF$ (with $a^\sharp$ we indicate either $a$ or $a^*$).
\end{lemma}

\begin{proof}
The first inequality simply follows from Cauchy-Schwarz
\[ \begin{split} 
\int dx dy  \, \k\ell^3 V(\ell (x-y)) \| a^\sharp (j_{1,x}) a^\sharp (j_{2,y}) \xi \|^2  &\leq \int dx dy \,\k\ell^3 V(\ell (x-y))\|j_{1,x}\|_2^2 \|j_{2,y}\|_2^2\| (\cN_++1)\psi \|^2 \\
&\leq C\k \min (M_1 \|j_2 \|^2_2 , M_2 \| j_1 \|_2^2 ) \| (\cN_++1) \psi \|^2 \\ \end{split} \]
The second inequality can be obtained similarly.
\end{proof}

\begin{proof}[Proof of Proposition \ref{prop:G3}]
We compute
\begin{equation}
\begin{split} 
e^{-B}a^*_ya_xe^B&=a^*_ya_x+\int_0^1 ds\, e^{-sB}[a_y^* a_x,B]e^{sB}=a^*_ya_x+\int_0^1 ds\, e^{-sB}\big(b(\eta_y)b_x+b^*(\eta_x)b^*_y\big)e^{sB}
\end{split} 
\end{equation}
We have therefore
\begin{equation}\label{eq:G3nl}
\begin{split} 
\cG_{n,\ell}^{(3)} &=n^{1/2}\int dx dy \, \k\ell^2 V(\ell(x-y)) \Big[e^{-B}b_x^*e^{B} a_y^* a_x + \hc\Big] \\
&\quad+n^{1/2}\int dx dy \, \k\ell^2 V(\ell(x-y)) \Big[e^{-B}b_x^* e^{B} \int_0^1 ds\, e^{-sB}b^*(\eta_x)b^*_ye^{sB}+ \hc\Big] \\
&\quad+n^{1/2}\int dx dy \, \k\ell^2 V(\ell(x-y)) \Big[e^{-B}b_x^* e^{B} \int_0^1 ds\, e^{-sB}b(\eta_y)b_xe^{sB}+ \hc\Big] \\
&=:\text{G}_1+\text{G}_2+\text{G}_3+\hc
\end{split} 
\end{equation}
We start analyzing $\text{G}_1$.
Using \eqref{eq:d} and
\begin{equation}
 b(\s_x)a_y^* a_x=\int dz\, \s(x,z)b_za_y^* a_x=a_y^* a_xb(\s_x) +\s(x,y)b_x
\end{equation}
we write it as (adopting always the notation $\s=\eta+r$)
\begin{equation}\label{eq:G1linear}
\begin{split} 
\text{G}_1 &=n^{1/2}\int dx dy \,\k \ell^2 V(\ell(x-y)) \big(b^*(\g_x)+b(\s_x)+d^*_x\big)a_y^* a_x \\
&=n^{1/2}\int dx dy \,\k \ell^2 V(\ell(x-y)) \eta(x,y)b_x +\text{G}_{11}+\text{G}_{12}+\text{G}_{13}+\text{G}_{14} 
\\
\end{split} 
\end{equation}
with
\begin{equation}
\begin{split} 
\text{G}_{11}&=n^{1/2}\int dx dy \,\k \ell^2V(\ell(x-y)) r(x,y)b_x \\
{\text{G}_{12}}&=n^{1/2}\int dx dy \,\k \ell^2 V(\ell(x-y)) b^*(\g_x)a_y^* a_x\\
\text{G}_{13}&=n^{1/2}\int dx dy \,\k \ell^2 V(\ell(x-y)) a_y^* a_xb(\s_x)  \\
\text{G}_{14}&=n^{1/2}\int dx dy \,\k \ell^2 V(\ell(x-y)) d^*_xa_y^* a_x\\
\end{split} 
\end{equation}
With \eqref{eq:supr} and Cauchy-Schwarz we see that for any normalized $\xi\in\cF_+^{\leq n}$
\begin{equation}
\begin{split} 
|\langle\xi,&\text{G}_{11}\xi\rangle|\leq Cn^{1/2}\|\eta\|\int dx dy \,\k \ell^2 V(\ell(x-y))\|\eta_x\|\|\eta_y\| \|b_x\xi\| \\
&\leq Cn^{1/2}\|\eta\|\left[\int dx dy \,\k \ell^2 V(\ell(x-y))\|b_x\xi\|^2 \right]^{1/2}\left[\int dx dy \,\k \ell^2 V(\ell(x-y))\|\eta_x\|^2\|\eta_y\|^2 \right]^{1/2}.
\end{split} 
\end{equation}
In the last factor we use \eqref{eq:supEtaL2}
%
and we arrive at
\begin{equation}\label{eq:G11}
\begin{split} 
|\langle\xi,\text{G}_{11}\xi\rangle|
&\leq C\k n^{1/2}\ell^{-1}\|(\cN_++1)^{1/2}\xi\|^2.
\end{split} 
\end{equation}
To bound $\text{G}_{12}$ we split $\gamma=1+p$ and use estimate \eqref{eq:supL2} and \eqref{eq:supEtaL2}, so that
\begin{equation}\label{eq:G12}
\begin{split} 
|\langle\xi,\text{G}_{12}\xi\rangle|&\leq n^{1/2}\int dx dy \,\k \ell^2 V(\ell(x-y)) \|a_yb_x\xi\|\| a_x\xi\|\\
&\quad+ n^{1/2}\int dx dy \,\k \ell^2 V(\ell(x-y)) \|\cN_+^{-1/2}a_yb(p_x)\xi\|\| \cN_+^{1/2}a_x\xi\|\\
&\leq Cn^{1/2}\k^{1/2}\ell^{-1/2}\|(\cN_++1)^{1/2}\xi\|\|\cV_\ell^{1/2}\xi\|+C\k n\ell^{-1}\|(\cN_++1)^{1/2}\xi\|^2.
\end{split} 
\end{equation}
In the last step we used the Cauchy-Schwarz inequality and note that the first term is proportional to $\cV_\ell$, as defined in \eqref{eq:excitHam}.
Similarly
\begin{equation}\label{eq:G13}
\begin{split} 
|\langle\xi,\text{G}_{13}\xi\rangle|&\leq  n^{1/2}\int dx dy \,\k \ell^2 V(\ell(x-y)) \|\cN_+^{-1/2}a_xb(\s_x)\xi\|\| \cN_+^{1/2}a_y\xi\|\\&\leq C\k n\ell^{-1}\|(\cN_++1)^{1/2}\xi\|^2
\end{split} 
\end{equation}
In order to bound $\text{G}_{14}$, we use \eqref{eq:dxy-bdsN2} to estimate
\begin{equation}\label{eq:G14}
\begin{split} 
|\langle&\xi,\text{G}_{14}\xi\rangle|\leq n^{1/2}\int dx dy \,\k \ell^2 V(\ell(x-y)) \|\cN_+^{-1/4}a_yd_x\xi\|\| \cN_+^{1/4}a_x\xi\|\\
&\leq n^{1/2}\k^{1/2}\ell^{-1/2}\|(\cN_++1)^{3/4}\xi\|\left[\int dx dy \,\k \ell^2 V(\ell(x-y)) \|d_x(\cN_++1)^{1/4}\xi\|^{1/2}\right]^{1/2}\\
&\leq Cn^{1/2}\k^{1/2}\ell^{-1/2}\|(\cN_++1)^{3/4}\xi\|\\
&\hspace{1cm}\times \left[\int dx dy \,\k \ell^2 V(\ell(x-y)) \Big(\|\eta_x\|\|(\cN_++1)^{3/4}\xi\|+\|\eta\|\|b_x(\cN_++1)^{1/4}\xi\|\Big)\right]^{1/2}\\
&\leq Cn^{1/2}\k\ell^{-1}\|(\cN_++1)^{3/4}\xi\|^2\\
\end{split} 
\end{equation}

Next we consider  $\text{G}_2$. With \eqref{eq:d}, we have
\begin{equation}
\begin{split} 
\text{G}_2=&n^{1/2}\int dx dy \, \k\ell^2 V(\ell(x-y)) \big(b^*(\g_x)+b(\s_x)+d_x^*\big) \int_0^1 ds\, e^{-sB}b^*_yb^*(\eta_x)e^{sB} \\
\end{split} 
\end{equation}
and we split $\text{G}_2=\text{G}_{21}+\text{G}_{22}+\text{G}_{23}$, where
\begin{equation}\label{eq:G2}
\begin{split} 
\text{G}_{21}=&n^{1/2}\int dx dy \, \k\ell^2 V(\ell(x-y)) b(\s_x) \int_0^1 ds\, e^{-sB}b^*_yb^*(\eta_x)e^{sB} \\
\text{G}_{22}=&n^{1/2}\int dx dy \, \k\ell^2 V(\ell(x-y)) b^*(\g_x) \int_0^1 ds\, e^{-sB}b^*_yb^*(\eta_x)e^{sB} \\
\text{G}_{23}=&n^{1/2}\int dx dy \, \k\ell^2 V(\ell(x-y))d_x^* \int_0^1 ds\, e^{-sB}b^*_yb^*(\eta_x)e^{sB} \\
\end{split} 
\end{equation}
In $\text{G}_{21}$ we expand further
\begin{equation}
\begin{split} 
\text{G}_{21}=&n^{1/2}\int dx dy \, \k\ell^2 V(\ell(x-y)) b(\s_x)\int_0^1 ds\,\big(b^*_y+b^*(p_y^{(s)})+b(\s^{(s)}_y)+(d^{(s)}_y)^*\big) e^{-sB}b^*(\eta_x)e^{sB} \\
\end{split} 
\end{equation}
where we denote again $\s^{(s)}=\sinh(s\eta)$, $p^{(s)}=\cosh(s\eta)-1$ and $(d^{(s)}_y)^*$ is defined as $d_y^*$ in \eqref{eq:d} with $\eta$ in $B$, $\gamma$ and $\sigma$ substituted with $s\eta$. We commute $b_y$ to the left using \eqref{eq:comm-b}, so that
\begin{equation}
\begin{split} 
 b(\s_x)b^*_y&=\int dw \,\s(w,x)b_wb^*_y=\int dw \,\s(w,x)\big(b^*_yb_w+\delta(y-w)(1-\cN_+/n)-n^{-1}a^*_ya_w\big)\\
 &=\s(y,x)(1-\cN_+/n)+b^*_yb(\s_x)-n^{-1}a^*_ya(\s_x)
\end{split} 
\end{equation}
Hence
\begin{equation}\label{eq:G21}
\begin{split} 
\text{G}_{21}=&n^{1/2}\int dx dy \, \k\ell^2 V(\ell(x-y)) \s(y,x)\int_0^1 ds\, e^{-sB}b^*(\eta_x)e^{sB} \\
&+n^{1/2}\int dx dy \, \k\ell^2 V(\ell(x-y)) \\
&\hspace{2cm}\times\big(-\s(y,x)\cN_+/n+b^*_yb(\s_x)-n^{-1}a^*_ya(\s_x)\big)\int_0^1 ds\, e^{-sB}b^*(\eta_x)e^{sB} \\
&+n^{1/2}\int dx dy \, \k\ell^2 V(\ell(x-y)) b(\s_x)\\
&\hspace{2cm}\times\int_0^1 ds\,\big(b^*(p_y^{(s)})+b(\s^{(s)}_y)+(d^{(s)}_y)^*\big) e^{-sB}b^*(\eta_x)e^{sB} \\
&=:\text{G}_{211}+\text{G}_{212}+\text{G}_{213}
\end{split} 
\end{equation}
We use that
\[
 \int_0^1 ds\int dz\,\eta(z,x)\cosh_{s\eta}(w,z)=\sinh_\eta(x,w)
\]
and 
\[
 \int_0^1 ds\int dz\,\eta(z,x)\sinh_{s\eta}(w,z)=(\cosh_\eta-1)(x,w),
\]
resulting in 
\begin{equation}
\begin{split} 
\text{G}_{211}=&n^{1/2}\int dx dy \, \k\ell^2 V(\ell(x-y)) \s(y,x) \int dz\,\eta(z,x)\int_0^1 ds\big(b^*(\g^{(s)}_z)+b(\s^{(s)}_z)\big) \\
&+n^{1/2}\int dx dy \, \k\ell^2 V(\ell(x-y)) \s(y,x) \int dz\,\eta(z,x)\int_0^1 ds(d^{(s)}_z)^* \\
=&n^{1/2}\int dx dy \, \k\ell^2 V(\ell(x-y)) \s(y,x) \big(b^*(\s_x)+b(p_x)\big) \\
&+n^{1/2}\int dx dy \, \k\ell^2 V(\ell(x-y)) \s(y,x) \int dz\,\eta(z,x)\int_0^1 ds(d^{(s)}_z)^* \\
=&:\text{G}_{2111}+\text{G}_{2112}
\end{split} 
\end{equation}
In $\text{G}_{2111}$ we split $\s=\eta+r$ and write
\begin{equation}
\begin{split} 
\text{G}_{2111}=&n^{1/2}\int dx dy \, \k\ell^2 V(\ell(x-y)) \eta(y,x) \big(b^*(\s_x)+b(p_x)\big) +\mathcal{E}_{1}
\end{split} 
\end{equation}
The first contribution plus its hermitian conjugate adds up to the first term in the second line of \eqref{eq:G1linear} plus its hermitian conjugate to give the first contribution in \eqref{eq:cGn3main}, while
\begin{equation}\label{eq:cE1}
\begin{split} 
|\langle\xi,\mathcal{E}_{1}\xi\rangle|&\leq n^{1/2}\int dx dy \, \k\ell^2 V(\ell(x-y)) |r(y,x)| \|\big(b^*(\s_x)+b(p_x)\big)\xi\| \\
&\leq C\k n^{1/2}\ell^{-1}\|(\cN_++1)^{1/2}\xi\|\|\xi\|
\end{split} 
\end{equation}
where we used \eqref{eq:supr} and \eqref{eq:supEtaL2}. We estimate $\text{G}_{2112}$ using estimate \eqref{eq:supEta}, \eqref{eq:supr} and \eqref{eq:supEtaL2} to bound $\sigma(y,x)$, and then estimates \eqref{eq:dxy-bds} and \eqref{eq:supEtaL2}
\begin{equation}\label{eq:G2112}
\begin{split} 
|\langle\xi,\text{G}_{2112}\xi\rangle|&\leq C n^{3/2}\|(\cN_++1)^{1/2}\xi\|\\
&\hspace{1cm}\times\int dx dy \, \k\ell^2 V(\ell(x-y)) \int dz\,|\eta(z,x)|\|(\cN_++1)^{-1/2}d_z\xi\| \\
&\leq C n^{3/2}\|(\cN_++1)^{1/2}\xi\|\left(\int dx dy \, \k\ell^2 V(\ell(x-y)) \int dz\,|\eta(z,x)|^2\right)^{1/2}\\
&\hspace{1cm}\times \left(\int dx dy \, \k\ell^2 V(\ell(x-y)) \int dz\,\|(\cN_++1)^{-1/2}d_z\xi\|^2\right)^{1/2} \\
&\leq C \k n^{1/2}\ell^{-1} \|(\cN_++1)^{1/2}\xi\|\|(\cN_++1)\xi\|
\end{split} 
\end{equation}

Next we analyze $\text{G}_{212}$ in \eqref{eq:G21}. Using \eqref{eq:supEta}, \eqref{eq:supr} and \eqref{eq:supEtaL2} to bound $\sigma(y,x)$, as well as $\eta_x$, Lemma \ref{lm:TNT} and  Lemma \ref{lm:LN4}, we estimate
\begin{equation}\label{eq:G212}
\begin{split} 
|\langle\xi,\text{G}_{212}\xi\rangle|&\leq n^{1/2}\int dx dy \, \k\ell^2 V(\ell(x-y)) \\
&\hspace{1cm}\times \|\big(-\s(y,x)\cN_+/n+b(\s_x)^*b_y-n^{-1}a^*(\s_x)a_y\big)\xi\|\int_0^1 ds\, \|b^*(\eta_x)e^{sB} \xi\|\\
&\leq \k n^{1/2}\ell^{-1}\|(\cN_++1)\xi\|\|(\cN_++1)^{1/2}\xi\|
\end{split} 
\end{equation}
Similar arguments lead to the estimate for the last term in \eqref{eq:G21}:
\begin{equation}\label{eq:G213}
\begin{split} 
|\langle\xi,\text{G}_{213}\xi\rangle|
&\leq n^{1/2}\int dx dy \, \k\ell^2 V(\ell(x-y)) \\
&\hspace{2cm}\times\int_0^1 ds\,\|\big(b(p_y^{(s)})+b^*(\s^{(s)}_y)+d^{(s)}_y\big)b^*(\s_x) \xi\|\|b^*(\eta_x)e^{sB}\xi\| \\
&\leq \k n^{1/2}\ell^{-1}\|(\cN_++1)\xi\|\|(\cN_++1)^{1/2}\xi\|
\end{split} 
\end{equation}

We consider  $\text{G}_{22}$ in \eqref{eq:G2} next. We expand it as
\begin{equation}
\begin{split} 
\text{G}_{22}
=&n^{1/2}\int dx dy \, \k\ell^2 V(\ell(x-y)) b^*(\g_x) \int_0^1 ds\, \big(b^*(\g^{(s)}_y)+b(\s^{(s)}_y)+(d^{(s)}_y)^*\big)e^{-sB}b^*(\eta_x)e^{sB} 
\end{split} 
\end{equation}
and estimate it as
\begin{equation}\label{eq:G22e}
\begin{split} 
|&\langle\xi,\text{G}_{22}\xi\rangle|
\leq n^{1/2}\int dx dy \, \k\ell^2 V(\ell(x-y)) \\
&\hspace{2cm}\times\int_0^1 ds\,\| \big(b_y+b(p^{(s)}_y)+b^*(\s^{(s)}_y)+(d^{(s)}_y)\big)\big(b_x+b(p_x)\big) \xi\|\|b^*(\eta_x)e^{sB}\xi\| \\
\leq&n^{1/2}\int dx dy \, \k\ell^2 V(\ell(x-y))\int_0^1 ds\,\|b^*(\eta_x)e^{sB}\xi\| \Big[\| b_yb_x \xi\|\\
&\hspace{3cm}+\| b_yb(p_x) \xi\| +\| \big(b(p^{(s)}_y)+b^*(\s^{(s)}_y)+d^{(s)}_y\big)\big(b_x+b(p_x)\big) \xi\|\Big]\\
&\leq C \k^{1/2}n^{1/2}\ell^{-1/2}\|(\cN_++1)^{1/2}\xi\|\|\cV_\ell^{1/2}\xi\|+C\k n^{1/2}\ell^{-1}\|(\cN_++1)^{1/2}\xi\|\|(\cN_++1)\xi\|
\end{split} 
\end{equation}
where we used Lemma \ref{lm:TNT},  Lemma \ref{lm:LN4} and Lemma \ref{lm:dx}. 
In order to bound the last contribution in \eqref{eq:G2}, we estimate, similarly as we did for $\text{G}_{22}$, 
\begin{equation}\label{eq:G23e}
\begin{split} 
|&\langle\xi,\text{G}_{23}\xi\rangle|
\leq n^{1/2}\int dx dy \, \k\ell^2 V(\ell(x-y))\int_0^1 ds\,\| \big(b(\g^{(s)}_y)+b^*(\s^{(s)}_y)+(d^{(s)}_y)\big)d_x \xi\|\|b^*(\eta_x)e^{sB}\xi\| \\
&\leq C \k^{1/2}n^{1/2}\ell^{-1/2}\|(\cN_++1)^{1/2}\xi\|\|\cV_\ell^{1/2}\xi\|+C\k n^{1/2}\ell^{-1}\|(\cN_++1)^{1/2}\xi\|\|(\cN_++1)\xi\|
\end{split} 
\end{equation}

We analyze finally the last contribution in \eqref{eq:G3nl}, given by 
\begin{equation}
\begin{split} 
\text{G}_3
&=n^{1/2}\int dx dy \, \k\ell^2 V(\ell(x-y)) e^{-B}b_x^* e^{B} \\
&\hspace{2cm}\times\int_0^1 ds\, \big(b((\eta\g^{(s)})_y)+b^*((\eta\s^{(s)})_y)\big)\big(b(\g^{(s)}_x)+b^*(\s^{(s)}_x)+d^{(s)}_x\big) \\
&\quad+n^{1/2}\int dx dy \, \k\ell^2 V(\ell(x-y)) e^{-B}b_x^* e^{B} \\
&\hspace{2cm}\times\int_0^1 ds\,\int dz \,\eta(z,y) d^{(s)}_z\big(b(\g^{(s)}_x)+b^*(\s^{(s)}_x)+d^{(s)}_x\big) \\
&=:\text{G}_{31}+\text{G}_{32}
\end{split} 
\end{equation}
With Lemma \ref{lm:TNT},  Lemma \ref{lm:LN4} and Lemma \ref{lm:dx} and the bounds \eqref{eq:supEtaL2} and \eqref{eq:normsp} we get
\begin{equation}\label{eq:G31}
\begin{split} 
|\langle\xi,\text{G}_{31}\xi\rangle|
&\leq n^{1/2}\int dx dy \, \k\ell^2 V(\ell(x-y)) \|b_x e^{B}\xi\| \\
&\hspace{2cm}\times\int_0^1 ds\,\|\big(b((\eta\g^{(s)})_y)+b^*((\eta\s^{(s)})_y)\big)\big(b(\g^{(s)}_x)+b^*(\s^{(s)}_x)+d^{(s)}_x\big) \xi\| \\
&\leq C\k n^{1/2}\ell^{-1}\|(\cN_++1)^{1/2}\xi\|\|(\cN_++1)\xi\|;
\end{split} 
\end{equation}
and
\begin{equation}\label{eq:G32}
\begin{split} 
|\langle\xi,\text{G}_{32}\xi\rangle|
&\leq n^{1/2}\int dx dy \, \k\ell^2 V(\ell(x-y)) \|b_x e^{B}\xi\|  \\
&\hspace{2cm}\times\int_0^1 ds\,\int dz \,|\eta(z,y)| \|d^{(s)}_z\big(b(\g^{(s)}_x)+b^*(\s^{(s)}_x)+d^{(s)}_x\big)\xi\| \\
&\leq C\k n^{1/2}\ell^{-1}\|(\cN_++1)^{1/2}\xi\|\|(\cN_++1)\xi\|
\end{split} 
\end{equation}
Putting together \eqref{eq:G11}, \eqref{eq:G12}, \eqref{eq:G13}, \eqref{eq:G14}, \eqref{eq:cE1}, \eqref{eq:G2112}, \eqref{eq:G212}, \eqref{eq:G213}, \eqref{eq:G22e}, \eqref{eq:G23e}, \eqref{eq:G31} and \eqref{eq:G32} 
we arrive at \eqref{eq:errL3}.
\end{proof}
\subsection{Analysis of $\cG_{n,\ell}^{(4)}$}\label{g4}

Recall the definition of $\cL_{n,\ell}^{(4)}$ in \eqref{eq:cLp}.
We define the error $\cE_n^{(4)}$ by
\begin{equation}\label{eq:cLn4E} 
\begin{split} 
 \cG_{n,\ell}^{(4)}=e^{-B}\cL_{n,\ell}^{(4)}e^{B}&=\cV_\ell+\frac{1}{2}\sum_{p,q,r,s\in\L^*_{1,+}}V_{\ell,pqrs}\langle\ph_s\otimes\ph_r,\eta\rangle\langle\eta,\ph_p\otimes\ph_q\rangle\\
 &\quad+\frac{1}{2}\sum_{p,q,r,s\in\L^*_{1,+}}\big( V_{\ell,pqrs}\langle\ph_s\otimes\ph_r,\eta\rangle \hat b^*_p\hat b^*_q+\hc\big)\\
 &\quad+\cE_n^{(4)}
\end{split} 
\end{equation}
It can be estimated as in the Proposition below.

\begin{prop}\label{prop:G4}
Let $\cE_n^{(4)}$ be defined as in \eqref{eq:cLn4E}. Then, under the same assumptions as Proposition \ref{prop:G} for every $\delta>0$ there exists a $C>0$ such that
\begin{equation}\label{eq:errL4}
\pm\cE_{n,\ell}^{(4)}\leq\delta\cV_\ell+ C\k n\ell^{-1}(\cN_++1)
\end{equation}
as operator inequalities on $\cF^{\leq N}_+$.
\end{prop}

\begin{proof} The proof of \eqref{eq:errL4} follows closely \cite[Section 5.6]{BS} and \cite[Section 4.5]{BBCS1}. 
 We write
 \begin{equation}
\begin{split} 
 e^{-B}\cL_n^{(4)}e^{B}&=\frac{1}{2}\int dxdy\,\k \ell^2V(\ell(x-y))e^{-B} a^*_xa^*_ya_xa_ye^B\\
 &=\cV_\ell+\frac{1}{2}\int dxdy\,\k \ell^2V(\ell(x-y))\int_0^1ds\,e^{-sB} [a^*_xa^*_ya_xa_y,B]e^{sB}\\
 &=\cV_\ell+\frac{1}{2}\int dxdy\,\k\ell^2V(\ell(x-y))\int_0^1ds\,e^{-sB}b^*_xb^*_y\big(a_xa^*(\eta_y)+a^*(\eta_x)a_y+\hc\big)e^{sB}\\
\end{split} 
\end{equation}
We expand $e^{-sB}\big(a_xa^*(\eta_y)+a^*(\eta_x)a_y\big)e^{sB}$ further and  get
 \begin{equation}
\begin{split} 
 e^{-B}\cL_n^{(4)}e^{B}-
\cV_\ell &=(\text{W}_1+\text{W}_2+\text{W}_3+\text{W}_4)+\hc
\end{split} 
\end{equation}
with
 \begin{equation}
\begin{split} 
 \text{W}_1 &=\frac{1}{2}\int dxdy\,\k \ell^2V(\ell(x-y))\eta(x,y)\int_0^1ds\,e^{-sB}b^*_xb^*_ye^{sB}\\
 \text{W}_2 &=\int dxdy\,\k \ell^2V(\ell(x-y))\int_0^1ds\,e^{-sB}b^*_xb^*_ye^{sB}a^*(\eta_x)a_y\\
  \text{W}_3 &=\int dxdy\,\k \ell^2V(\ell(x-y))\int_0^1ds\,e^{-sB}b^*_xb^*_ye^{sB}\int_0^sdt\,e^{-tB}b(\eta^{2}_x)b_ye^{tB}\\
  \text{W}_4 &=\int dxdy\,\k \ell^2V(\ell(x-y))\int_0^1ds\,e^{-sB}b^*_xb^*_ye^{sB}\int_0^sdt\,e^{-tB}b^*(\eta_y)b^*(\eta_x)e^{tB}\\
\end{split} 
\end{equation}
The term $\text{W}_1$ results from normal ordering of $a_xa^*(\eta_y)$; in $\text{W}_3$  the notation $\eta^2_x(w)$ stands for $\int dz \,\eta(w,z)\eta(z,x)$.
We start with analyzing $\text{W}_1 $. With the relations \eqref{eq:d} we further expand it as
 \begin{equation}
\begin{split} 
 \text{W}_1 &=\frac{1}{2}\int dxdy\,\k \ell^2V(\ell(x-y))\eta(x,y)\int_0^1ds\,\big(b^*(\g_x^{(s)})+b(\s_x^{(s)})\big)\big(b^*(\g_y^{(s)})+b(\s_y^{(s)})\big)\\
 &\quad+\frac{1}{2}\int dxdy\, \k\ell^2V(\ell(x-y))\eta(x,y)\int_0^1ds\,\big(d_x^{(s)}\big)^*\big(b^*(\g_y^{(s)})+b(\s_y^{(s)})\big)\\ &\quad+\frac{1}{2}\int dxdy\, \k\ell^2V(\ell(x-y))\eta(x,y)\int_0^1ds\,\big(b^*(\g_x^{(s)})+b(\s_x^{(s)})\big)\big(d_y^{(s)}\big)^*\\
  &\quad+\frac{1}{2}\int dxdy\,\k \ell^2V(\ell(x-y))\eta(x,y)\int_0^1ds\,\big(d_x^{(s)}\big)^*\big(d_y^{(s)}\big)^*\\
  &=: \text{W}_{11} +\text{W}_{12}+ \text{W}_{13}+\text{W}_{14},
\end{split} 
\end{equation}
where, for $x,y\in\L_1$, $\g^{(s)}_x(y)$, $\s^{(s)}_x(y)$ and $d^{(s)}_x$ are defined as in \eqref{eq:defGaSi} and \eqref{eq:d} respectively, with $\eta$ substituted by $s\eta$.
Multiplying out the product in $ \text{W}_{11}$ and normal ordering the term $b(\s_x^{(s)})b^*(\g_y^{(s)})$ leads to
 \begin{equation}\label{eq:mainW11}
\begin{split} 
 \text{W}_{11} &=\frac{1}{2}\int dxdy\, \k\ell^2V(\ell(x-y))\eta(x,y)ds\,\Big(b^*_xb^*_y+\frac{1}{2}\eta(x,y)\Big)\\
 &\quad+\text{W}_{112}\\
\end{split} 
\end{equation}
with
 \begin{equation}
\begin{split} 
\text{W}_{112} &=\frac{1}{2}\int dxdy\, \k\ell^2V(\ell(x-y))\eta(x,y)\int_0^1ds\,\Big[b^*_yb(\s_x)-n^{-1}a^*_ya(\s_x)+r^{(s)}(x,y)\\
&\quad-\cN_+n^{-1}\s^{(s)}(x,y)+b^*_xb^*(p^{(s)}_y)+b^*(p^{(s)}_x)\big(b^*_y+b^*(p^{(s)}_y)\big)+\big(b^*_x+b^*(p^{(s)}_x)\big)b(\s^{(s)}_y)\\&\quad+b(\s^{(s)}_x)\big(b^*(p^{(s)}_y)+b^*(\s^{(s)}_y)\big)\Big]\\
\end{split} 
\end{equation}
where $p^{(s)}_x(y)=\g^{(s)}_x(y)-\d(x-y)$ and $r^{(s)}_x(y)=\s^{(s)}_x(y)-s\eta_x(y)$.
The first line in \eqref{eq:mainW11}, together with its hermitian conjugate, gives the main terms in \eqref{eq:cLn4E}. To estimate $\text{W}_{112}$ we first use \eqref{eq:supEta} and then apply Lemma \ref{lm:L2hat} with estimate \eqref{eq:normsp} (for the term proportional to $r^{(s)}$ we use \eqref{eq:supr} and for the term proportional to $\s^{(s)}$ we use \eqref{eq:supEta} and \eqref{eq:supr}). This way we obtain 
 \begin{equation}\label{eq:W112}
\begin{split} 
|\langle\xi,\text{W}_{112}\xi\rangle| &\leq C\k n\ell^{-1}\|(\cN_++1)^{1/2}\xi\|^2
\end{split} 
\end{equation}
To control $\text{W}_{12}$ we use Lemma \ref{lm:dx} and \eqref{eq:supEta}:
 \begin{equation}\label{eq:W12}
\begin{split} 
 |\langle\xi,&\text{W}_{12}\xi\rangle| \leq C\int dxdy\,\k \ell^2V(\ell(x-y))|\eta(x,y)|\int_0^1ds\,\|d_x^{(s)}\xi\|\|\big(b^*(p_y^{(s)})+b(\s_y^{(s)})\big)\xi\|\\ 
 &\quad+C\|(\cN_++1)^{1/2}\xi\|\int dxdy\,\k \ell^2V(\ell(x-y))|\eta(x,y)|\int_0^1ds\,\|(\cN_++1)^{-1/2}b_yd_x^{(s)}\xi\|\\
  &\leq C\k n\ell^{-1}\|(\cN_++1)^{1/2}\xi\|^2+C\k^{1/2}n^{1/2}\ell^{-1/2}\|(\cN_++1)^{1/2}\xi\|\|\cV_\ell^{1/2}\xi\|
\end{split} 
\end{equation}
Similarly, for $\text{W}_{13}$ and for $\text{W}_{14}$ we have
 \begin{equation}\label{eq:W13}
\begin{split} 
 |\langle\xi,\text{W}_{13}\xi\rangle|&\leq C\|(\cN_++1)^{1/2}\xi\| \int dxdy\, \k\ell^2V(\ell(x-y))|\eta(x,y)|\\
 &\qquad\int_0^1ds\,\|(\cN_++1)^{-1/2}d_y^{(s)}\big(b(\g_x^{(s)})+b^*(\s_x^{(s)})\big)\xi\|\\
   &\leq C\k n\ell^{-1}\|(\cN_++1)^{1/2}\xi\|^2+C\k^{1/2}n^{1/2}\ell^{-1/2}\|(\cN_++1)^{1/2}\xi\|\|\cV_\ell^{1/2}\xi\|
\end{split} 
\end{equation}
and
 \begin{equation}\label{eq:W14}
\begin{split} 
 |\langle\xi&,\text{W}_{14}\xi\rangle|\leq C\|(\cN_++1)^{1/2}\xi\| \int dxdy\,\k \ell^2V(\ell(x-y))|\eta(x,y)|\int_0^1ds\,\|(\cN_++1)^{-1/2}d_y^{(s)}d_x^{(s)}\xi\|\\
   &\leq C\k n\ell^{-1}\|(\cN_++1)^{1/2}\xi\|^2+C\k^{1/2}n^{1/2}\ell^{-1/2}\|(\cN_++1)^{1/2}\xi\|\|\cV_\ell^{1/2}\xi\|
\end{split} 
\end{equation}

Next we consider  $ \text{W}_2$. By using \eqref{eq:d} and Lemma \ref{lm:dx}, wee observe that
 \begin{equation}\label{eq:TbbT}
\begin{split} 
\|&(\cN_++1)^{1/2}e^{-sB}b_xb_ye^{sB}\xi\|\leq C\Big[\|a_xa_y(\cN_++1)^{1/2}\xi\|+\|\eta\|\|\eta_y\|\|a_x(\cN_++1)\xi\|\\
&\quad+|\eta(x,y)|\|(\cN_++1)^{1/2}\xi\|+\|\eta\|\|\eta_x\|\|a_y(\cN_++1)\xi\|+\|\eta\|\|\eta_x\|\|\eta_y\|\|(\cN_++1)^{3/2}\xi\|\Big],
\end{split} 
\end{equation}
Combining this with estimate \eqref{eq:supEtaL2} we conclude that
 \begin{equation}\label{eq:W2}
\begin{split} 
 |&\langle\xi,\text{W}_2\xi\rangle|\\
 &\leq C\int dxdy\,\k \ell^2V(\ell(x-y))\int_0^1ds\,\|(\cN_++1)^{1/2}e^{-sB}b_xb_ye^{sB}\xi\|\|(\cN_++1)^{-1/2}a^*(\eta_x)a_y\xi\|\\
    &\leq C\k n\ell^{-1}\|(\cN_++1)^{1/2}\xi\|^2+C\k^{1/2}n^{1/2}\ell^{-1/2}\|(\cN_++1)^{1/2}\xi\|\|\cV_\ell^{1/2}\xi\|
\end{split} 
\end{equation}
With similar arguments to those used to prove \eqref{eq:TbbT} (in particular, using the last two estimates in  Lemma \ref{lm:dx}), we also obtain
 \begin{equation}
\begin{split} 
\|(\cN_++1)^{-1/2}e^{-sB}b(\eta^{(2)}_x)b_ye^{sB}\xi\|&=\|(\cN_++1)^{-1/2}\int dz\,\eta^{(2)}(x,z)e^{-sB}b_zb_ye^{sB}\xi\|\\
&\leq C\Big[\|\eta\|\|\eta_x\|\|a_y\xi\|+\|\eta\|\|\eta_x\|\|\eta_y\|\|(\cN_++1)^{1/2}\xi\|\Big]
\end{split} 
\end{equation}
and
 \begin{equation}
\begin{split} 
\|(\cN_++1)^{-1/2}e^{-sB}b(\eta^{(2)}_x)b(\eta_y)e^{sB}\xi\|&=\|(\cN_++1)^{-1/2}\int dzdt\,\eta(x,z)\eta(y,t)e^{-sB}b_zb_te^{sB}\xi\|\\
&\leq C\|\eta\|\|\eta_x\|\|\eta_y\|\|(\cN_++1)^{1/2}\xi\|
\end{split} 
\end{equation}
leading to
 \begin{equation}\label{eq:W3W4}
\begin{split} 
 |&\langle\xi,\text{W}_3\xi\rangle|,|\langle\xi,\text{W}_4\xi\rangle|\leq C\k n\ell^{-1}\|(\cN_++1)^{1/2}\xi\|^2+C\k^{1/2}n^{1/2}\ell^{-1/2}\|(\cN_++1)^{1/2}\xi\|\|\cV_\ell^{1/2}\xi\|.
\end{split} 
\end{equation}
Estimate \eqref{eq:W3W4}, together with \eqref{eq:W112}, \eqref{eq:W12}, \eqref{eq:W13}, \eqref{eq:W14} and \eqref{eq:W2} conclude the proof of \eqref{eq:errL4}.
\end{proof}

\subsection{Proof of Proposition \ref{prop:G}}\label{proofProp}

\begin{proof}[Proof of Prop. \ref{prop:G}]
From Propositions \ref{prop:G0}, \ref{prop:G1}, \ref{prop:G2K}, \ref{prop:G2V},  \ref{prop:G3} and \ref{prop:G4} we conclude that the excitation Hamiltonian $\cG_{n,\ell}$ can be written as
 \begin{equation}\nonumber
\begin{split} 
e^{-B}\cL_{n,\ell}e^{B} &= C_{n,\ell}+L_{n,\ell}+\cK+Q_{n,\ell} +\cV_\ell+\cE_{n,\ell}
\end{split} 
\end{equation}
where the operators $\cK$ and $\cV_\ell$ are defined as in \eqref{eq:excitHam}, the constant contribution (i.e. the term not depending on operators) $C_{n,\ell}$ is given by \eqref{eq:constantTerm}, the linear terms are given by
 \begin{equation}\label{eq:linear} 
\begin{split} 
L_{n,\ell}&=n^{3/2}\int dx dy\, \k\ell^2V(\ell(x-y)) \big[b(\gamma_x)+b^*(\sigma_x) +\hc\big] \\
&\quad+n^{1/2}\int dx dy\,\k\ell^2V(\ell(x-y))\eta(y,x)  \big[b(\gamma_x)+b^*(\sigma_x) +\hc\big] 
\end{split} 
\end{equation}
and the quadratic terms are
 \begin{equation}\label{eq:TLT} 
\begin{split} 
Q_{n,\ell}&=-\frac{1}{2}\sum_{p,r\in\L^*_{1,+}}\Big(\langle\ph_p\otimes\ph_r,\big(\Delta_1+\Delta_2\big)\eta\rangle\hat  b_p^*\hat b^*_r+\hc\Big)\\
%
&\quad+\frac{1}{2}\sum_{p,q\in\L^*_{1,+}}\Big(nV_{\ell,pq00}\hat b^*_p\hat b^*_q+\hc\Big)\\
&\quad+\frac{1}{2}\sum_{p,q,r,s\in\L^*_{1,+}}\Big( V_{\ell,pqrs}\langle\ph_s\otimes\ph_r,\eta\rangle \hat b^*_p\hat b^*_q+\hc\Big)\\
\end{split} 
\end{equation}
The error term $\cE_{n,\ell}$ satisfies
 \begin{equation}\nonumber
\begin{split} 
 \pm\cE_{n,\ell}\leq \delta (\cK+\cV_{\ell})+C\k n\ell^{-1}(\cN_++1).
\end{split} 
\end{equation}

We first consider the linear terms in \eqref{eq:linear}. Decomposing $\eta$ in the second line of \eqref{eq:linear} with the aid of \eqref{eq:etaQkQIK}, we have
 \begin{equation}
\begin{split} 
L_{n,\ell}
&= \text{L}_{1}+\text{L}_2+\text{L}_3
\end{split} 
\end{equation}
with
 \begin{equation}\label{eq:Lterms}
\begin{split} 
\text{L}_{1}&=n^{3/2}\int dx dy\,\k  \ell^2V(\ell(x-y))\ell^3f_\ell(x,y) \big[b(\gamma_x)+b^*(\sigma_x) +\hc\big]\\
\text{L}_2&=n^{3/2}\int dx dy\,\k \ell^2 V(\ell(x-y))\int dz\big(w_\ell(z,y)+w_\ell(x,z)\big) \big[b(\gamma_x)+b^*(\sigma_x) +\hc\big]\\
\text{L}_3&=-n^{3/2}\int dx dy\,\k \ell^2 V(\ell(x-y))\int dz_1dz_2\,w_\ell(z_1,z_2)\big[b(\gamma_x)+b^*(\sigma_x) +\hc\big]\\
\end{split} 
\end{equation}
By estimates \eqref{eq:L1norm} and \eqref{eq:normsp}, it follows that for any $\xi\in\cF^{\leq n}_+$
  \begin{equation}\label{eq:lin3}
 \begin{split} 
 |\langle\xi,\text{L}_3\xi\rangle|&\leq C n^{3/2}\int dx dy\,\k \ell^2 V(\ell(x-y))|\langle\xi,\big(b(\gamma_x)+b^*(\sigma_x)\big)\xi\rangle|\int dz_1dz_2\,|w_\ell(z_1,z_2)|\\
 &\leq C\k n^{3/2}\ell^{-2}\int dx |\langle\xi,\big(b(\gamma_x)+b^*(\sigma_x)\big)\xi\rangle|\int dy\,\k \ell^3 V(\ell(x-y))\\
 &\leq C\k  n^{1/2}\ell^{-1}\|\cN_+^{1/2}\xi\|\|\xi\|.
\end{split} 
\end{equation}
For the term $\text{L}_2$ we estimate
 \begin{equation}\nonumber
\begin{split} 
|\langle\xi,\text{L}_2\xi\rangle|&\leq n^{3/2}\int dx dy\,\k \ell^2 V(\ell(x-y))\int dz\,|w_\ell(z,y)+w_\ell(x,z)| |\langle\xi,\big(b(\gamma_x)+b^*(\sigma_x)\big)\xi\rangle|\\
&\leq n^{3/2}\int dx|\langle\xi,\big(b(\gamma_x)+b^*(\sigma_x)\big)\xi\rangle|\int dy\,\k \ell^2 V(\ell(x-y))\int dz\,|w_\ell(z,y)| \\
&\quad + n^{3/2}\int dz dx \,|w_\ell(x,z)| |\langle\xi,\big(b(\gamma_x)+b^*(\sigma_x)\big)\xi\rangle|\int dy\,\k \ell^2 V(\ell(x-y))\\
\end{split} 
\end{equation}
Using  \eqref{eq:norm}, \eqref{eq:PWdecay} and \eqref{eq:normsp}, we get
 \begin{equation}\label{eq:lin2}
\begin{split} 
|\langle\xi,\text{L}_2\xi\rangle|
&\leq \k\frac{n^{3/2}}{\ell^2}\int dx|\langle\xi,\big(b(\gamma_x)+b^*(\sigma_x)\big)\xi\rangle| \\
&\quad + \frac{n^{3/2}}{\ell}\int dz dx \,|w_\ell(x,z)| |\langle\xi,\big(b(\gamma_x)+b^*(\sigma_x)\big)\xi\rangle|\int dy\,\k \ell^3 V(\ell(x-y))\\
&\leq C\k n^{3/2}\ell^{-2}\|\cN_+^{1/2}\xi\|\|\xi\|\leq C\k n^{1/2}\ell^{-1}\|\cN_+^{1/2}\xi\|\|\xi\|
\end{split} 
\end{equation}

In order to control  $\text{L}_{1}$ in \eqref{eq:Lterms}, we write it as 
 \begin{equation}\label{eq:L1term}
\begin{split} 
\text{L}_{1}&=n^{3/2}\ell^{-1}c\int dx\, \big[b_x+b(p_x)+b^*(\sigma_x) +\hc\big]\\
&\quad+n^{3/2}\ell^{-1}\int dx \,\big[b(\gamma_x)+b^*(\sigma_x) +\hc\big]\left[\int dy\, \k \ell^3V(\ell(x-y))\ell^3f_\ell(x,y)-c\right]\\
&=:\text{L}_{11}+\text{L}_{12}\\
\end{split} 
\end{equation}
for a constant $c\in \mathbb{R}$. The expectation on $\xi\in \cF_+^{\leq n}$ of $\text{L}_{11}$ in \eqref{eq:L1term} vanishes for any $c$, since $\s,p\in L^2_+(\RRR^3)\times L^2_+(\RRR^3)$. We define
\[
 h_\ell(x)=\int dy\,\k \ell^3V(\ell(x-y))\ell^3f_\ell(x,y)
\]
and we set $c=h_\ell(0)$, where $h_\ell(0)$ is the function $h_\ell$ evaluated at the center of the box.
Let  $d(x)$ denote the distance of  $x$ from the boundary of the box. We denote with $S_{4/\ell}$ the set of all $x\in\L_1$ with  $d(x)<4R_0/\ell$, where $R_0$ is the diameter of the support of $V$. We call $\chi_{S_{4/\ell}}$ the characteristic function of this set.
We split $\text{L}_{12}$ as
 \begin{equation}
\begin{split} 
\text{L}_{12}&=n^{3/2}\ell^{-1}\int dx \,\big[b(\gamma_x)+b^*(\sigma_x) +\hc\big]\left[h_\ell(x)-h_\ell(0)\right]\chi_{S_{4/\ell}^c}(x)\\
&\quad+n^{3/2}\ell^{-1}\int dx \,\big[b(\gamma_x)+b^*(\sigma_x) +\hc\big]\left[h_\ell(x)-h_\ell(0)\right]\chi_{S_{4/\ell}}(x)\\
&=\text{L}_{121}+\text{L}_{122}
\end{split} 
\end{equation}
From \eqref{eq:sup} it follows that $ \sup_{x\in \L_1}h_\ell(x)\leq C \kappa$,
for a constant $C>0$; therefore
 \begin{equation} 
\begin{split} 
|\langle\xi,\text{L}_{122}\xi\rangle|&\leq C \kappa n^{3/2}\ell^{-1}\|\xi\|\int dx \,\|(b(\gamma_x)+b^*(\sigma_x))\xi\|\,\chi_{S_{4/\ell}}(x)\\
&\leq C \kappa n^{3/2}\ell^{-1}\|\xi\|\|(\cN_++1)^{1/2}\xi\|\left(\int dx \,\chi_{S_{4/\ell}}(x)\right)^{1/2}\\
&\leq  C \kappa n^{3/2}\ell^{-3/2}\|\xi\|\|(\cN_++1)^{1/2}\xi\|
\end{split} 
\end{equation}
where we used Cauchy-Schwarz, \eqref{eq:normsp} and \eqref{eq:normEta}.
Using the same bounds, we obtain for $\text{L}_{121}$
\begin{equation}
\begin{split} 
|\langle\xi,\text{L}_{121}\xi\rangle|&\leq Cn^{3/2}\ell^{-1}\|\xi\|\int dx \,\|\big[b(\gamma_x)+b^*(\sigma_x) +\hc\big]\xi\||h_\ell(x)-h_\ell(0)|\chi_{S_{4/\ell}^c}(x)\\
&\leq C n^{3/2}\ell^{-1}\|\xi\|\|(\cN_++1)^{1/2}\xi\|\||h_\ell-h_\ell(0)|\chi_{S_{4/\ell}^c}\|_2
\end{split} 
\end{equation}
Calling $ h(x)=\int_{\L_\ell} dy\, \k V(x-y)f(x,y)$, we have
\begin{equation}\label{eq:PoincLell}
\begin{split} 
\int_{S^c_{4/\ell}} dx\,|h_\ell(x)-h_\ell(0)|^2&=\int_{S^c_{4/\ell}} dx\left|\int dy\, \k\ell^3V(\ell(x-y))\ell^3f_\ell(x,y)-h_\ell(0)\right|^2\\
&=\ell^{-3}\int_{S^c_4} dx\left|\ell^3h(x)-\ell^3h(0)\right|^2
\end{split} 
\end{equation}
where $S_4^{c}$ is the set of points in $\L_\ell$ whose coordinates are at a distance bigger than $4R_0$ from the boundary.
We write
\begin{equation}\nonumber
\begin{split} 
 h(x)-h(0)=\int_{0}^1 dt\, \nabla h (tx)\,x
\end{split}
\end{equation}
and it remains to calculate  $\nabla h$. We have
\begin{equation}\nonumber
\begin{split} 
\partial_{x_i}h(x)
&=-\int_{\L_\ell} dy\,\k \partial_{y_i}V(x-y)f(x,y)+\int_{\L_\ell} dy\,\k V((x-y)\partial_{x_i}f(x,y)\\
&=-\int_{\partial\L_\ell} d\sigma_y\,\k V(x-y)f(x,y)\n_i+\int_{\L_\ell} dy\,\k V(x-y)(\partial_{x_i}+\partial_{y_i})f(x,y)
\end{split} 
\end{equation}
The boundary contribution above vanishes for $x\in S_4^c$. 
Moreover, using  \eqref{eq:PWder} and the fact that $V$ is bounded and compactly supported we obtain
\begin{equation}
\begin{split} 
|\nabla_{x}h(x)|&\leq C\kappa\ell^{-3}\big(d\big(x\big)+1\big)^{-5/3}.
\end{split} 
\end{equation}
Therefore
\begin{equation}\nonumber
\begin{split} 
 \ell^3|h(x)-h(0)|\leq C\k\Big[\int_{0}^1 dt\,\big(d\big(tx\big)+1\big)^{-5/3}|x|\Big]
\end{split}
\end{equation}
To compute the integral, assume that $x^{(3)} \geq \max\{ |x^{(1)}|, |x^{(2)}|\}$. Than $d(tx) = \ell/2 - t x^{(3)}$, and hence 
$$
\int_{0}^1 dt\,\big(d\big(tx\big)+1\big)^{-5/3} = \frac 3{2 x^{(3)} } \left( \frac 1{\left(\ell/2 + 1 - x^{(3)} \right)^{2/3}} - \frac 1{ \left( \ell/2 + 1 \right)^{2/3}} \right) \leq \frac{ 3 \sqrt{3}}{2 |x|} \frac 1{ \left( d(x) + 1\right)^{2/3}}
$$
where we used that $|x|^2\leq 3 \, |x^{(3)}|^2$.  In particular, 
\begin{equation}
\begin{split} 
 \ell^3|h(x)-h(0)|\leq \frac{C\k}{\left( d(x) + 1 \right)^{2/3}}
\end{split}
\end{equation}
from which it easily follows that
\begin{equation}
\begin{split} 
\int_{S^c_4} dx\left| \ell^3\big(h(x)-h(0)\big)\right|^2\leq C\k^2\ell^2
\end{split} 
\end{equation}
We have therefore proved that
\begin{equation}
\begin{split} 
|\langle\xi,\text{L}_{121}\xi\rangle|\leq C\k n^{3/2} \ell^{-3/2}\|(\cN_++1)^{1/2}\xi\|\|\xi\|.
\end{split} 
\end{equation}

We examine now the quadratic contributions in \eqref{eq:TLT}, given by 
\begin{equation}\nonumber 
\begin{split} 
Q_{n,\ell}
&=\frac{n}{2}\int dx dy\,\left[ (\D_x+\D_y)w_\ell(x,y)+\k \ell^2V(\ell(x-y))\big(1-w_\ell(x,y)\big)\right][b_xb_y+b_x^*b_y^*]\\
\end{split} 
\end{equation}
By equation \eqref{eq:6dScatRescOm} and $1-w_\ell(x,y)=\ell^3f_\ell(x,y)$, we have
\begin{equation}\nonumber 
\begin{split} 
Q_{n,\ell}
&=\frac{n\ell^5}{2}\l_\ell\int dx dy\,f_\ell(x,y)[b_xb_y+b_x^*b_y^*]\\
\end{split} 
\end{equation}
For any $\xi\in\cF_+^{\leq n}$ we estimate
\begin{equation}\nonumber 
\begin{split} 
|\langle\xi,Q_{n,\ell}\xi\rangle|
&\leq C\k n\ell^2\int dy\,|\langle b^*(f_\ell(\cdot,y))\xi,b_y\xi\rangle|\\&\leq C\k n\ell^2\|f_\ell\|_2\|(\cN_++1)^{1/2}\xi\|^2\leq C\k \frac{n}{\ell}\|(\cN_++1)^{1/2}\xi\|^2
\end{split} 
\end{equation}
where we used \eqref{eq:slength} and the fact that $f$ is normalized to $1$ in $\L_\ell\times\L_\ell$, so $\|f_\ell\|_2 = \ell^{-3}$.
This concludes the proof of Prop.~\ref{prop:G}.
\end{proof}

\section{Proof of Theorem \ref{mainTheorem} and Corollary \ref{LHY}}\label{proofs}

We shall now use Proposition  \ref{prop:G} and Lemma \ref{lm:TNT} in order to prove Theorem \ref{mainTheorem}.

\begin{proof}[Proof of Theorem \ref{mainTheorem}]
Using the bounds $\cK=\sum_{p\in \L_{1,+}^*}p^2a^*_pa_p\geq \pi^2\sum_{p\in \L_{1,+}^*}a^*_pa_p=\pi^2\cN_+$, $\cV_\ell\geq 0$ and setting $\delta=1/3$, we have, from \eqref{eq:estForCond},
\begin{equation}\label{eq:FirstLB}
\begin{split}
\cG_{n,\ell}&\geq C_{n,\ell}+(1-\delta)(\cK+\cV_\ell)-\kappa C \frac{n}{\ell}(\cN_++1)\geq C_{n,\ell}+\Big(\frac{2}{3}-\kappa \frac{C}{\pi^2} \frac{n}{\ell}\Big)\cK- C\kappa\frac{n}{\ell}
\end{split}
\end{equation}
Assuming $\kappa n/\ell$ small enough we get 
 \begin{equation}\label{eq:low}
\begin{split} 
\cG_{n,\ell}\geq C_{n,\ell}+ \frac{\pi^2}{2}\cN_+ - C\kappa\frac{n}{\ell} \geq C_{n,\ell} - C\kappa\frac{n}{\ell}
\end{split} 
\end{equation}
Equation \eqref{eq:estForCond} also implies (taking $\delta=1$) the upper bound
 \begin{equation}\label{eq:up}
\begin{split} 
\cG_{n,\ell}  &\leq C_{n,\ell}+2(\cK +\cV_\ell)+C\k\frac{n}{\ell}(\cN_++1)
\end{split} 
\end{equation}
From \eqref{eq:up} (evaluated on the vacuum)  and \eqref{eq:low}
it follows that
\begin{equation}\label{eq:ec}
\begin{split} 
\left|e_{n,\ell}-C_{n,\ell}\right|\leq C\k\frac{n}{\ell}
\end{split} 
\end{equation}
Using equation \eqref{eq:constantTerm}, the definition of $\eta$ (in equation \eqref{eq:etaQkQ}) and the fact that it is orthogonal to the condensate wave function $\ph_0$ we  write
 \begin{equation}\label{eq:functional4} 
\begin{split} 
C_{n,\ell}&=\frac{n^2}{2\ell^4} \int_{\L_\ell\times\L_\ell} dxdy\, \Big[\k V(x-y)|1-w( x, y)|^2+|\nabla_xw(x, y)|^2+|\nabla_yw(x, y)|^2\Big]+R_{n,\ell}
\end{split} 
\end{equation}
where
 \begin{equation}\label{eq:fMu} 
\begin{split} 
R_{n,\ell}&=-\frac{n^2}{2}\int dx dydz\, \Big[\,w_\ell(z,y) +\,w_\ell(x,z)-\int dt\,w_\ell(z,t)
\Big](\Delta_x+\Delta_y)w_\ell(x,y)
\end{split} 
\end{equation}
Recalling the definition $1-w=\ell^3f$, where $f$ is the minimizer of \eqref{eq:Fphi} in Proposition \ref{prop:Ff}, we conclude that
 \begin{equation}\label{eq:clog}
\begin{split} 
C_{n,\ell}&=4\pi \mathfrak{a} \frac{n^2}{\ell}\Big(1+\mathcal{O}\Big(\frac{\mathfrak{a}}{\ell}\ln (\ell/\mathfrak{a})\Big)\Big)+R_{n,\ell}
\end{split} 
\end{equation}
The error $R_{n,\ell}$ can be controlled by substituting equation  \eqref{eq:6dScatRescOm} for $(\Delta_x+\Delta_y)w_\ell(x,y)$ and using estimates \eqref{eq:sup} for $f_\ell$ and \eqref{eq:L1norm} for $w_\ell$. This gives  $|R_{n,\ell}|\leq C\k n^2\ell^{-2}$. Equations \eqref{eq:ec} and \eqref{eq:clog} imply \eqref{eq:gpEnergy}.

Let now $\psi_n \in L^2_s (\Lambda_1^n)$ be a normalized wave function, with
\[
\langle \psi_n , H_{n,\ell} \psi_n \rangle \leq e_{n,\ell} + \zeta  
\]
for some $\zeta > 0$ and $e_{n,\ell}$ the ground state energy of $H_n$. We define $\xi_n = e^{-B} U_n \psi_n \in \cF_+^{\leq n}$. Therefore
\[ \langle \xi_n, \cG_{n,\ell} \xi_n \rangle = \langle \psi_n , H_n \psi_n \rangle \leq e_{n,\ell} + \zeta \]
From \eqref{eq:low} and \eqref{eq:ec} we have
\begin{equation} 
\begin{split} 
 \frac{\pi^2}{2}\langle \xi_n,\cN_+  \xi_n \rangle \leq \zeta+ C\kappa\frac{n}{\ell} 
\end{split} 
\end{equation}
Using \eqref{eq:U}, Lemma \ref{lm:TNT} and \eqref{eq:normEta} we have
\begin{equation} 
\begin{split} 
 n - \langle \psi_n,\hat  a_0^* \hat a_0 \psi_n \rangle =\langle \psi_n,U^*_n\cN_+ U_n \psi_n \rangle\leq C \langle \xi_n,\cN_+  \xi_n \rangle \leq \frac{2C}{\pi^2}(\zeta+\kappa n\ell^{-1})
\end{split}
\end{equation}
which implies \eqref{eq:conv-thm}.
\end{proof}

Corollary \ref{LHY} follows from Theorem \ref{mainTheorem}.
\begin{proof}[Proof of Corollary \ref{LHY}] 
Inequality \eqref{eq:gpEnergy} implies that for $n<\frac{c}{\k}\ell=:p$ (where $c$ is a small enough number) there exists a $C>0$ such that
\begin{equation}\label{eq:lboundnell}
\begin{split}
 E(n,\ell)\geq 4\pi\mathfrak{a}\Big[\frac{n^2}{\ell^3}-C\frac{n}{\ell^3}-C\frak a\frac{n^2}{\ell^4} \ln(\ell/\frak a)\Big].
\end{split}
\end{equation}
We need now a bound in the case $n\geq p$.
Following \cite{LY}, we observe that since $V$ is non-negative,
\[
 E(n+n',\ell)\geq E(n,\ell)+E(n',\ell),
\]
where we dropped the interactions between the $n$ particles and the $n'$ particles. It follows that 
\begin{equation}
\begin{split}
 E(n,\ell)\geq \left[\frac{n}{p}\right]E(p,\ell)\geq \frac{n}{2p}E(p,\ell)
\end{split}
\end{equation}
where $\left[\frac{n}{p}\right]$ is the largest integer smaller than $\frac{n}{p}$. We use the latter estimate for $n\geq p$.
Calling $c_n$ the relative number of cells containing $n$ particles, we have that
\begin{equation}
\begin{split}
 \frac{E(N,L)}{N}\geq & \frac{4\pi\mathfrak{a}}{\rho\ell^6} \inf \Big\{\sum_{n<p}c_n\Big(n^2-Cn-C\frak a\frac{n^2}{\ell} \ln(\ell/\frak a)\Big)\\
 &+\frac{1}{2}\sum_{n\geq p}c_nn\Big(p-C-C\frak a\frac{p}{\ell} \ln(\ell/\frak a)\Big)\Big\}
\end{split}
\end{equation}
Defining $A=1-C\frak a\frac{ \ln(\ell/\frak a)}{\ell}$, we need therefore to minimize
\begin{equation}
\begin{split}
\sum_{n<p}c_n\big(n^2A-nC\big)+\frac{1}{2}\sum_{n\geq p}c_nn\big(pA-C\big)
\end{split}
\end{equation}
with the constraints
\[
 \sum_{n\geq 0}c_n=1, \qquad \sum_{n\geq 0}c_nn=\rho\ell^3.
\]
We define the variable
\[
 t= \sum_{n<p}c_nn\leq \rho\ell^3;
\]
we have therefore, by Cauchy-Schwarz,
\begin{equation}\label{eq:min}
\begin{split}
\sum_{n<p}c_n\big(n^2A-nC\big)+\frac{1}{2}\sum_{n\geq p}c_nn\big(pA-C\big)\geq t^2A-tC+\frac{1}{2}(\rho\ell^3-t)(pA-C)
\end{split}
\end{equation}
which we minimize for $1\leq t\leq \rho\ell^3$. If $p$ is large enough, for example  $p\geq 4\rho\ell^3$ (note that this imposes that $\ell^2\geq c(4\kappa\rho)^{-1}$), we obtain that $t=\rho\ell^3$ and the minimum of \eqref{eq:min} is $(\rho\ell^3)^2A-\rho\ell^3C$. This means that
\begin{equation}
\begin{split}
  \frac{E(N,L)}{N}\geq & 4\pi\mathfrak{a}\rho\Big[1 -C\frak a\frac{\ln(\ell/\frak a)}{\ell}-\frac{C}{\rho\ell^3}\Big]
\end{split}
\end{equation}
We set $\ell=(c/4)^{1/2}(\kappa\rho)^{-1/2}$ and we obtain, for a new constant $C>0$,
\begin{equation}
\begin{split}
  \frac{E(N,L)}{N}\geq & 4\pi\mathfrak{a}\rho\Big[1 -C(\rho\frak a^3)^{1/2}\ln(\rho/\frak a)-C(\rho\frak a^3)^{1/2}\Big]
\end{split}
\end{equation}
\end{proof}

\appendix
\section{The two-body problem in the Neumann box} \label{functional}

This Appendix is devoted to proving Propositions \ref{prop:Ff} and \ref{prop:eta}. We  will use the following Lemma.

\begin{lemma}\label{lm:poisson}
 Let $\O=[-\ell/2,\ell/2]^6$ and let $\varepsilon$ be such that $0<\varepsilon\ell^{2}\leq1$. For $y\in \Omega$ let $G_\varepsilon(x,y)$ be the solution of
\begin{equation}\label{eq:gepsilon}
  \big(-\D_x+\varepsilon\big)G_\varepsilon(x,y)=\delta_y(x)
\end{equation}
on $\O$ with Neumann boundary conditions. There exists a constant $C>0$ (independent of $\varepsilon$ and $\ell$) such that
\begin{equation}\label{eq:EstGreen}
G_\varepsilon(x,y)\leq C\Big(\frac{1}{|x-y|^4}+\frac{1}{\ell^6\varepsilon}\Big)
\end{equation}
for every $x,y\in \O$. 
Moreover, let $\tilde G_\varepsilon$ be the unique solution of 
 \begin{equation}\label{eq:GreenR6def}
 \big(-\D_x+\varepsilon\big)\tilde G_\varepsilon(x-y)=\delta_y(x)
 \end{equation}
 on $\mathbb{R}^6$ decaying at infinity. Then
 there exists  a constant $C>0$ such that for $1\leq i\leq 6$
 \begin{equation}\label{eq:DGDG}
  |\partial_{x^{(i)}}G_\varepsilon(x,y)-\partial_{x^{(i)}}\tilde G_\varepsilon(x-y)| \leq  C\left[\sum_{n}\frac{1}{|x-y_n|^5}+\frac{1}{\varepsilon^{1/2}\ell^6}\right]
 \end{equation}
 where the $y_n$ are the (at most $3^6-1$) points obtained by reflecting $y, y_n$ with respect to the planes generated by the sides of the box, whose distance from $y$ is less than $\ell$ (each reflected point is counted only once, and among the $y_n$ we don't include $y$ itself).
\end{lemma}

\begin{proof}
 The solution $\tilde G_\varepsilon$ to \eqref{eq:GreenR6def}
can be expressed as
\begin{equation}\label{eq:GreenR6}
 \tilde G_\varepsilon(x)=\frac{\varepsilon}{2^3\pi^{3}}\frac{\text{K}_2(\sqrt{\varepsilon}|x|)}{|x|^{2}}
\end{equation}
where $\text{K}_2$ is the modified Bessel function of the third kind of order 2 (see \cite{AS}). From the properties of $\text{K}_2$ we deduce that for large $\varepsilon^{1/2}|x|$
\begin{equation}\label{eq:GreenLarge}
  \tilde G_\varepsilon(x)=\frac{\varepsilon^{3/4} e^{-\sqrt{\varepsilon}|x|}}{|x|^{2+1/2}}\Big(1+\mathcal{O}((\sqrt{\varepsilon}|x|)^{-1})\Big)
\end{equation}
while for small $\varepsilon^{1/2}|x|$ there exists a constant $C_1>0$ such that
\begin{equation}\label{eq:GreenSmall}
  \tilde G_\varepsilon(x)=\frac{C_1}{|x|^4}+\mathcal{O}(\varepsilon |x|^{-2})
\end{equation}
We obtain the Green function $G_\varepsilon$ on $\O$ with Neumann boundary conditions as follows. For $x,y\in\O$,
\begin{equation}\label{eq:GreenNeumann}
 G_\varepsilon(x,y)=\tilde G_\varepsilon(x-y)+\sum_{n\in\mathbb{Z}^6\backslash\{0\}} \tilde G_\varepsilon(x-y_n)
\end{equation}
where the positions $y_n$ are all possible reflections (each counted only once) of $y$ and $y_n$ with respect to the infinite planes obtained by extending the sides of the box $\O$ and their periodic replicas over all $\mathbb{R}^6$. This operation gives rise to a grid, and each six-dimensional cell contains one and only one $y_n$ (therefore the label $n\in\mathbb{Z}^6\backslash\{0\}$ also identifies the cell where $y_n$ belongs). The positions $y_n$ can be thought as positions of image charges, whose contributions cancels the normal derivative of $G_\varepsilon$ on $\partial\O$. Given a point $y=(y^{(1)},\dots,y^{(6)})\in\O$, the coordinates of its image charges are, for $j=1, \dots, 6$,
\[
 y^{(j)}_n=n^{(j)}\ell+(-1)^{n^{(j)}}y^{(j)}.
\]

In order to estimate \eqref{eq:GreenNeumann}, we deduce from \eqref{eq:GreenLarge} and \eqref{eq:GreenSmall}  that for any  $0<\lambda<1$ there exists a $C_\lambda>0$ such that
\begin{equation}\label{eq:GreenEstimate}
  \tilde G_\varepsilon(x)\leq\frac{C_\lambda e^{-\lambda\sqrt{\varepsilon}|x|}}{|x|^{4}}.
\end{equation}
Using the estimate above, for the charges that are such that $|x-y_n|\geq \ell$  we bound the contribution in the second term on the right-hand side of \eqref{eq:GreenNeumann} as
\[
 \left|\sum_{n\in\mathbb{Z}^6\backslash\{0\}} \frac{e^{-\lambda\sqrt{\varepsilon}|x-y_n|}}{|x-y_n|^{4}}\right|
 \leq \frac{C_\lambda}{\ell^{4}} \sum_{n\in\mathbb{Z}^6\backslash\{0\}} \frac{ e^{-\lambda\sqrt{\varepsilon}|n|\ell}}{|n|^{4}}
\]
We estimate the sum with an integral (this can be done since the summand is a continuous decreasing function of $n$ on $\mathbb{R}^6\backslash B_1(0)$, where $B_1(0)$ is the ball or radius one centered in zero), so that
\[
 \sum_{n\in\mathbb{Z}^6\backslash\{0\}} \frac{ e^{-\lambda\sqrt{\varepsilon}|n|\ell}}{|n|^{4}}\leq \int_{\mathbb{R}^6\backslash B_1(0)} dn  \frac{ e^{-\lambda\sqrt{\varepsilon}n\ell}}{|n|^{4}}=\frac{1}{(\lambda\sqrt{\varepsilon}\ell)^{2}}\int_{\mathbb{R}^6\backslash B_1(0)} dn  \frac{ e^{-|n|}}{|n|^{4}}
\]
and therefore
 \begin{equation}\label{eq:sumCharges}
 \left|\sum_{n\in\mathbb{Z}^6\backslash\{0\}} \frac{ e^{-\lambda\sqrt{\varepsilon}|x-y_n|}}{|x-y_n|^{4}}\right|\leq \frac{C_\lambda}{\varepsilon\ell^6}
 \end{equation}
Only a finite number of $y_n$ are such that  $|x-y_n|<\ell$, and for those we bound $|x-y|\leq|x-y_n|$. We thus obtain  \eqref{eq:EstGreen}.

 We consider now $\partial_{x_i}\tilde G_\varepsilon(x)$, given by 
 \begin{equation}
 \begin{split}
  \partial_{x^{(i)}}\tilde G_\varepsilon(x)&
   =-x_i\frac{\varepsilon^{3/2}}{2^3\pi^{3}}\frac{\text{K}_3(\sqrt{\varepsilon}|x|)}{|x|^3}\\
 \end{split}
 \end{equation}
 (see  \cite[Chapter 3]{AS} for  properties of the Bessel function of the third kind). 
 For large $\varepsilon^{1/2}|x|$, 
 \begin{equation}
  \begin{split}
   \partial_{x^{(i)}}\tilde G_\varepsilon(x)  &\simeq
   Cx_i\frac{\varepsilon^{5/4} }{|x|^{3+1/2}}e^{-\sqrt{\varepsilon}|x|}\\
  \end{split}
 \end{equation}
 For small $\varepsilon^{1/2}|x|$,
 \begin{equation}\label{eq:sd}
  \begin{split}
  \partial_{x^{(i)}}\tilde G_\varepsilon(x) &\simeq C \frac{x_i}{|x|^6}\\
  \end{split}
 \end{equation}
 The two equations above imply that for any $0<\lambda<1$ there exists a $C_\lambda>0$ such that 
  \begin{equation}
  \begin{split}
 |  \partial_{x^{(i)}}\tilde G_\varepsilon(x) |&\leq  \frac{C_\lambda}{|x|^5}e^{-\lambda\sqrt{\varepsilon}|x|}.
  \end{split}
 \end{equation}
Similarly as above, we sum the contribution from charges such that $|x-y_n|>\ell$, so that
 \begin{equation}
  \begin{split}
 \Big|\sum_{|x-y_n|>\ell}\partial_{x^{(i)}} \tilde G_\varepsilon(x-y_n)\Big|&\leq C_\lambda\sum_{n\in\mathbb{Z}^6\backslash\{0\}}\frac{1}{|x-y_n|^{5}}e^{-\lambda\sqrt{\varepsilon}|x-y_n|}
 \leq\frac{C_\lambda}{\varepsilon^{1/2}\ell^6}
    \end{split}
\end{equation}
 Therefore
 \begin{equation}\label{eq:estDerG}
  \begin{split}
   |\partial_{x^{(i)}} G_\varepsilon(x,y)-\partial_{x^{(i)}} \tilde G_\varepsilon(x-y)| \leq  \sum_{n\neq 0, \, |x-y_n|<\ell}\frac{C}{|x-y_n|^5}+\frac{C}{\varepsilon^{1/2}\ell^6}
  \end{split}
 \end{equation}
 for a constant $C>0$. 
\end{proof}

\begin{proof}[Proof of Proposition \ref{prop:Ff}.] Existence and uniqueness of minimizers can be proved by standard methods. We start by proving \eqref{eq:slength}.
Let $f_0$ be the zero-energy scattering solution defined in \eqref{eq:0en}, and $f(x_1,x_2) = f_0(x_1-x_2)$ for $x_1,\, x_2\in\L_\ell$. We write $\psi = f g$ and integrate by parts. 
Calling $\L_\ell\times\L_\ell=\O$ and writing $\nabla$ for $\nabla_x$, with $x=(x_1,x_2)$, we have
$$
\int_{\Omega} \left( |\nabla \psi |^2 + \k V |\psi|^2 \right) = \int_{\Omega}  f^2 |\nabla g |^2 + \int_{\partial \Omega}  g^2 f \hat n \cdot \nabla f  
$$
where $\hat n$ is the unit outward normal vector, and we use the shorthand notation $V(x) = V(x_1-x_2)$ for simplicity. Note that $\hat n\cdot \nabla f > 0$ since $f_0$ is an increasing function. By assumption $V$ is regular enough such that $f_0 \geq c_0 > 0$ (see \cite[Lemma 5.1]{ESY2} for  properties of the zero energy scattering equation). Let us write $\tau =  \delta_{\partial \Omega} f \hat n \cdot \nabla f $, so that the second term is simply $\int g^2 \tau$. We thus have
\begin{equation}\label{lala}
\int_{\Omega} \left( |\nabla \psi |^2 + \k V |\psi|^2 \right) \geq c_0^2  \int_{\Omega}  |\nabla g |^2 + \int_{\Omega}  g^2 \tau
\end{equation}
Let us look for the lowest eigenvalue of the right-hand side, i.e., the largest $\lambda$ such that 
$$
c_0^2  \int_{\Omega}  |\nabla g |^2 + \int_{\Omega}  g^2 \tau \geq \lambda \int_{\Omega} g^2
$$
Since $f \leq 1$, this is also a lower bound to the eigenvalue we are looking for, i.e.,
$$
\int_{\Omega} \left( |\nabla \psi |^2 + \k V |\psi|^2 \right) \geq \lambda  \int_{\Omega}  g^2 \geq \lambda \int_{\Omega} \psi^2
$$
Clearly  $\lambda \leq \ell^{-6} \int \tau$. Using that $f\leq 1$ we have
$$
\int\tau \leq \int_{\partial \Omega} \hat n \cdot \nabla f =  \int_{\Omega} \Delta f \leq 2 \int_{\L_\ell\times \mathbf{R}^3}dx_1dx_2 (\Delta f_0)(x_1-x_2)  = 8\pi \frak a \ell^3
$$
We may assume that $g$ shares the symmetries of $\Omega$, in which case 
\begin{equation}\nonumber
\begin{split}
 &\int_{\Omega} g^2 \tau = 12 \int_{\Omega} g^2 \tau_1 \\
 &= 12 \int_{\L_\ell} dx_2 \int_{[-\ell/2,\ell/2]^2} dx_1^\perp  g(-\ell/2,x_1^\perp,x_2)^2 f(-\ell/2,x_1^\perp, x_2) \frac{\big(-\ell/2-x_2^{(1)}\big)}{|(-\ell/2,x_1^\perp) - x_2|} f_0'( (-\ell/2,x_1^\perp) - x_2)
\end{split}
\end{equation}
where we write the vector $x_j$ as $(x_j^{(1)}, x_j^\perp)$, and denote the radial derivative of $f_0$ by $f_0'$. 
Using the Schur complement formula, we have, with $Q$ the projection orthogonal to the constant function on $\Omega$,
$$
\lambda \geq \ell^{-6} \int \tau - \frac{12^2}{\ell^6} \langle \tau_1,  Q[ Q(-c_0^2\Delta - 8\pi \frak a \ell^{-3}) Q ]^{-1} Q\,\tau_1\rangle.
$$
Since the spectral gap of $-\Delta$ equals $(\pi/\ell)^2$ we can further bound  
$$
Q(-c_0^2\Delta - 8\pi a \ell^{-3}) Q  \geq \frac{c_0^2}2  Q \left(-\Delta + \ell^{-2} \right)  Q\geq \frac{c_0^2}2  Q \left(-\Delta_{x_1} + \ell^{-2} \right)  Q
$$
as long as $c_0^2(\pi^2-1)/2  \geq 8\pi \frak a/\ell$, which we assume henceforth. In particular,
$$
Q\left[ Q(-c_0^2\Delta - 8\pi a \ell^{-3}) Q  \right]^{-1}Q \leq \frac 2{c_0^2}  Q \left[-\Delta_{x_1} + \ell^{-2} \right]^{-1}  Q = \frac 2{c_0^2}  \left[-\Delta_{x_1} + \ell^{-2} \right]^{-1}  - \frac {2\ell^2} {c_0^2}  P  
$$
with $P=1-Q$ the projection onto the constant function. 
Observing that $-2\ell^2c_0^{-2}P$ can be dropped for an upper bound, we thus have
$$
\lambda \geq \ell^{-6} \int \tau   
  - \frac{2 \cdot 12^2}{c_0^2\ell^6} \langle  \tau_1 ,  [ -\Delta_{x_1} + 1/\ell^2 ]^{-1}   \tau_1\rangle  
$$
An analysis similar to Lemma~\ref{lm:poisson} shows that the integral kernel of $ [ -\Delta_{x_1} + \ell^{-2} ]^{-1}$ on $[-\ell/2,\ell/2]^3$ is bounded above by $c_1 |x_1-y_1|^{-1}$, hence
$$
\langle  \tau_1 ,  [ -\Delta_{x_1} + 1/\ell^2 ]^{-1}   \tau_1\rangle   \leq 
c_1 \int_{\L_\ell^{3}}dx_1 dy_1 dx_2  \frac{ \tau_1(x_1,x_2) \tau_1(y_1, x_2)  }{|x_1-y_1|} 
$$
Using that $f \leq 1$ as well as $f_0'(x_1) \leq \frak a/|x_1|^2$, we have, for fixed $x_2$, 
\begin{align*}
& \int_{\L_\ell^{2}} dx_1 dy_1 \frac{ \tau_1(x_1,x_2) \tau_1(y_1, x_2) }{|x_1-y_1|} \\ 
&\leq \left(\frak a \big(\ell/2+x_2^{(1)}\big)\right)^2 \int_{[-\ell/2,\ell/2]^4} dx_1^\perp dy_1^\perp \frac{ 1}{|x_1^\perp - y_1^\perp|}  \frac 1{| (-\ell/2,x_1^\perp) - x_2 |^3} \frac 1{| (-\ell/2,y_1^\perp) - x_2 |^3} \\ 
& \leq \frac{\frak a^2}{\ell/2+x_2^{(1)}}   \int_{\R^4} dx_1^\perp dy_1^\perp \frac{ 1}{|x_1^\perp - y_1^\perp|}  \frac 1{ ( 1+ (x_1^\perp)^2 )^{3/2} } \frac 1{ ( 1 + (y_1^\perp)^2 )^{3/2}}
\end{align*}
where $\ell/2+x_2^{(1)}$ has been scaled out after extending the integral to $\mathbb{R}^4$.  The final integral is finite by the Hardy-Littlewood-Sobolev inequality.  
In order to obtain a better bound for   $x_2^{(1)}$ close to $-\ell/2$, we use  in addition that $f_0'$ is bounded, and hence that  $f_0'(x) \leq c_2 \frak{a}^{1/2} / |x|^{3/2}$ for some $c_2>0$. Thus
\begin{align*}
& \int_{\L_\ell^{2}} dx_1 dy_1 \frac{ \tau_1(x_1,x_2) \tau_1(y_1, x_2) }{|x_1-y_1|} \\
&\leq \frak{a}\,  c_2^2 (\ell/2+x_2^{(1)} )^2 \int_{[-\ell/2,\ell/2]^4} dx_1^\perp  dy_1^\perp \frac{ 1}{|x_1^\perp - y_1^\perp|}  \frac 1{| (-\ell/2,x_1^\perp) - x_2 |^{5/2}} \frac 1{| (-\ell/2,y_1^\perp) - x_2 |^{5/2}} \\
& \leq \frak{a} \,c_2^2    \int_{\R^4} dx_1^\perp dy_1^\perp \frac{ 1}{|x_1^\perp - y_1^\perp|}  \frac 1{ ( 1+ (x_1^\perp)^2 )^{5/4} } \frac 1{ ( 1 + (y_1^\perp)^2 )^{5/4}}
\end{align*}
where the integral is again finite by the Hardy-Littlewood-Sobolev inequality. 

Altogether, we have thus shown that
$$
 \int_{\O} dx_1 dy_1 \frac{ \tau_1(x_1,x_2) \tau_1(y_1, x_2) }{|x_1-y_1|} \leq C \frak a^2  \min\left\{ \frac {1}{\ell/2+x_2^{(1)}} , \frac{1}{\frak{a}} \right\} 
$$
Integrating this over $x_2$ yields
$$
\int_{\L_\ell} dx_2 \int_{\O} dx_1 dy_1 \frac{ \tau_1(x_1,x_2) \tau_1(y_1, x_2) }{|x_1-y_1|} \leq C \frak a^2 \ell^2 \ln \frac \ell {\frak a}
$$
and thus
$$
\lambda \geq \ell^{-6} \int \tau - C \frac {\frak a^2}{\ell^4} \ln \frac \ell {\frak a}
$$
To complete the lower bound on $\lambda$, we need a lower bound on $\int \tau$. We have
\begin{equation}\label{eq:ub}
\int\tau = \frac 12 \int_{\partial \Omega} \hat n \cdot \nabla f^2 = \frac 12 \int_{\Omega} \Delta f^2 = \int_{\Omega} ( |\nabla f|^2 + \kappa V f^2 ) = 8\pi\frak a \ell^3 - \int_{\L_\ell} dx_1 \int_{\L_\ell^c}  ( |\nabla f|^2 + \kappa V f^2 )
\end{equation}
Using that $\Delta f^2 = 2 ( |\nabla f|^2 + \kappa V f^2 )  \leq C  \min\{ \frak a^{-2} ,\frak a^2 /|x_1-x_2|^4\}$, the error term is bounded by
$$
 C \frak a^2 \int_{\L_\ell} dx_1  \min\{ (\ell/2 + x_1^1)^{-1} , 1 /\frak a \}  = C\frak a^2 \ell^2 \ln \frac {\ell}{\frak a}
$$
We thus conclude that 
$$
\lambda \geq 8\pi\frak a \ell^3 - C\frak a^2 \ell^2 \ln \frac {\ell}{\frak a}
$$
and from \eqref{lala}
\begin{equation}\label{lll|a}
\int_{\Omega} \left( |\nabla \psi |^2 + \k V |\psi|^2 \right) \geq  \lambda \int_{\Omega}  |g|^2 \geq \lambda \int_\Omega |\psi|^2 
\end{equation}
since $f_0 \leq 1$. In particular, $\lambda_\ell \geq \lambda$, and this concludes the lower bound.
The upper bound follows by taking the trial function $\psi=f$ corresponding to $g=1$ and using again \eqref{eq:ub} together with
\[
\|f\|^2_2\geq  \ell^6-C \frak{a} \ell^5,
\]
where the latter follows from \eqref{eq:0enSol}. This completes the proof of \eqref{eq:slength}.

The estimate \eqref{eq:der} in point $i)$ clearly follows  from \eqref{eq:slength}.
We proceed with point $ii)$. 
The minimizer satisfies the eigenvalue equation on $\O$ with Neumann boundary conditions
  \begin{equation}\label{eq:ELeq} 
\begin{split}
\Big[-\Delta_x+\kappa V(x)\Big]f( x)=\lambda_\ell f( x),
\end{split} 
\end{equation}
with $\l_\ell=8\pi \mathfrak{a}\ell^{-3}\big(1+\mathcal{O}(\frak{a}\ell^{-1}\ln (\ell/\mathfrak{a}))\big)$. As before $x=(x_1,x_2)\in\L_\ell\times\L_\ell=\O$ and $\D_x=\D_{x_1}+\D_{x_2}$. Abusing notation we wrote   $V(x)=V(x_1-x_2)$. It is useful to introduce a parameter $0<\varepsilon\leq\ell^{-2}$ and write \eqref{eq:ELeq} as
\begin{equation}\label{eq:ELeqEps} 
\begin{split}
\big(-\Delta_x+\varepsilon\big)f( x)=\big(\lambda_\ell+\varepsilon-\kappa V(x)\big) f( x).
\end{split} 
\end{equation}
We can express the solution to \eqref{eq:ELeqEps} as
\[
 f(x)=\int_\O dy\,G_\varepsilon(x,y)\big(\lambda_\ell+\varepsilon-\kappa V(y)\big) f( y)
\]
with $G_\varepsilon(x,y)$ defined in \eqref{eq:gepsilon}. Lemma \ref{lm:poisson} and the positivity of the minimizer $f$, of $G_\varepsilon(x,y)$ and of the potential $V$ imply that
\begin{equation}\label{eq:InfNorm}
  f(x)\leq C(\lambda_\ell+\varepsilon)\int_\O dy\,\frac{ f( y)}{|x-y|^4}+\frac{C(\lambda_\ell+\varepsilon)}{\ell^6\varepsilon} \int_\O dy\,f( y)
\end{equation}
The last term can be bounded as
\[
\int_\O dy\,f( y)\leq \|f\|_2\|\chi_\O\|_2= \ell^3 
\]
We split the first integral in \eqref{eq:InfNorm} as
\begin{equation}
\int_\O dy\,\frac{ f( y)}{|x-y|^4}=\int_{\O\cap B_\delta(x)} dy\,\frac{ f( y)}{|x-y|^4}+\int_{\O\backslash B_\d(x)} dy\,\frac{ f( y)}{|x-y|^4}
\end{equation}
for $0<\delta\leq\ell$ and $B_\d(x)=\{y\in\mathbb{R}^6:|x-y|\leq\delta\}$.
We have
\begin{equation}
\int_{B_\d(x)} dy\,\frac{ f( y)}{|x-y|^4}\leq C\delta^2\|f\|_\infty
\end{equation}
and 
\begin{equation}
\int_{\O\backslash B_\d(x)} dy\,\frac{ f( y)}{|x-y|^4} 
\leq  \|f\|_2\left(\int_{\mathbb{R}^6\backslash B_\d(x)} dy\,\frac{ 1}{|x-y|^8}\right)^{1/2}= \frac{C}{\delta}
\end{equation}
Hence
\begin{equation}
\|f\|_\infty\leq C(\lambda_\ell+\varepsilon)\Big[\delta^2\|f\|_\infty+\frac {1}{\delta}+ \frac{1}{\ell^3\varepsilon}\Big].
\end{equation}
We set $\varepsilon=\ell^{-2}$ and $\delta^2=\big(2C(\l_\ell+\varepsilon)\big)^{-1}$, so that
$
\|f\|_\infty\leq C'\ell^{-3} 
$, proving \eqref{eq:sup}.

In order to prove \eqref{eq:norm} in point $iii)$, we decompose $f$ as $f=c+g$, with $\int_\O g=0$ and  $c=\ell^{-6}\int f$. We shall show that
\begin{equation}\label{eq:normg}
 \|g\|_2\leq C\kappa\ell^{-1}
\end{equation}
for a constant $C>0$.
Since
\[
 \|f-1/\ell^3\|_2^2\leq 2\|f-c\|_2^2+2\|c-1/\ell^3\|_2^2=2\|g\|_2^2+2|\ell^3c-1|^2
\]
and, since
\[
\|f-c\|_2^2=1-c^2\ell^6\geq|1-c\ell^3|^2,
\]
we have
\[
 \|f-1/\ell^3\|_2^2\leq4\|g\|_2^2
\]
Hence \eqref{eq:norm} follows from \eqref{eq:normg}.
To prove  \eqref{eq:normg} we write equation \eqref{eq:ELeqEps} as
\begin{equation}
\begin{split}
\big(-\Delta_x+\varepsilon\big)g( x)=\big(\lambda_\ell-\kappa V(x)\big) f( x)+\varepsilon g( x)
\end{split} 
\end{equation}
for some $0<\varepsilon\leq \ell^{-2}$. 
We have
\begin{equation}\label{eq:L2normGreen}
 g(x)=\int_\O dy\,G_\varepsilon(x,y)\big(\lambda_\ell-\kappa V(y)\big) f( y)+\varepsilon\int_\O dy\,G_\varepsilon(x,y)g( y)
\end{equation}
By Lemma \ref{lm:poisson} and the Hardy-Littlewood-Sobolev  and H\"older inequalities we have
\begin{equation}\label{eq:norm2first}
\begin{split}
  \Big\|\lambda_\ell\int_\O dy\,G_\varepsilon(\,\cdot\,,y) f( y)\Big\|_2&\leq C\lambda_\ell \Big\|\int_\O dy\,\Big(\frac{1}{|\,\cdot\,-y|^4}+\frac{1}{\ell^6\eps}\Big) f( y)\Big\|_2\\
  &\leq C\lambda_\ell \|f\|_{6/5}+C\frac{\lambda_\ell}{\ell^3\eps}\|f\|_1 \leq \frac{C\k}{\ell^3\eps}
\end{split}
\end{equation}
To bound the contribution  proportional to $V$ in \eqref{eq:L2normGreen}, we use \eqref{eq:sup} and  estimate
\begin{equation}\nonumber
\begin{split}
\int_\O dy\,\Big(\frac{1}{|x-y|^4}+\frac{1}{\ell^6\eps}\Big)V(y) f( y)&\leq \frac{C}{\ell^3}\int_\O dy\,\frac{1}{|x-y|^4}V(y) + \frac{C}{\ell^9\eps}\int_\O dy\,V(y)
\end{split}
\end{equation}
Using the notation $y=(y_1,y_2)\in\L_\ell\times\L_\ell$, we observe that
\begin{equation}\label{eq:estimateV}
\begin{split}
\int_{\L_\ell\times\L_\ell} &dy_1dy_2\,\frac{V(y_1-y_2)}{\big[|x_1-y_1|^2+|x_2-y_2|^2\big]^2}\\
&\leq\int_{\mathbb{R}^3}dy_2\,V(y_2) \int_{\mathbb{R}^3} dy_1\frac{1}{\big[|x_1-y_1|^2+|x_2-y_1+y_2|^2\big]^2}\\
&= C\int_{\mathbb{R}^3}dy_2\, \frac{V(y_2)}{|x_1-x_2-y_2|} \leq \frac {C}{|x_1-x_2|} \int_{\mathbb{R}^3}dy\,V(y)  \\
\end{split}
\end{equation}
where we have used Newton's theorem in the last step. The $L^2$ norm of the last expression is thus bounded by $(\int V) \ell^2$,  
and we conclude that
\begin{equation}\label{eq:norm2second}
\begin{split}
  \Big\|\int_\O dy\,G_\varepsilon(\,\cdot\,,y)V(y) f( y)\Big\|_2&\leq \frac{C}{\varepsilon\ell^3}\|V\|_1
\end{split}
\end{equation}
We are left with the last contribution in \eqref{eq:L2normGreen}. Since $g$ is orthogonal to the constant function, we can use the spectral gap $(\pi/\ell)^2$ of the Laplacian to obtain the bound
%
\begin{equation}\label{eq:norm2third}
\begin{split}
\varepsilon\|G_\varepsilon g\|_2
\leq \varepsilon \frac{\ell^2}{\pi^2}\|g\|_2
\end{split}
\end{equation}
By \eqref{eq:norm2first}, \eqref{eq:norm2second} and \eqref{eq:norm2third} and with the choice $\varepsilon=\ell^{-2}$ (so that $\varepsilon \frac{\ell^2}{\pi^2}<1$) we therefore arrive at 
\eqref{eq:normg}, proving \eqref{eq:norm}. 
The estimate \eqref{eq:L1norm} follows by Cauchy-Schwarz.

Next we examine  point $iv)$. Again, we decompose $f$ as $f=c+g$, with $\int_\O g=0$ and $c$ a positive constant. We observe that
\begin{equation}\nonumber
\begin{split}
|1-\ell^3f(x_1,x_2)|\leq |1-\ell^3c|+\ell^3|g(x_1,x_2)|\leq \|g\|_2+\ell^3|g(x_1,x_2)|
\end{split} 
\end{equation}
Hence, if we show that
\begin{equation}\label{eq:PWg}
\begin{split}
 \sup_{x\in\O}\,(|x_1-x_2|+1)|g(x_1,x_2)|\leq C\kappa\ell^{-3},
\end{split} 
\end{equation}
the bound \eqref{eq:PWdecay} follows. To show \eqref{eq:PWg}, we multiply \eqref{eq:L2normGreen} by $|x_1-x_2|+1$ to obtain
\begin{equation}\label{eq:PWGreen}
\begin{split}
 (|x_1-x_2|&+1)g(x_1,x_2)\\
 &=\int_{\L_\ell\times\L_\ell} dy_1dy_2\, (|x_1-x_2|+1)G_\varepsilon(x_1,x_2,y_1,y_2)\big(\lambda_\ell-\k V(y_1-y_2)\big) f( y_1,y_2)\\
 &\quad+\varepsilon\int_{\L_\ell\times\L_\ell} dy_1dy_2\, (|x_1-x_2|+1)G_\varepsilon(x_1,x_2,y_1,y_2)g( y_1,y_2)
 \end{split} 
\end{equation}
We use Lemma \ref{lm:poisson} to estimate $G_\varepsilon$ and \eqref{eq:sup} as well as \eqref{eq:slength} to get
\begin{equation}\nonumber
\begin{split}
\lambda_\ell\int_{\L_\ell\times\L_\ell} dy_1dy_2&\,(|x_1-x_2|+1)G_\varepsilon(x_1,x_2,y_1,y_2)f( y_1,y_2)\\
&\leq\frac{C}{\ell^6}\int_{\L_\ell\times\L_\ell} dy_1dy_2\,(|x_1-x_2|+1)\left[\frac{1}{\big[|x_1-y_1|^2+|x_2-y_2|^2\big]^2}+\frac{1}{\ell^6\eps}\right]\\
&\leq C\kappa\ell^{-3}+ C\kappa\varepsilon^{-1}\ell^{-5}.
\end{split} 
\end{equation}
Moreover, with \eqref{eq:estimateV} and \eqref{eq:sup}, we have
\begin{equation}\nonumber
\begin{split}
\int_{\L_\ell\times\L_\ell} dy_1dy_2&\,(|x_1-x_2|+1)G_\varepsilon(x_1,x_2,y_1,y_2)\k V(y_1-y_2)f( y_1,y_2)\\
&\leq\frac{C\k}{\ell^3}\int_{\mathbb{R}^3}dy_2\, (|x_1-x_2|+1)\frac{V(y_2)}{|x_1-x_2-y_2|}+\frac{C\k}{\varepsilon\ell^5}\int_{\RRR^3}V(y_2)dy_2\\
\end{split} 
\end{equation}
By Newton's theorem we see  that
\begin{equation}\label{eq:assumption}
\begin{split}
\int_{\mathbb{R}^3}dy_2\,& (|x_1-x_2|+1)\frac{\k V(y_2)}{|x_1-x_2-y_2|}\\&\leq C\int_{\mathbb{R}^3}dy_2\, \k V(y_2)+\frac{1}{|x_1-x_2|}\int_{|y_2|\leq |x_1-x_2|}dy_2\,\k V(y_2)+\int_{|y_2|> |x_1-x_2|}dy_2\, \frac{\k V(y_2)}{|y_2|}\\
&\leq C\k,
\end{split} 
\end{equation}
where we used that $\int dx\, V(x) |x|^{-1}$ is finite.
We conclude  that
\begin{equation}\nonumber
\begin{split}
\int_{\L_\ell\times\L_\ell} dy_1dy_2&\,(|x_1-x_2|+1)G_\varepsilon(x_1,x_2,y_1,y_2)\k V(y_1-y_2)f( y_1,y_2)\leq C\kappa(\ell^{-3}+\varepsilon^{-1}\ell^{-5})
\end{split} 
\end{equation}
We are left with the last contribution in \eqref{eq:PWGreen}. We write, using \eqref{eq:L1norm},
\begin{equation}
\begin{split}
\varepsilon\int_{\L_\ell\times\L_\ell} &dy_1dy_2\, (|x_1-x_2|+1)G_\varepsilon(x_1,x_2,y_1,y_2)g( y_1,y_2)\\&\leq C\varepsilon\int_{\L_\ell\times\L_\ell} dy_1dy_2\, (|x_1-x_2|+1)\left[\frac{1}{|x-y|^4}+\frac{1}{\varepsilon\ell^6}\right]g( y_1,y_2)\\
&\leq C\varepsilon\int_{\L_\ell\times\L_\ell} dy_1dy_2\, \frac{(|x_1-x_2|+1)}{|x-y|^4}g( y_1,y_2)+\frac{C\k}{\ell^3}
 \end{split} 
\end{equation}
We bound the first term above as follows
\begin{equation}\label{eq:PWGreen3}
\begin{split}
C\varepsilon\int_{\L_\ell\times\L_\ell}& dy\,\frac{(|x_1-x_2|+1)}{|x-y|^4}\frac{(|y_1-y_2|+1)g(y)}{(|y_1-y_2|+1)}\\
&\leq C\varepsilon\Big[\sup_{y\in\L_\ell\times\L_\ell}(|y_1-y_2|+1)g(y)\Big]\int_{\L_\ell\times\L_\ell} dy\,\frac{|x_1-x_2|+1}{|y_1-y_2|+1}\frac{1}{|x-y|^4}
 \end{split} 
\end{equation}
Similarly as we did in \eqref{eq:estimateV} we estimate

\begin{equation}
\begin{split}
\int_{\L_\ell\times\L_\ell}& dy\,\frac{1}{|y_1-y_2|+1}\frac{1}{|x-y|^4}\\
&\leq\int_{[-\ell,\ell]^3}dy_2\,\frac{1}{|y_2|+1} \int_{\mathbb{R}^3} dy_1\frac{1}{\big[|x_1-y_1|^2+|x_2-y_1+y_2|^2\big]^2}\\
&\leq\int_{|y_2|\leq \sqrt{3}\ell} dy_2\, \frac{1}{|y_2|}\frac{1}{|x_1-x_2-y_2|} ]\leq C \frac{ \ell^2 }{ |x_1 - x_2| + 1} 
\\
\end{split}
\end{equation}
where we again applied  Newton's theorem in the last step. 
Thus
\begin{equation}
\begin{split}
\int_{\L_\ell\times\L_\ell} dy\,\frac{|x_1-x_2|+1}{|y_1-y_2|+1}\frac{1}{|x-y|^4}&\leq C\ell^2
 \end{split} 
\end{equation}
In conclusion, we have
\begin{equation}\nonumber
\begin{split}
  (|x_1-x_2|+1)|g(x_1,x_2)|\leq& \frac{C\k}{\varepsilon\ell^5}\int_{\mathbb{R}^3}dy\, V(y)+\frac{C\k}{\ell^3} +C\varepsilon\ell^2\Big[\sup_{y\in\L_\ell\times\L_\ell}(|y_1-y_2|+1)|g(y)|\Big]
 \end{split} 
\end{equation}
therefore, by setting $\varepsilon=(2C\ell^2)^{-1}$, we obtain \eqref{eq:PWg}.

Finally we investigate point $v)$. 
As above, we decompose $f=c+g$ with $c=\ell^{-6}\int f$. We shall prove  that
\begin{equation}\label{eq:pwDg}
\begin{split}
 \big[d(\tfrac{x_1+x_2}2 )^{5/3}+1\big]|\nabla_{x_1+x_2}g(x)|\leq C\k\ell^{-3}
\end{split}
\end{equation}
where $d(x)$ is the distance of $x$ from the boundary of the box $\Lambda_\ell$.
By \eqref{eq:L2normGreen}, we have
\begin{equation}\label{eq:supDg2}
\begin{split}
\nabla_{x_1+x_2}g(x)&=-\int_\O dy\,\nabla_{y_1+y_2}\tilde G_\varepsilon(x-y)\big(\lambda_\ell-V(y)\big) f( y)-\varepsilon\int_\O dy\,\nabla_{y_1+y_2}\tilde G_\varepsilon(x-y)g( y)\\
&\quad+\int_\O dy\,\nabla_{x_1+x_2}\Big[G_\varepsilon(x,y)-\tilde G_\varepsilon(x-y)\Big]\big(\lambda_\ell-\k V(y)\big) f( y)\\
&\quad+\varepsilon\int_\O dy\,\nabla_{x_1+x_2}\Big[G_\varepsilon(x,y)-\tilde G_\varepsilon(x-y)\Big]g( y)\\
  \end{split} 
\end{equation}
We integrate by parts in the first line, and obtain
\begin{equation}\label{a43}
\begin{split}
\nabla_{x_1+x_2}g(x)&=\int_\O dy\,\tilde G_\varepsilon(x-y)\big(\lambda_\ell-\kappa V(y)\big) \nabla_{y_1+y_2}f( y)+\varepsilon\int_\O dy\,\tilde G_\varepsilon(x-y)\nabla_{y_1+y_2}g( y)\\
&\quad+\int_{\partial\O} d\s_y\,\hat{n}\,\tilde G_\varepsilon(x-y)\big(\lambda_\ell-\k V(y)\big) f( y)+\varepsilon\int_{\partial\O} d\s_y\,\hat{n}\,\tilde G_\varepsilon(x-y)g( y)\\
&\quad+\int_\O dy\,\nabla_{x_1+x_2}\Big[G_\varepsilon(x,y)-\tilde G_\varepsilon(x-y)\Big]\big(\lambda_\ell-\k V(y)\big) f( y)\\
&\quad+\varepsilon\int_\O dy\,\nabla_{x_1+x_2}\Big[G_\varepsilon(x,y)-\tilde G_\varepsilon(x-y)\Big]g( y)=\sum_{j=1}^6\text{D}_j(x)
  \end{split} 
\end{equation}
where $d\sigma_y$ is the surface element of the boundary of the box $\partial\Omega$ and $\hat n$ is the unit vector pointing outwards.
We start by considering $\text{D}_2$. Using \eqref{eq:GreenEstimate}, we can bound, for every $x\in\Omega$,
\begin{equation} 
\begin{split}
 &|\big[d( \tfrac{x_1+x_2}2)^{5/3}+1\big]\text{D}_2(x)|\\
 &\leq C\varepsilon \big[d( \tfrac{x_1+x_2}2)^{5/3}+1\big]\int_{\O} dy\,\frac{\big|\big[d(y_1+y_2)^{5/3}+1\big]\nabla_{y_1+y_2}g(y)\big|}{|x-y|^4\big[d(\tfrac{y_1+y_2}2)^{5/3}+1\big]}\\
 &\leq C\varepsilon\sup_{y\in \O}\big|\big[d(\tfrac{y_1+y_2}2)^{5/3}+1\big]\nabla_{y_1+y_2}g(y)\big| \int_{\O} dy\,\frac{d(\tfrac{x_1+x_2}2)^{5/3}+1}{|x-y|^4\big[d(\tfrac{y_1+y_2}2)^{5/3}+1\big]}
  \end{split} 
\end{equation}
In the following we shall show that 
$$
\int_{\O} dy\,\frac{1}{|x-y|^4\big[d(\tfrac{y_1+y_2}2)^{5/3}+1\big]}   \leq \frac{ C \ell^2}{  d(\tfrac{x_1+x_2}2)^{5/3} +1}
$$
Since 
\begin{equation}\label{dp}
\frac{ 1}{  d(\tfrac{x_1+x_2}2)^{5/3} +1} \leq  \sum_{i=1}^3 \sum_{j=1}^2 \frac{ 1}{2^{-5/3} |x_1^{(i)}+x_2^{(i)}- (-1)^j \ell|^{5/3} +1} \leq \frac{ 6}{  d(\tfrac{x_1+x_2}2)^{5/3} +1}
\end{equation}
it is sufficient to prove that 
\begin{equation}\label{stp}
 \int_{\O} dy\,\frac{1}{|x-y|^4\big[|y_1^{(1)}+y_2^{(1)}-\ell|^{5/3}+1\big]}  \leq \frac{ C \ell^2}{  |x_1^{(1)}+x_2^{(1)}-\ell| ^{5/3} +1}
\end{equation}
For this purpose, we shall write
 \begin{equation}\label{eq:rewrite}
\begin{split}
|x_1- y_1|^2+|x_2- y_2|^2=\frac{1}{2}\big|(x_1+x_2)-( y_1+ y_2)\big|^2+\frac{1}{2}\big|(x_1-x_2)-( y_1- y_2)\big|^2;
\end{split} 
\end{equation}
with the change of variable $ y_1+ y_2=b, \,  y_1- y_2=a$ 
we have
\begin{equation}\label{eq:integral}
\begin{split}
\int_{\O}& dy\,\frac{|x_1^{(1)}+x_2^{(1)}-\ell|^{5/3}+1}{|x-y|^4\big[|y_1^{(1)}+y_2^{(1)}-\ell|^{5/3}+1\big]}\\
&=\frac 12\int_{[-\ell,\ell]^3} db\int_{\omega(b)}da\,\frac{|x_1^{(1)}+x_2^{(1)}-\ell|^{5/3}+1}{\big[\big|(x_1+x_2)-b\big|^2+\big|(x_1-x_2)-a\big|^2\big]^2\big[|b^{(1)}-\ell|^{5/3}+1\big]} 
  \end{split} 
\end{equation}
where $\omega(b)=[|b^{(1)}|-\ell,\ell-|b^{(1)}|]\times[|b^{(2)}|-\ell,\ell-|b^{(2)}|]\times[|b^{(3)}|-\ell,\ell-|b^{(3)}|]$. Let us introduce the notation $a= (a^{(1)},a^\perp)$ and $b=(b^{(1)}, b^{\perp})$.
To bound \eqref{eq:integral} we bound the numerator with $2 (2\ell)^{5/3}$ (assuming $2\ell \geq 1$) and  extend the integration domain of the variable $a^\perp$ to $[-\ell,\ell]^2$; dropping the term involving $a^{(1)}$ in the denominator, we can integrate over $a^{(1)}$ to obtain the bound
\begin{equation}\nonumber
\begin{split}
&\int_{[-\ell,\ell]^3} db\int_{\omega(b)}da\,\frac{|x_1^{(1)}+x_2^{(1)}-\ell|^{5/3}+1}{\big[\big|(x_1+x_2)-b\big|^2+\big|(x_1-x_2)-a\big|^2\big]^2\big[|b^{(1)}-\ell|^{5/3}+1\big]} \\
&\leq 2 (2\ell)^{5/3}\int_{[-\ell,\ell]} db^{(1)}\frac{1}{|b^{(1)}-\ell|^{2/3}} \int_{[-\ell,\ell]^4}db^\perp da^\perp\,\frac{1}{\big[\big|(x_1+x_2)-b\big|^2+\big|(x_1^\perp-x_2^\perp)-a^\perp\big|^2\big]^2}\\
  \end{split} 
\end{equation}
We estimate
\begin{equation}\nonumber
\begin{split}
 \int_{[-\ell,\ell]^4}db^\perp da^\perp\,&\frac{1}{\big[\big|(x_1+x_2)-b\big|^2+\big|(x_1^\perp-x_2^\perp)-a^\perp\big|^2\big]^2}\leq C\int_{|(x_1^{(1)}+x_2^{(1)})-b^{(1)}|}^{3\ell}dz\,\frac{1}{z}\\&= C\ln\left(\frac{3\ell}{|(x_1^{(1)}+x_2^{(1)})-b^{(1)}|}\right)\\
  \end{split} 
\end{equation}
and
\begin{equation}\nonumber
\begin{split}
\int_{[-\ell,\ell]} db^{(1)}&\frac{1}{|b^{(1)}-\ell|^{2/3}}\ln\left(\frac{3\ell}{|(x_1^{(1)}+x_2^{(1)})-b^{(1)}|}\right) \leq (2\ell)^{1/3} \sup_{0<s<1} \int_0^1 dt \, t^{-2/3} \ln \frac {3/2}{ t - s } \leq C \ell^{1/3}
\end{split} 
\end{equation}
This proves \eqref{stp}. 
We have thus shown that 
\begin{equation}\label{eq:D2}
\begin{split}
 &|\big[d(\tfrac{x_1+x_2}2)^{5/3}+1\big]\text{D}_2(x)|\leq C\varepsilon\ell^2\sup_{y\in \O}\big|\big[d(\tfrac{y_1+y_2}2)^{5/3}+1\big]\nabla_{y_1+y_2}g(y)\big|
  \end{split} 
\end{equation}

We proceed with $\text{D}_1$, which we write as 
 $\text{D}_1=\text{D}_{11}+\text{D}_{12}$, with 
\begin{equation}
\begin{split}
\text{D}_{11}(x)&=\lambda_\ell\int_\O dy\,\tilde G_\varepsilon(x-y) \nabla_{y_1+y_2}f( y)
  \end{split} 
\end{equation}
and
\begin{equation}
\begin{split}
\text{D}_{12}(x)&=-\int_\O dy\,\tilde G_\varepsilon(x-y)\k V(y) \nabla_{y_1+y_2}f( y)
  \end{split} 
\end{equation}
Using the same method as above we estimate
\begin{equation}\label{eq:D11}
\begin{split}
|\big[d(\tfrac{x_1+x_2}2)^{5/3}+1\big]\text{D}_{11}(x)|&\leq C\l_\ell\ell^2\sup_{y\in \O}|\big[d(\tfrac{y_1+y_2}2)^{5/3}+1\big]\nabla_{y_1+y_2}g(y)|.
\end{split} 
\end{equation}
For $\text{D}_{12}$ we have
\begin{equation}\label{eq:D12}
\begin{split}
|\big[&d(\tfrac{x_1+x_2}2)^{5/3}+1\big]\text{D}_{12}(x)|\\
&\leq \sup_{y\in \O}|\big[d(\tfrac{y_1+y_2}2)^{5/3}+1\big]\nabla_{y_1+y_2}g(y)|\int_\O dy\,\frac{\k V(y)\big[d(\tfrac{x_1+x_2}2)^{5/3}+1\big]}{|x-y|^4\big[d(\tfrac{y_1+y_2}2)^{5/3}+1\big]}
  \end{split} 
\end{equation}
Because of \eqref{dp} it again suffices to bound the last integral with $d(z)$ replaced by $|z^{(1)} - \ell|$ for both $z=x_1+x_2$ and $z=y_1+y_2$. 
With the same change of variables as in \eqref{eq:integral} we have
\begin{equation}
\begin{split}
\int_{\O}& dy\,\frac{\k V(y)}{|x-y|^4\big[\big(y_1^{(1)}+y_2^{(1)}\big)^{5/3}+1\big]}\\
&\leq \frac 12\int_{\RRR^3} da\int_{\RRR^3}db\,\frac{\k V(a)}{\big[\big|(x_1+x_2)-b\big|^2+\big|(x_1-x_2)-a\big|^2\big]^2\big[|b^{(1)}-\ell|^{5/3}+1\big]} 
  \end{split} 
\end{equation}
where we extended integration domain to $\RRR^6$.
Integrating first in the variable $b^\perp$ we have
\begin{equation}\label{eq:intV}
\begin{split}
\int_{\RRR^3} da&\int_{\RRR^3}db\,\frac{\k V(a)}{\big[\big|(x_1+x_2)-b\big|^2+\big|(x_1-x_2)-a\big|^2\big]^2\big[|b^{(1)}-\ell|^{5/3}+1\big]}\\
&= C\int_{\R} db^{(1)} \frac{1}{|b^{(1)}-\ell|^{5/3}+1}\int_{\RRR^3} da\frac{\k V(a)}{\big|\big(x_1^{(1)}+x_2^{(1)}\big)-b^{(1)}\big|^2+\big|(x_1-x_2)-a\big|^2}
  \end{split} 
\end{equation}
Using that $V$ is bounded and of compact support, one readily checks that 
\begin{equation}\label{sito}
\int_{\RRR^3} da \frac{V(a)}{X+|Y-a|^2} \leq\frac{C}{X+ |Y|^2+1}
\end{equation}
for all $X\geq0$ and $Y\in\RRR^3$.
Hence we find that
\begin{equation}
\begin{split}
\int_{\O} dy\,\frac{\k V(y)}{|x-y|^4\big[|y_1^{(1)}+y_2^{(1)}-\ell|^{5/3}+1\big]}&\leq C\kappa\int_{\R} db^{(1)} \frac{1}{|b^{(1)}-\ell|^{5/3}+1} \frac{1}{|(x_1^{(1)}+x_2^{(1)})-b^{(1)}|^{2}+1}\\
&\leq \frac{C\kappa}{\big[|x_1^{(1)}+x_2^{(1)}-\ell|^{5/3}+1\big]}
  \end{split} 
\end{equation}
which is the desired bound, allowing us to conclude that  
\begin{equation}\label{eq:ksmall}
\begin{split}
|\big[&d(\tfrac{x_1+x_2}2)^{5/3}+1\big]\text{D}_{12}(x)|\leq C\kappa \sup_{y\in \O}|\big[d(\tfrac{y_1+y_2}2)^{5/3}+1\big]\nabla_{y_1+y_2}g(y)|.
  \end{split} 
\end{equation}

In order to bound $\text{D}_5$ we split it into 
\begin{equation}
\begin{split}
\text{D}_{51}(x)&=\lambda_\ell\int_\O dy\,\nabla_{x_1+x_2}\Big[G_\varepsilon(x,y)-\tilde G_\varepsilon(x-y)\Big]f( y)\\
  \end{split} 
\end{equation}
and
\begin{equation}
\begin{split}
\text{D}_{52}(x)&=-\int_\O dy\,\nabla_{x_1+x_2}\Big[G_\varepsilon(x,y)-\tilde G_\varepsilon(x-y)\Big]\k V(y)f( y)\\
  \end{split} 
\end{equation}
We easily bound
\begin{equation}\label{eq:D51}
\begin{split}
 |\big[d(\tfrac{x_1+x_2}2)^{5/3}&+1\big]\text{D}_{51}(x)|\\
 & \leq
 \lambda_\ell\big[d(\tfrac{x_1+x_2}2)^{5/3}+1\big]\int_\O dy\,\Big[ \sum_{ n}\frac{1}{|x-y_n|^5}+\frac{1}{\varepsilon^{1/2}\ell^6}\Big]f( y)\\
 &\leq C \kappa (\ell^{-3-1/3}+\ell^{-4-1/3}\varepsilon^{-1/2})
 \end{split} 
\end{equation}
where we used \eqref{eq:DGDG}, \eqref{eq:slength} and \eqref{eq:sup} and we estimated  $d(\tfrac{x_1+x_2}2)\leq\ell/2$.
For $\text{D}_{52}$, we use again \eqref{eq:DGDG} and \eqref{eq:sup} to estimate
\begin{equation}\label{eq:D52}
\begin{split}
|\big[&d(\tfrac{x_1+x_2}2)^{5/3}+1\big]\text{D}_{52}(x)|\leq \frac{C\big[d(\tfrac{x_1+x_2}2)^{5/3}+1\big]}{\ell^3}\int_\O dy\,\Big[ \sum_{ n}\frac{1}{|x- y_n|^5}+\frac{1}{\varepsilon^{1/2}\ell^6}\Big]\kappa V(y).
  \end{split} 
\end{equation}
For the second contribution we have 
\begin{equation}
\begin{split}
\frac{\big[d(\tfrac{x_1+x_2}2)^{5/3}+1\big]}{\varepsilon^{1/2}\ell^9}\int_\O dy\,\kappa V(y)\leq \frac{C\kappa}{\varepsilon^{1/2}\ell^{4+1/3}}\\
  \end{split} 
\end{equation}
For the first contribution in \eqref{eq:D52}, among the  image charges $ y_n$ we start considering the one that has all coordinates equal to those of $y$ except the first one. We rename it as $\tilde y$, and we have $\tilde y_1^{(1)}=-\ell- y_1^{(1)}$. We write 
\begin{equation} 
\begin{split}
\int_\O dy\,&\frac{1}{|x-\tilde y|^5} V(y)=\int_\O dy_1dy_2\,\frac{V(y_1-y_2)}{(|x_1-\tilde y_1|^2+|x_2-\tilde y_2|^2)^{5/2}}\\
&=\int_{[-3\ell/2,-\ell/2]\times[-\ell/2,\ell/2]^5} d\tilde y_1d\tilde y_2\,\frac{ V(\tilde y^{(1)}_1+\tilde y^{(1)}_2+\ell,\tilde y^{\perp}_1-\tilde y^{\perp}_2)}{(|x_1-\tilde y_1|^2+|x_2-\tilde y_2|^2)^{5/2}}
  \end{split} 
\end{equation}
The denominator can be expressed as in \eqref{eq:rewrite};
with the change of variables $ \tilde y_1+\tilde y_2=b, \, \tilde y_1-\tilde y_2=a$, we have (extending the integration domain to $\RRR^6$)
\begin{equation} 
\begin{split}
\int_\O dy\,\frac{1}{|x-\tilde y|^5}V(y) & \leq \frac{1}{\sqrt{2}}\int_{\RRR^3} db^{(1)}da^{\perp}\, V(b^{(1)}+\ell, a^{\perp})\\
& \quad \times \int_{\RRR^3}da^{(1)} db^{\perp}\,\frac{1}{\big(\big|x_1+x_2-b\big|^2+\big|x_1-x_2-a\big|^2\big)^{5/2}}\\
&= C \int_{\mathbb{R}^3} db^{(1)}da^{\perp}\frac{ V(b^{(1)}+\ell,  a^{\perp})}{\big(x_1^{(1)}+x_2^{(1)}-b^{(1)}\big)^2+\big(x_1^\perp-x_2^\perp-a^\perp\big)^2}
  \end{split} 
\end{equation}
Using now again \eqref{sito}, we arrive at 
\begin{equation} 
\begin{split}
\frac{1}{\ell^3}&\left|\int_\O dy\,\frac{1}{|x-\tilde y|^5}\k V(y)\right|\leq \frac{C}{\ell^3}\frac{\kappa}{|x_1^{(1)}+x_2^{(1)}-\ell|^2+1} \leq \frac{C}{\ell^3}\frac{\kappa}{d(x_1+x_2)^2+1} 
  \end{split} 
\end{equation}
The contribution from the other image charges can be estimated similarly, and we omit the details. We conclude that 
\begin{equation}\label{eq:D52b}
\begin{split}
|\big[d(\tfrac{x_1+x_2}2)^{5/3}+1\big]\text{D}_{52}(x)|&\leq\frac{C\kappa}{\varepsilon^{1/2}\ell^{4+1/3}}+\frac{C\kappa}{\ell^{3}}
  \end{split} 
\end{equation}

Next we investigate $\text{D}_6$. With \eqref{eq:DGDG} and \eqref{eq:PWg} we have
\begin{equation}\label{eq:D6est}
\begin{split}
|\text{D}_6(x)|&=\Big|\varepsilon\int_\O dy\,\nabla_{x_1+x_2}\Big[G_\varepsilon(x,y)-\tilde G_\varepsilon(x-y)\Big]g( y)\Big|\\
&\leq \frac{C\varepsilon^{1/2}}{\ell^6}\int_\O dy\,|g( y)|+\frac{\varepsilon\kappa}{\ell^3}\sum_{n}\int_\O dy\,\frac{1}{|x- y_n|^5}\frac{1}{|y_1-y_2|+1}\\
  \end{split} 
\end{equation}
To bound the first term we can also use \eqref{eq:PWg}, which gives $\int|g|\leq C\k \ell^{2}$.
To estimate the second term in \eqref{eq:D6est} we start, as above, by considering the image charge $\tilde y$ such that $\tilde y_1^{(1)}=-\ell- y_1^{(1)}$, $\tilde y_1^{(i)}= y_1^{(i)}$ for $i=2,3$ and $\tilde y_2^{(j)}= y_2^{(j)}$  for $j=1,2,3$. We perform again the change of variables $ \tilde y_1+\tilde y_2=b, \, \tilde y_1-\tilde y_2=a$ and extend the integration domains 
so that 
\begin{equation}\nonumber
\begin{split}
&\int_\O dy\,\frac{1}{|x- \tilde y|^5}\frac{1}{|y_1-y_2|+1}\\
& \leq \frac 1 {\sqrt{2}}\int_{[-\ell,\ell]^2} da^{\perp}\,\int_{\RRR^4}db^{(1)}db^{\perp}da^{(1)}\\& \qquad\quad \times \frac{1}{\big[\big|x_1+x_2-b\big|^2+\big|x_1-x_2-a\big|^2\big]^{5/2}}\frac{1}{\big[|b^{(1)}+\ell|^2+|a^{\perp}|^2\big]^{1/2}+1}\\
&\leq  C \int_{ [-\ell,\ell]^2} da^{\perp}\,\frac{1}{ |a^{\perp}|+1}\frac{1}{\big|x_1^\perp-x_2^\perp-a^\perp\big|}
  \end{split} 
\end{equation}
where we dropped the term $|b^{(1)}+\ell|^2$ in the last step in order to be able to explicitly integrate over $b^{(1)}$. It is easy to see that the remaining integral is bounded by $\ln \ell$, uniformly in $x_1^\perp - x_2^\perp$. The same estimate can be applied to the other imaged charges, with the result that 
\begin{equation}\nonumber
\begin{split}
|\text{D}_6(x)|&\leq C\k \varepsilon^{1/2}\ell^{-4}+\textcolor{black}{C\k \varepsilon\ell^{-3}\ln(\ell)}
  \end{split} 
\end{equation}
In particular, 
\begin{equation}\label{eq:D6}
\begin{split}
|\big[d(\tfrac{x_1+x_2}2)^{5/3}+1\big]\text{D}_6(x)|&\leq C\k \varepsilon^{1/2}\ell^{-7/3}+C\k \varepsilon\ell^{-4/3}\ln(\ell)
  \end{split} 
\end{equation}

We are left with considering $\text{D}_3$ and $\text{D}_4$. With the aid of \eqref{eq:GreenEstimate} and \eqref{eq:PWg} we can bound 
\begin{equation}
\begin{split}
|\text{D}_4(x)|&=\Big|\varepsilon\int_{\partial\O} d\s_y\,\hat n\,\tilde G_\varepsilon(x-y) g( y)\Big|\leq \frac{C\varepsilon \kappa}{\ell^3}\int_{\partial\O} d\s_y  \frac{1}{|x-y|^4} \frac{1}{|y_1-y_2|+1}
  \end{split} 
\end{equation}
It clearly suffices to consider the contribution to the boundary integral coming from $y_1^{(1)}= -\ell/2$. 
With the change of variables $ y_1^{\perp}+ y_2^{\perp}=b^{\perp}, \,  y_1^{\perp}- y_2^{\perp}=a^{\perp}$ we have, similarly as above, 
\begin{equation}\nonumber
\begin{split}
&\int_{\partial\O} d\s_y\,\frac{1}{|x-y|^4}\frac{1}{|y_1-y_2|+1}\\
&\leq   \int_\RRR  dy_2^{(1)}\int_{[-\ell,\ell]^2}da^\perp\,\frac{1}{\big[|y_2^{(1)}+\ell/2|^2+\big| a^{\perp}\big|^2\big]^{1/2}+1}\\
&\quad \times \int_{\RRR^2} db^\perp\,\frac{1}{\Big[\big|x_1^\perp+x_2^\perp-b^\perp\big|^2+\big|x_1^\perp-x_2^\perp-a^\perp\big|^2+ 2\big|x_1^{(1)}+\ell/2\big|^2 + 2 \big|x_2^{(1)}-y_2^{(1)}\big|^2\Big]^2}\\
& \leq C \int_{[-\ell,\ell]^2}da^\perp\,\frac{1} {  \big| a^{\perp}\big|+1}  \frac{1}{\big|x_1^\perp-x_2^\perp-a^\perp\big|}  \leq C \ln \ell\\  \end{split} 
\end{equation}
and thus
\begin{equation}\label{eq:D4}
\begin{split}
|\big[d(\tfrac{x_1+x_2}2)^{5/3}+1\big]\text{D}_4(x)|&\leq C\k   \varepsilon\ell^{-4/3}\ln(\ell)
  \end{split} 
\end{equation}
In $\text{D}_3$ we estimate the contribution proportional to $\lambda_\ell$ as
\begin{equation}\nonumber
\begin{split}
\lambda_\ell\int_{\partial\O} d\s_y\,\tilde G_\varepsilon(x-y)f( y)\leq C\frac{\lambda_\ell}{\ell^3}\int_{\partial\O} d\s_y \frac{1}{|x-y|^4}\leq  \frac {C \k }{\ell^{5}} ,
  \end{split} 
\end{equation}
where we used \eqref{eq:sup} and \eqref{eq:slength}. 
For the contribution proportional to $V$, we use again \eqref{eq:sup} to bound it as
\begin{equation}\nonumber
\begin{split}
&\int_{\partial\O} d\s_y\,\tilde G_\varepsilon(x-y)\k V(y)f( y)\leq \frac{C}{\ell^3}\int_{\partial\O} d\s_y\,  \frac{\k V(y) }{|x-y|^4}  \\
  \end{split} 
\end{equation}
To estimate the first term on the right-hand side, we perform the same change of variables as in $\text{D}_4$. Extending the domain of integration to $\RRR^5$ and doing the integration over $b^\perp$ we have
\begin{equation}\nonumber
\begin{split}
\int_{\partial\O,\, y_1^{(1)}=-\ell/2} d\s_y\,\frac{V(y)}{|x-y|^4} &\leq  C \int_{\RRR^3} dy_2^{(1)}da^{\perp}   \frac{V(y_2^{(1)}+\ell/2,a^{\perp})}{\big|x_1^{\perp}-x_2^{\perp}-a^{\perp}\big|^2+  \big|x_1^{(1)} + x_2^{(1)}+\ell/2 -y_2^{(1)}\big|^2} \\ & \leq  C\frac{1}{|x_1^{(1)}+x_2^{(1)}-\ell|^2+1} 
  \end{split} 
\end{equation}
where we used again \eqref{sito} in the last step. Hence
%
\begin{equation}\label{eq:D3}
\begin{split}
|\big[d(\tfrac{x_1+x_2}2)^{5/3}+1\big]\text{D}_3(x)|&\leq C\k  \ell^{-3}
  \end{split} 
\end{equation}

By combining \eqref{a43}, \eqref{eq:D2}, \eqref{eq:D11}, \eqref{eq:ksmall}, \eqref{eq:D51}, \eqref{eq:D52b}, \eqref{eq:D6}, \eqref{eq:D4} and \eqref{eq:D3} we have thus shown that
\begin{equation}\nonumber
\begin{split}
  \big[d(\tfrac{x_1+x_2}2)^{5/3}+1\big]|\nabla_{x_1+x_2}g(x)| &\leq C\k \left( \ell^{-3} +  \varepsilon \ell^{-4/3} \ln (\ell)   \right) \\ & \quad +C\left(\varepsilon \ell^2 +\l_\ell \ell^2 + \kappa\right) \sup_{y\in \O}\big|\big[d(\tfrac{y_1+y_2}2)^{5/3}+1\big]\nabla_{y_1+y_2}g(y)\big|.
\end{split}
\end{equation}
We choose $\varepsilon=c\ell^{-2}$ with small enough $c$ so that the factor $C( \varepsilon\ell^2 + \l_\ell \ell^2 + \kappa)$ is smaller than one for large $\ell$ and small $\kappa$, concluding  the proof of \eqref{eq:PWder}.
\end{proof}

\begin{proof}[Proof of Proposition \ref{prop:eta}]
From \eqref{eq:der} it follows that
 \begin{equation}\label{eq:derOmega} 
\begin{split} 
 \int_{\L_1\times\L_1} dxdy\,\Big[|\nabla_xw_\ell(x, y)|^2+|\nabla_yw_\ell(x, y)|^2\Big]\leq \frac{C\k}{\ell},
\end{split} 
\end{equation}
estimate \eqref{eq:sup} implies
  \begin{equation}\label{eq:supOmegal} 
\begin{split}
|w_\ell(x, y)|\leq C
\end{split} 
\end{equation}
and from \eqref{eq:norm} it follows that
  \begin{equation}\label{eq:normOmega} 
\begin{split}
 \int_{\L_1\times\L_1} dxdy\,\big|w_\ell(x, y)\big|^2\leq \frac{C\k^2}{\ell^2},
\end{split} 
\end{equation}
while \eqref{eq:L1norm} shows that
  \begin{equation}\label{eq:L1normOmega} 
\begin{split}
 \int_{\L_1\times\L_1} dxdy\,|w_\ell(x,y)|\leq \frac{C\k}{\ell}.
\end{split} 
\end{equation}
By equation \eqref{eq:etaDecomp} and bounds \eqref{eq:normOmega}, \eqref{eq:derOmega} we find
\begin{equation}\nonumber
 \begin{split}
 \int_{\L_1\times\L_1} dxdy\,\big|\mu(x, y)\big|^2&\leq C\k\frac{n^2}{\ell^2}\\ \int_{\L_1\times\L_1} dxdy\,\Big[\big|\nabla_x\mu(x, y)\big|^2+\big|\nabla_y\mu(x, y)\big|^2\Big]&\leq C\k\frac{n^2}{\ell}
\end{split} 
\end{equation}
which imply \eqref{eq:normEta} and \eqref{eq:normDEta}.
By \eqref{eq:sup} we have
\begin{equation}
 |\eta(x,y)|\leq n|w_\ell(x,y)|+|\mu(x,y)|\leq C n
\end{equation}
which proves \eqref{eq:supEta}. Estimate \eqref{eq:supEtaBetter} follows from \eqref{eq:PWdecay}.
Point $ii)$ follows from \eqref{eq:supEtaBetter}.

We consider now point $iii)$. From the definition of $r$, we find
\begin{equation}\label{eq:serier}
\begin{split}
 r(x,y)& =  \sum_{n=1}^\infty \frac{1}{(2n+1)!} \eta^{2n+1} (x,y) \\
& = \sum_{n=1}^\infty \frac{1}{(2n+1)!} \int dz dw \, \eta(x,z)\,\eta^{2n-1}(z,w)\, \eta (w,y);
\end{split} 
\end{equation}
using \eqref{eq:normEta}, which implies $\|\eta\|_2\leq C$, we arrive at
\begin{equation}
\begin{split} 
|r (x,y)| \leq \; &\sum_{n=1}^\infty \frac{1}{(2n+1)!} \left[ \int dw dz |\eta(x,z)|^2 |\eta (w,y)|^2 \right]^{1/2}  \left[ \int dw dz \, |\eta^{2n-1}(z,w)|^2 \right]^{1/2} \\
\leq \; &\sum_{n=1}^\infty \frac{1}{(2n+1)!} \| \eta \|_2^{2n-1} \, \| \eta (x,\cdot) \|_2 \| \eta (\cdot,y) \|_2 \leq C \|\eta\|_2 \| \eta (x,\cdot) \|_2 \| \eta (\cdot,y) \|_2 
\end{split}
\end{equation}
for every $x,y \in \L_1$. The bound for $p$ can be proven analogously. This proves \eqref{eq:supr} and consequently \eqref{eq:supL2} and \eqref{eq:normsp}. 
\end{proof}

\section*{Acknowledgments}
Funding from the European Union’s Horizon 2020 research and innovation programme
under the ERC grant agreement No 694227 is gratefully acknowledged.



\begin{thebibliography}{55}

\bibitem{ABS}
A.~{Adhikari}, C.~{Brennecke}, B.~{Schlein}.  Bose-Einstein condensation beyond the Gross-Pitaevskii regime. \emph{Ann. Henri Poincare.} {\bf 22}  (2021), 1163--1233.


\bibitem{AS}
N.~{Aronszajn}, K.T.~{Smith}.
Theory of Bessel potentials. I. \emph{Annales de l'institut Fourier}, tome 11 (1961), 385--475

\bibitem{BDS} 
N.~{Benedikter}, G.~{de Oliveira}, B.~{Schlein}. {Quantitative derivation of the Gross-Pitaevskii equation}. \emph{Comm.\ Pure Appl.\ Math.} {\bf 68} (2014), 1399--1482.


\bibitem{BasCS} G.~Basti, S.~Cenatiempo, B.~Schlein. A new second order upper bound for the ground state energy of dilute Bose gases.\emph{Forum Math. Sigma} {\bf 9}(2021), e74.
%
\bibitem{BBCS1}
C.\ Boccato, C.\ Brennecke, S.\ Cenatiempo, B.\ Schlein. Complete Bose-Einstein condensation in the Gross-Pitaevskii regime.  \emph{Commun.\ Math.\ Phys.\ } {\bf 359} (2018), 975--1026.


\bibitem{BBCS3}
C. Boccato, C. Brennecke, S. Cenatiempo, B. Schlein. Bogoliubov theory in the Gross-Pitaevskii limit. {\it Acta Mathematica} \textbf{222} (2019), 219–335.

\bibitem{BBCS4}
C. Boccato, C. Brennecke, S. Cenatiempo, B. Schlein. Optimal rate for Bose-Einstein condensation in the Gross-Pitaevskii regime.   \emph{Commun.\ Math.\ Phys.\ } {\bf 376} (2020), 1311--1395.

\bibitem{BCS}
C.~{Boccato}, S.~{Cenatiempo}, B.~{Schlein}. {Quantum many-body fluctuations around nonlinear Schr\"odinger dynamics}. \emph{Ann.\ Henri Poincar\'e} \textbf{18} (2017), 113--191.

%
\bibitem{BCapS}
C. Brennecke, M. Caporaletti, B. Schlein.  Excitation spectrum for Bose Gases beyond the Gross-Pitaevskii regime. Preprint arXiv:2104.13003.

\bibitem{BS}
C. Brennecke, B. Schlein. Gross-Pitaevskii dynamics for Bose-Einstein condensates. {\it Anal. PDE} {\bf 12} (2019), 1513--1596.

\bibitem{BSS1}
C. Brennecke, B. Schlein, S. Schraven. Bose-Einstein condensation with optimal rate for trapped bosons in the Gross-Pitaevskii regime.  {\it Math. Phys. Anal. Geom.} {\bf 25} (2022).

\bibitem{BSS2}
C. Brennecke, B. Schlein, S. Schraven.  Bogoliubov Theory for trapped bosons in the Gross-Pitaevskii regime. \emph{Ann.\ Henri Poincar\'e} \textbf{23} (2022), 1583--1658.

\bibitem{BFSol}
B.~Brietzke, S.~Fournais, J.P.~Solovej. A simple 2nd order lower bound to the energy of dilute Bose gases. {\it Commun. Math. Phys.}, {\bf 376} (2020), 323-351.


%

\bibitem{ESY2}
L.~{Erd{\H{o}}s}, B.~{Schlein}, H.-T.~{Yau}.
\newblock Derivation of the {G}ross-{P}itaevskii equation for the dynamics of  {B}ose-{E}instein condensate,
\newblock \emph{Annals of Math.} \textbf{172} (2010), no. 1, 291--370.
%

%

\bibitem{F}
S.~Fournais. Length scales for BEC in the dilute Bose gas. \emph{Partial Differential Equations, Spectral Theory, and Mathematical Physics}, EMS Series of Congress Reports (2021).

\bibitem{FSol}
S.~Fournais, J.P.~Solovej. The energy of dilute Bose gases. \newblock \emph{Annals of Math.}  \textbf{192} (2020), 893--976.
 

%
\bibitem{GS}
P. Grech, R. Seiringer. The excitation spectrum for weakly interacting bosons in a trap. {\it Comm. Math. Phys.} {\bf 322} (2013), no. 2, 559--591.

\bibitem{H}
C. Hainzl. Another proof of BEC in the GP-limit. \emph{J. Math. Phys.} {\bf 62} (2021), 051901.

\bibitem{HST}
C. Hainzl, B. Schlein, A. Triay. Bogoliubov Theory in the Gross-Pitaevskii limit: a simplified approach. Preprint arXiv:2203.03440


\bibitem{LeeY}
T.D.~Lee, C.N.~Yang. {Many-body problem in quantum mechanics and in quantum statistical mechanics}, \emph{Phys.\ Rev.} {\bf 105} (1957), 1119--1120.

\bibitem{LeeHY}
T.D.~Lee, K.~Huang, C.N.~Yang. {Eigenvalues and eigenfunctions of a Bose system of hard spheres
and its low-temperature properties}, \emph{Phys.\ Rev.} {\bf 106} (1957), 1135--1145.
%

%
\bibitem{LNSS} M.~Lewin, P.~T.~{Nam}, S.~{Serfaty}, J.P. {Solovej}. Bogoliubov spectrum of interacting Bose gases.  \newblock \emph{Comm. Pure Appl. Math.} \textbf{68} (2014), no. 3, 413--471 
%
\bibitem{LS}
E.~H.~Lieb and R.~Seiringer.
\newblock Proof of {B}ose-{E}instein condensation for dilute trapped gases.
\newblock \emph{Phys. Rev. Lett.} \textbf{88} (2002), 170409.

\bibitem{LS2}
E.~H.~Lieb and R.~Seiringer.
\newblock Derivation of the Gross-Pitaevskii equation for rotating Bose gases.
\newblock \emph{Comm. Math. Phys. } \textbf{264}:2 (2006), 505--537.
%
\bibitem{LSSY}
E.~H.~Lieb, R.~Seiringer, J. P. Solovej, J. Yngvason. 
The mathematics of the Bose gas and its condensation. 
Oberwolfach Seminars {\bf 34} (2005), Birkh\"auser Basel.
%

%
%
\bibitem{LY} 
E.~H.~Lieb, J.~Yngvason. Ground state energy of the low density Bose Gas. {\it Phys. Rev. Lett.} {\bf 80} (1998), 2504–2507. 
%
\bibitem{NNRT}
P.~T.~{Nam}, M. Napi{\'o}rkowski,  J.~{Ricaud}, A.~{Triay}.  Optimal rate of condensation for trapped bosons in the Gross-Pitaevskii regime.  arXiv:2001.04364.
%
\bibitem{NRS}
P.~T.~{Nam}, N. ~{Rougerie}, R.~Seiringer. Ground states of large bosonic systems: The Gross-Pitaevskii limit revisited. 
\emph{Analysis and PDE.} {\bf 9} (2016), no. 2, 459--485.

\bibitem{NT}
P.~T.~{Nam}, A.~{Triay}. Bogoliubov excitation spectrum of trapped Bose gases in the Gross-Pitaevskii regime .   Preprint arXiv:2106.11949.


%

\bibitem{Sei1}R.\ Seiringer.   Free energy of a dilute Bose gas: Lower Bound.
{\it Commun.\ Math.\ Phys.} {\bf 279} (2008), 595--636.

\bibitem{Sei}
R.~{Seiringer}. The Excitation Spectrum for Weakly Interacting Bosons. {\it Comm. Math. Phys.} {\bf 
306} (2011), 565--578. 


%
\bibitem{YY}
H.-T. Yau, J. Yin. The second order upper bound for the ground state energy of a Bose gas. {\it J. Stat. Phys.} {\bf 136} (2009), no. 3, 453--503.

\bibitem{Y}
J.\ Yin. Free energies of dilute Bose Gases: upper bound. {\it J.\ Stat.\ Phys.} {\bf 141} (2010), 683--726.
\end{thebibliography}
\end{document}